\newif\ifLIPICS
\newif\ifLIPICSappendix
\newif\ifARXIV

\ARXIVtrue

\documentclass[a4paper,english]{article}
\usepackage[utf8]{inputenc}
\usepackage[stmaryrd,todo=hide]{generic}
\usepackage[margin=3.25cm]{geometry}

\usepackage{logic9}

\usepackage{thmtools,thm-restate}

\usepackage{scrextend}

\newcommand{\gabriele}[2][]{{\color{black}\todo[color=green!30,#1]{{\bf G.} #2}}\ignorespaces}

\definecolor{nicecyan}{HTML}{006165}
\definecolor{nicered}{HTML}{DB3A34}
\newcommand\nicecyan[1]{{\color{nicecyan}#1}}
\newcommand\nicered[1]{{\color{nicered}#1}}

\newcommand\dom{\textsf{dom}}

\newcommand\oeq{=_o}

\newcommand\ocont{\subseteq_o}
\newcommand\nocont{\not\subseteq_o}
\newcommand\synch{G}
\newcommand\delay{\textsf{delay}}

\newcommand\ipar{\textsf{ipar}}
\newcommand\opar{\textsf{opar}}
\newcommand\move{\textsf{move}}
\newcommand\nxt{\textsf{next}}
\newcommand\angled[1]{\langle #1 \rangle}
\newcommand\In{\textsf{in}}
\newcommand\Out{\textsf{out}}
\newcommand\Orig{\textsf{orig}}
\newcommand\otype{\gamma}

\newcommand{\Qright}{Q_{\succ}}
\newcommand{\Qleft}{Q_{\prec}}

\newcommand{\bound}{B}

\renewcommand{\act}[2][]{\mathrel{\raisebox{-1pt}[10pt][0pt]{%
			\ensuremath{\underset{^{\raisebox{-6pt}[0pt][0pt]{\ensuremath{^{^{#1}}}}}}%
				{\raisebox{0pt}[3pt][0pt]{\ensuremath{\relbar\mspace{-8mu}\xrightarrow{#2}}}}}}}}

\newcommand{\fin}{\omega}

\newcommand{\lftmark}{\mathop{\vdash}}
\newcommand{\rgtmark}{\mathop{\dashv}}

\newcommand\true{\mathit{true}}

\newcommand{\om}{\textsf{omatch}}
\newcommand{\im}{\textsf{imatch}}
\newcommand{\match}{\textsf{match}}
\newcommand{\il}{\textsf{lag}}

\newcommand{\eqS}{=_{\Sigma}}
\newcommand{\eqG}{=_{\Gamma}}

\newcommand{\projS}{\pi_{\S}}
\newcommand{\projG}{\pi_{\G}}
\newcommand{\guess}{\textit{before}}

\newcommand{\chck}{\textit{after}}

\usepackage{wrapfig}
\usepackage{hyperref}

\newsavebox{\wraptext}
\newsavebox{\wrapfig}
\newsavebox{\deftext}
\savebox{\deftext}{Xp}
\newcommand\mywraptextjustified[2]{\savebox{\wraptext}{\parbox[t]{\linewidth-#1}%
                                                                 {#2\linebreak}}}
\newcommand\mywraptext[2]{\savebox{\wraptext}{\parbox[t]{\linewidth-#1}%
                                                        {#2\par\phantom{X}}}}
\newcommand\mywrapfig[2]{\savebox{\wrapfig}{\parbox[t]{#1}{\centering #2}}%
                         \usebox{\wraptext}%
                         \raisebox{-\ht\wrapfig}{\usebox{\wrapfig}}%
                         \vspace{-\ht\deftext minus 1pt}\linebreak}

\usetikzlibrary{fit}
\usetikzlibrary{automata,patterns,shadows,shadings,shapes,arrows,calc,backgrounds}
\usetikzlibrary{decorations.markings,decorations.pathmorphing,decorations.pathreplacing}
\tikzstyle{dot} = [draw,shape=circle,fill, minimum size=1mm, inner sep=0pt, outer sep=0pt]
\tikzstyle{ball} = [draw,shape=circle, minimum size=1cm, inner sep=0pt, outer sep=0pt]
\tikzstyle{arrow} = [->, >=stealth, line width=0.5pt, shorten >=1pt, shorten <=1pt, 
preaction={-, draw, shorten >=1pt, shorten <=1pt, line width=1.75pt, white}]
\tikzstyle{link} = [line width=0.5pt, shorten >=1pt, shorten <=1pt, 
preaction={-, draw, shorten >=1pt, shorten <=1pt, line width=1.75pt, white}]
\tikzstyle{hidden}=[minimum size=1mm, shape=circle, inner sep=0pt, outer sep=0pt, style={font=\vphantom{Ag}}]
\tikzstyle{state}=[draw,shape=circle, fill=gray!25, minimum size=6mm, inner sep=0pt, outer sep=0pt, style={font=\vphantom{Ag}}]
\tikzstyle{letter}=[draw,shape=rectangle, fill=gray!25, minimum size=5mm, inner sep=0pt, outer sep=0pt, style={font=\vphantom{Ag}}]
\tikzstyle{cir}=[draw=violet, fill=violet!20!white, rounded rectangle,minimum size=1.5em,inner sep=0em]                                                                                                      
\tikzstyle{cir1}=[draw=green!50!blue, fill=green!40!white, circle,minimum size=1.5em,inner sep=0.1em]                                                                                                      
\tikzstyle{cir2}=[draw=red,fill=red!20!white,rounded rectangle,minimum size=1.5em,inner sep=0em]                                                                                                    
\tikzstyle{cir11}=[draw=green!20!blue,fill=green!40!white,rounded rectangle,minimum size=1.5em,inner sep=0em]                                                                                                    
\tikzstyle{cir22}=[draw=green!70!violet,fill=green!15!violet!40!white,rounded rectangle,minimum size=1.5em,inner sep=0em]

\ifLIPICS
\title{On Synthesis of Resynchronizers for Transducers}
\fi
\ifARXIV
\title{\bf On Synthesis of Resynchronizers for Transducers}
\fi

\ifLIPICS
\author{Sougata Bose}{LaBRI, University of Bordeaux, France}{}{}{}
\author{Shankara Narayanan Krishna}{Department of Computer Science \& Engineering IIT Bombay, India}{}{}{} 
\author{Anca Muscholl}{LaBRI, University of Bordeaux, France}{}{}{}
\author{Vincent Penelle}{LaBRI, University of Bordeaux, France}{}{}{}
\author{Gabriele Puppis}{CNRS, LaBRI, University of Bordeaux, France}{}{}{}
\funding{DeLTA project (ANR-16-CE40-0007)}
\authorrunning{S. Bose, S. Krishna, A. Muscholl, V. Penelle, G. Puppis}
\Copyright{S. Bose, S. Krishna, A. Muscholl, V. Penelle, G. Puppis}
\acknowledgements{}
\ccsdesc[500]{Theory of computation~Transducers}
\keywords{String transducers, resynchronizers, synthesis}
\category{}
\relatedversion{A full version of the paper is available at \url{https://arxiv.org/abs/1906.08688}.}
\nolinenumbers

\EventEditors{Peter Rossmanith, Pinar Heggernes, and Joost-Pieter Katoen}
\EventNoEds{3}
\EventLongTitle{The 44th International Symposium on Mathematical Foundations of Computer Science (MFCS 2019)}
\EventShortTitle{MFCS 2019}
\EventAcronym{MFCS}
\EventYear{2019}
\EventDate{August 26-30, 2019}
\EventLocation{Aachen, Germany}
\EventLogo{}
\SeriesVolume{}
\ArticleNo{}
\fi

\ifARXIV
\author{Sougata Bose \\ LaBRI, University of Bordeaux, France \and
        Shankara Narayanan Krishna \\ Department of Computer Science \& Engineering IIT Bombay, India \and
        Anca Muscholl \\ LaBRI, University of Bordeaux, France \and
        Vincent Penelle \\ LaBRI, University of Bordeaux, France \and
        Gabriele Puppis \\ CNRS, LaBRI, University of Bordeaux, France}
\date{}
\fi

\begin{document}

\sloppy

\maketitle

\makeatletter
\let\mydollar=$  
\makeatother
\def\<{\ifinlineelse{\mydollar}{\begin{equation*}}} 
\def\>{\ifinlineelse{\mydollar}{\end{equation*}}}

\newenvironment{itemize*}{\ifinlineelse{%
                            \let\olditem=\item%
                            \def\item{}%
                            \leavevmode\unskip%
                          }{%
                            \begin{itemize}
                          }}%
                         {\ifinlineelse{%
                            \let\item=\olditem%
                            \leavevmode\unskip%
                          }{%
                            \end{itemize}
                          }}
\newenvironment{enumerated*}{\ifinlineelse{%
                            \let\olditem=\item%
                            \def\item{}%
                            \leavevmode\unskip%
                          }{%
                            \begin{enumerated}
                          }}%
                         {\ifinlineelse{%
                            \let\item=\olditem%
                            \leavevmode\unskip%
                          }{%
                            \end{enumerated}
                          }}




\begin{abstract}
We study two formalisms that allow to compare transducers over words under
origin semantics:  rational and  regular resynchronizers, and show
that the former are captured by the latter.  We then consider some instances of the following synthesis
problem: given transducers $T_1,T_2$, 
construct a rational (resp.~regular) resynchronizer $R$, if it exists,
such that  $T_1$ is  contained in $R(T_2)$ under the origin semantics. We show that synthesis of rational resynchronizers 
is decidable for functional, and even finite-valued, one-way transducers,
and undecidable for relational one-way transducers. In the two-way setting,
synthesis of regular resynchronizers is shown to be decidable for 
unambiguous two-way transducers. For larger classes of two-way 
transducers, the decidability status is open.


\end{abstract}

\section{Introduction}\label{sec:intro}

The notion of word transformation is pervasive in computer science,
as computers typically process streams of data and transform them between 
different formats. 
The most basic form of word transformation is realized 
using finite memory. Such a model is 
called finite-state transducer and was studied from the early 
beginnings of automata theory.
Differently from automata, the expressiveness of  transducers
is significantly affected by the presence of non-determinism 
(even when the associated transformation is a function),
and by the capability of processing the input in both directions 
(one-way vs two-way transducers).
Another difference is that many problems, notably, equivalence
and containment, become undecidable when moving from automata to 
transducers \cite{FischerR68,iba78siam}.

An alternative semantics for transducers, called \emph{origin
  semantics}, was introduced in~\cite{boj14icalp} in order to obtain
canonical two-way word transducers.  In the origin semantics, the
output is tagged with positions of the input, called origins, that
describe where each output element was produced.  According to this
semantics, two transducers may be non-equivalent even when they
compute the same relation in the classical semantics. 
From a
computational viewpoint, the
origin semantics has the advantage that it allows to recover the
decidability of equivalence and containment of
non-deterministic (and even two-way) transducers~\cite{bmpp18}.

It can be argued that comparing two transducers in the origin
semantics is rather restrictive, because it requires that the same
output is generated at precisely the same place. A natural approach to
allow some 'distortion' of the origin information when comparing two
transducers was proposed in~\cite{FiliotJLW16}. \emph{Rational}
resynchronizers allow to
compare \emph{one-way} transducers (hence, the name 'rational') under
origin distortions that are generated with finite control. A rational resynchronizer is simply
a one-way transducer that processes an interleaved input-output
string, producing another interleaved interleaved input-output string
with the same input and output projection. For two-way transducers (or
equivalently, streaming
string transducers~\cite{AC10}) a different formalism is
required to capture origin distortion, since the representation of the origin
information through interleaved input-output pairs does not work
anymore. To this purpose, \emph{regular resynchronizers} were introduced in~\cite{bmpp18}
as a logic-based transformation of origin graphs, in the spirit of 
 Courcelle's monadic second-order logic definable graph
transductions~\cite{CE12}. In~\cite{bmpp18} it  was shown that containment of two-way
transducers up to a (bounded) regular resynchronizer is decidable.


In this paper we first show that bounded regular resynchronizers capture the
rational ones. This result is rather technical, because rational
resynchronizers work on explicit origin graphs, encoded as
input-output pairs, which is not the case for regular
resynchronizers. Then we consider the following problem: given two
transducers $T_1,T_2$, we ask whether some rational, or bounded regular,
resynchronizer  $R$ exists such that $T_1$ is origin-contained in $T_2$
up to $R$.  So here, the resynchronizer $R$ is not part of the input,
and we want to synthesize such a resynchronizer, if one exists.

Our main contributions can be summarized as follows:
\begin{enumerate}
	\item synthesis of rational resynchronizers 
	      for functional (or even finite-valued) one-way transducers is decidable,
	\item synthesis of rational resynchronizers 
	      for unrestricted one-way transducers is undecidable,
	\item synthesis of bounded regular resynchronizers 
		  for unambiguous two-way transducers is decidable.
\end{enumerate}
Somewhat surprisingly, for both decidable cases above the existence of a
resynchronizer turns out to be equivalent to the classical inclusion
of the two transducers.



\section{Preliminaries}\label{sec:prelims}

\paragraph*{One-way transducers.}
One of the simplest transducer model is the one-way 
non-deterministic finite-state transducer (hereafter, 
simply one-way transducer), capturing the class of 
so-called \emph{rational relations}. 
This is basically an automaton in which every transition 
consumes one letter from the input and appends a word 
of any length to the output. 

Formally, a \emph{one-way transducer} is a tuple $T=(\S,\G,Q,I,E,F,L)$, where 
$\S,\G$ are finite input and output alphabets, $Q$ is a finite set of states,
$I,F\subseteq Q$ are subsets of initial and final states, $E\subseteq Q\times\Sigma\times Q$
is a finite set of transition rules, and $L:E\uplus F\rightarrow 2^{\Gamma^*}$ is a 
function 
specifying a regular 
language of partial outputs for each 
transition rule and each final state. 
The relation defined by $T$ contains 
pairs $(u,v)$ of input and output words,
where $u=a_1\dots a_n$ and $v=v_1\dots v_n\,v_{n+1}$,
for which there is a run
$q_0 \act{a_1 \:|\: v_1} q_1 \act{a_2 \:|\: v_2} \dots q_n \act{\:|\: v_{n+1}}$
such that
$q_0\in I$, $q_n\in F$, $(q_{i-1},a_i,q_i)\in E$, 
$v_i\in L(q_{i-1},a_i,q_i)$, and $v_{n+1}\in L(q_n)$.
The transducer is called \emph{functional} if it associates at most one output
with each input, namely, if it realizes a partial function. 
For example, the figure below shows two one-way transducers 
\ifARXIV
with input alphabet $\S=\{a,b\}$ and output alphabet $\G\supseteq\S$.
\fi
\mywraptext{6.3cm}{%
The first transducer is functional, and realizes the 
cyclic rotation $f: cu \mapsto uc$, for any letter $c\in\{a,b\}$ and any word
$u\in\{a,b\}^*$.
The second transducer is not functional, and associates with an input 
$u\in\S^*$ any possible word $v\in\G^*$ as output such that $u$ is a 
sub-sequence of $v$.
}%
\mywrapfig{6.3cm}
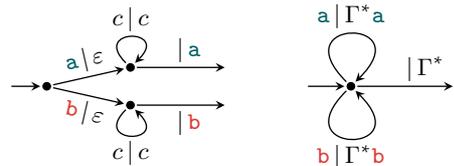
{%
	\vspace{-2mm}%
\begin{tikzpicture}
%
\begin{scope}
	\draw (0,0) node [dot] (A) {};
	\draw (1.1,0.25) node [dot] (B) {};
	\draw (1.1,-0.25) node [dot] (C) {};
	\draw (2.5,0.25) node (D) {};
	\draw (2.5,-0.25) node (E) {};

	\draw [arrow] (-0.5,0) to (A);
	
	\draw [arrow] (A) to node [black, above=-0.05, sloped] 
	      {\small $\texttt{\nicecyan{a}} \:|\: \emptystr$} (B);
	\draw [arrow] (A) to node [black, below=-0.05, sloped] 
	      {\small $\texttt{\nicered{b}} \:|\: \emptystr$} (C);
	
	\draw [arrow] (B) to [out=135, in=45, looseness=25] 
	      node [black, above=-0.05] {\small $c \:|\: c$} (B);
	\draw [arrow] (C) to [out=-135, in=-45, looseness=25] 
	      node [black, below=-0.05] {\small $c \:|\: c$} (C);
	
	\draw [arrow] (B) to node [black, above=-0.05, sloped] 
	      {\small $~~|\:\texttt{\nicecyan{a}}$} (D);
	\draw [arrow] (C) to node [black, below=-0.05, sloped] 
	      {\small $~~|\:\texttt{\nicered{b}}$} (E);
\end{scope}

\begin{scope}[xshift=4cm]
	\draw (0,0) node [dot] (A) {};
	\draw (1.5,0) node (B) {};

	\draw [arrow] (-0.6,0) to (A);
	
	\draw [arrow] (A) to [out=135, in=45, looseness=45] node [black, above=-0.05] 
	      {\small $\texttt{\nicecyan{a}} \:|\: \G^*\texttt{\nicecyan{a}}$} (A);
	\draw [arrow] (A) to [out=-135, in=-45, looseness=45] node [black, below=-0.05] 
	      {\small $\texttt{\nicered{b}} \:|\: \G^*\texttt{\nicered{b}}$} (A);
	
	\draw [arrow] (A) to node [black, above=-0.05] 
	      {\small $~~\quad|\: \G^*$} (B);
\end{scope}
\end{tikzpicture}
\hspace{-8mm}
%
}

\paragraph*{Two-way transducers.}
Allowing the input head to move in any direction, 
to the left or to the right, gives a more 
powerful model of transducer, which captures 
e.g.~the relation $\{(u,u^n) ~:~ u\in\S^*,n\in\Nat\}$.
To define two-way transducers, we adopt 
the convention that, for any given input $u\in\S^*$, 
$u(0)=\lftmark$ and $u(|u|+1)=\rgtmark$, where 
$\lftmark,\rgtmark\notin\S$ are special markers 
used as delimiters of the input.
In this way, a transducer can detect when an endpoint of the 
input has been reached.

A \emph{two-way transducer} is a tuple $T=(\S,\G,Q,I,E,F,L)$, 
whose components are defined just like those of a one-way transducer, 
except that 
the state set $Q$ is partitioned into two subsets,
$\Qleft$ and $\Qright$, the set $I$ of initial states 
is contained in $\Qright$, and the set $E$ of transition rules is contained in
$(Q\times\Sigma\times Q) \:\uplus\: 
 (\Qleft\times\{\lftmark\}\times\Qright) \:\uplus\:
 (\Qright\times\{\rgtmark\}\times\Qleft)$.
The partitioning of the set of states is useful for specifying which
letter is read from each state: states from $\Qleft$ read the letter 
to the left, whereas states from $\Qright$ read the letter to the right.
Given an input $u\in\S^*$, a configuration of a two-way transducer is a 
pair $(q,i)$, with $q\in Q$ and $i\in\{1,\dots,|u|+1\}$. Based on the types 
of source and target states in a transition rule, we can distinguish four types
of transitions between configurations (the output $v$ is always assumed
to range over the language $L(q,a,q')$):
\begin{itemize}
	\item $(q,i) \act{a \:|\: v} (q',i+1)$ if $(q,a,q')\in E$, $q,q'\in\Qright$, and $a=u(i)$,
	\item $(q,i) \act{a \:|\: v} (q',i)$ if $(q,a,q')\in E$, $q\in\Qright$, $q'\in\Qleft$, and $a=u(i)$, 
	\item $(q,i) \act{a \:|\: v} (q',i-1)$ if $(q,a,q')\in E$, $q,q'\in\Qleft$, and $a=u(i-1)$, 
    \item $(q,i) \act{a \:|\: v} (q',i)$ if $(q,a,q')\in E$, $q\in\Qleft$, $q'\in\Qright$, and $a=u(i-1)$.
\end{itemize}
Note that, when reading a marker $\lftmark$ or $\rgtmark$,
the transducer is obliged to make a U-turn, either left-to-right or right-to-left.
The notions of successful run, realized relation, and functional 
transducer are naturally generalized from the one-way to the two-way variant, 
(we refer to \cite{bmpp18} for more details). 

In \cite{BDGPicalp17}, a slight extension of two-way transducers,
called two-way transducers \emph{with common guess}, was proposed. 
Before processing its input, such a transducer can non-deterministically
guess some arbitrary annotation of the input over a fixed alphabet. 
Once an annotation is guessed, it remains the same during the computation. 
Transitions may then depend on the input letter and the guessed annotation 
at the current position. 
For example, this extension allows to define relations of
the form $\{ (u,vv) ~\mid~ u\in\S^*, v\in\G^*, |u|=|v|\}$.
Note that the extension with common guess does not increase the expressiveness of 
one-way transducers, since these are naturally closed under input projections. 
Likewise, common guess does not affect the expressive power of functional two-way 
transducers, since one can guess a canonical annotation at runtime.

\paragraph*{Classical vs origin semantics.}
In the previous definitions, we associated a classical semantics
to transducers (one-way or two-way), which gives rise to
relations or functions between input words over $\S$ and output 
words over $\G$.
In \cite{boj14icalp} an alternative semantics for transducers,
called \emph{origin semantics}, was introduced with the goal of
getting canonical transducers for any given word function.
Roughly speaking, in the origin semantics, every position of the output
word is annotated with the position of the input where that particular
output element was produced. 
This yields a bipartite graph, called \emph{origin graph}, with two linearly
ordered sets of nodes, representing respectively the input and the output elements,
and edges directed from output nodes to input nodes, representing the
so-called \emph{origins}.
%
\mywraptext{6cm}{%
The figure depicts an input-output pair $(a^n,b^n)$
annotated with two different origins: in the first graph, a position $i$ in the
output has its origin at the same position $i$ in the input, while in the second
graph it has origin at position $n-i$.
}%
\mywrapfig{6cm}
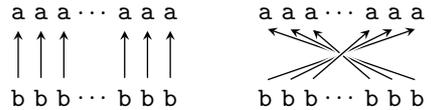
{%
\begin{tikzpicture}[yscale=0.9]
\begin{scope}
	
	\draw (0,1.25) node (1) {$\texttt{a}$};
	\draw (0.3,1.25) node (2) {$\texttt{a}$};
	\draw (0.6,1.25) node (3) {$\texttt{a}$};
	\draw (1,1.25) node (4) {$\dots$};
	\draw (1.4,1.25) node (5) {$\texttt{a}$};
	\draw (1.7,1.25) node (6) {$\texttt{a}$};
	\draw (2,1.25) node (7) {$\texttt{a}$};
	
	\draw (0,0) node (1') {$\texttt{b}$};
	\draw (0.3,0) node (2') {$\texttt{b}$};
	\draw (0.6,0) node (3') {$\texttt{b}$};
	\draw (1,0) node (4') {$\dots$};
	\draw (1.4,0) node (5') {$\texttt{b}$};
	\draw (1.7,0) node (6') {$\texttt{b}$};
	\draw (2,0) node (7') {$\texttt{b}$};
	
	\draw [arrow] (1'.north) to (1.south);
	\draw [arrow] (2'.north) to (2.south);
	\draw [arrow] (3'.north) to (3.south);
	\draw [arrow] (5'.north) to (5.south);
	\draw [arrow] (6'.north) to (6.south);
	\draw [arrow] (7'.north) to (7.south);
\end{scope}
	
\begin{scope}[xshift=3.25cm]
	\draw (0,1.25) node (1) {$\texttt{a}$};
	\draw (0.3,1.25) node (2) {$\texttt{a}$};
	\draw (0.6,1.25) node (3) {$\texttt{a}$};
	\draw (1,1.25) node (4) {$\dots$};
	\draw (1.4,1.25) node (5) {$\texttt{a}$};
	\draw (1.7,1.25) node (6) {$\texttt{a}$};
	\draw (2,1.25) node (7) {$\texttt{a}$};
	
	\draw (0,0) node (1') {$\texttt{b}$};
	\draw (0.3,0) node (2') {$\texttt{b}$};
	\draw (0.6,0) node (3') {$\texttt{b}$};
	\draw (1,0) node (4') {$\dots$};
	\draw (1.4,0) node (5') {$\texttt{b}$};
	\draw (1.7,0) node (6') {$\texttt{b}$};
	\draw (2,0) node (7') {$\texttt{b}$};
	
	\draw [arrow] (7'.north) to (1.south);
	\draw [arrow] (1'.north) to (7.south);
	\draw [arrow] (6'.north) to (2.south);
	\draw [arrow] (2'.north) to (6.south);
	\draw [arrow] (5'.north) to (3.south);
	\draw [arrow] (3'.north) to (5.south);
\end{scope}
\end{tikzpicture}
\hspace{-6mm}	

%
}

Formally, the origin semantics of a transducer is a
relation $S_o \subseteq \S^* \times(\G \times \Nat)^*$ consisting of
pairs $(u,\nu)$, where $u=a_1\dots a_n\in\S^*$ is a possible input and
$\nu=\nu_1\dots\nu_{m+1}\in(\G\times\Nat)^*$ is the corresponding
output tagged with input positions, as induced by a successful run of
the form
$(q_0,i_0)\act{a_1 \:|\: \nu_1} (q_1,i_1) \act{a_2 \:|\: \nu_2} \dots
(q_m,i_m) \act{\:|\: \nu_{m+1}}$, with each
$\nu_j\in (\G\times\{i_j\})^*$. 
We identify a pair $(u,\nu)$ with the origin graph obtained by arranging
the input elements and the output elements along two lines (we omit 
the successor relation in the graph notation), and adding edges from 
every output element $(a,i)$ to the $i$-th element of the input.
Given an origin graph $\synch=(u,\nu)$, we denote by
$\In(\synch)$, $\Out(\synch)$, and $\Orig(\synch)$ respectively the
input word $u$, the output word obtained by projecting $\nu$ onto the
finite alphabet $\G$, and the sequence of input positions (origins)
obtained by projecting $\nu$ onto $\Nat$.

For one-way transducers, there is a simpler
 presentation of origin graphs in the form of interleaved words. Assuming
that the alphabets $\S$ and $\G$ are disjoint, we interleave the input
and output word by appending after each input symbol the output word
produced by reading that symbol.
For example, if $\S=\{a\}$ and $\G=\{b\}$, then a word 
of the form $abb\dots abb$ represents an origin graph $(a^n,\nu)$, 
where $|\nu|=2n$ and $\nu(2i-1)=\nu(2i)=(b,i)$, for all $i=1,\dots,n$.
Words over $\S\uplus\G$ are called \emph{synchronized words}.
Just as every synchronized word represents an origin graph,  
a regular language over $\S\uplus\G$ represents a rational relation
with origins, or equally the origin semantics of a one-way transducer.

In general, when comparing transducers, we can refer to one of the two 
possible semantics. Clearly, two transducers that 
are equivalent in the origin semantics are also equivalent in the 
classical semantics, but the converse  is not true.



\section{Resynchronizations}\label{sec:expressiveness}

The central concept of this paper is that of resynchronization,
which is a transformation of origin graphs 
that preserves the underlying input and output words. 
The concept was originally introduced in \cite{FiliotJLW16},
and mostly studied in the setting of rational relations.
Here we use the concept in the more general setting of relations
definable by two-way transducers.

Formally, a \emph{resynchronization} is any relation
$R \subseteq (\S^*\times(\G\times\Nat)^*)^2$ 
that contains only pairs $(\synch,\synch')$ of origin graphs
such that $\In(\synch) = \In(\synch')$
and $\Out(\synch) = \Out(\synch')$, namely, with the same projections onto 
the input and output alphabets.%
\footnote{In \cite{FiliotJLW16}, resynchronizers were further restricted 
	to contain at least the pairs of identical origin graphs.
	Here we prefer to avoid this additional restriction and reason with 
	a more general class of resynchronizations.}
A resynchronization $R$ can be used to modify the origin information of a relation,
while preserving the underlying input-output pairs. 
Formally, for every relation $S_o \subseteq \S^*\times (\G\times\Nat)^*$ 
with origins, we define the \emph{resynchronized relation} 
$
  R(S_o) = \{ \synch' \in S_o \:\mid\: (\synch,\synch')\in R, ~ \synch\in S_o\}.
$
Note that if the origin information is removed from both $R(S_o)$ and $S_o$, 
then $R(S_o) \subseteq S_o$. 
Moreover, $R(S_o)=S_o$ when $R$ is the \emph{universal resynchronization}, 
that is, when $R$ contains all pairs $(\synch,\synch')$, 
with $\synch,\synch'\in\S^*\times(\G\times\Nat)^*$,
$\In(\synch) = \In(\synch')$, and $\Out(\synch) = \Out(\synch')$.

\paragraph*{Definability of resynchronized relations.} 
An important property that we need to guarantee in order to enable
some effective reasoning on resynchronizations is the definability
of the resynchronized relations. More precisely, given a class $\Cc$ 
of transducers, we say that a resynchronization $R$ \emph{preserves
	definability in $\Cc$} if for every transducer $T\in\Cc$, the relation 
$R(T)$ is realized by some transducer $T'\in\Cc$, that can be effectively 
constructed from $R$ and $T$. The class $\Cc$ will usually be the class
of one-way transducers or the class of two-way transducers, and this
will be clear from the context.

Below, we recall the definitions of two important classes of resynchronizations,
called rational \cite{FiliotJLW16} and regular resynchronizers \cite{bmpp18}, 
that preserve definability by one-way transducers and by two-way transducers, respectively.
We will then compare the expressive power of these two formalisms, showing that 
rational resynchronizers are strictly less expressive than 
regular resynchronizers.

\paragraph*{Rational resynchronizers.}	
A natural definition of resynchronizers for one-way transducers
is obtained from rational relations over the disjoint union 
$\Sigma\uplus\Gamma$ of the input and output alphabets. 
Any such relation consists of pairs of synchronized words $(w,w')$,
and thus represents a transformation of origin graphs. 
In addition, if the induced synchronized words $w$ and $w'$ 
have the same projections over the input and output alphabets, 
then the relation represents a resynchronization. 
We also recall that rational relations are captured by one-way transducers,
so, by analogy, we call \emph{rational resynchronizer} any one-way transducer 
over $\Sigma\uplus\Gamma$ that preserves the input and output projections. 

It is routine to see that rational resynchronizers preserve definability 
of relations by one-way transducers.
It is also worth noting that every rational resynchronizer is a length-preserving 
transducer. By a classical result of Elgot and Mezei~\cite{ElgotM65}
every rational resynchronizer can be assumed to be a \emph{letter-to-letter} 
one-way transducer, 
namely, a transducer with transitions 
of the form $q\act{a \:|\: b} q'$, with $a,b\in\S\uplus\G$.

\begin{figure}
\centering  
\begin{tikzpicture}
	
\begin{scope}[color=nicecyan]
    \draw (0,3.5) node {$T_1$};
    
	\draw (0.25,2.75) node [dot] (A) {};
	\draw (1.75,2.75) node [dot] (B) {};
	
	\draw [arrow] (-0.25,3) to (A);
	\draw [arrow] (A) to (-0.25,2.5);
	\draw [arrow] (A) to [bend left] node [above] {\small $\texttt{a} \:|\: \texttt{b}$} (B);
	\draw [arrow] (B) to [bend left] node [below] {\small $\texttt{a} \:|\: \texttt{b}$} (A);
\end{scope}
\begin{scope}[yshift=0.5cm,color=nicecyan]
	\draw (0,0.9) node (1) {$\texttt{a}$};
	\draw (0.3,0.9) node (2) {$\texttt{a}$};
	\draw (0.6,0.9) node (3) {$\texttt{a}$};
	\draw (0.9,0.9) node (4) {$\texttt{a}$};
	\draw (1.3,0.9) node (5) {$\dots$};
	\draw (1.7,0.9) node (6) {$\texttt{a}$};
	\draw (2,0.9) node (7) {$\texttt{a}$};
	
	\draw (0,0) node (1') {$\texttt{b}$};
	\draw (0.3,0) node (2') {$\texttt{b}$};
	\draw (0.6,0) node (3') {$\texttt{b}$};
	\draw (0.9,0) node (4') {$\texttt{b}$};
	\draw (1.3,0) node (5') {$\dots$};
	\draw (1.7,0) node (6') {$\texttt{b}$};
	\draw (2,0) node (7') {$\texttt{b}$};
	
	\draw [arrow] (1'.north) to (1.south);
	\draw [arrow] (2'.north) to (2.south);
	\draw [arrow] (3'.north) to (3.south);
	\draw [arrow] (4'.north) to (4.south);
	\draw [arrow] (6'.north) to (6.south);
	\draw [arrow] (7'.north) to (7.south);
\end{scope}
	
\begin{scope}[xshift=5cm,color=nicered]
    \draw (0,3.5) node {$T_2$};
    
	\draw (0.25,2.75) node [dot] (A) {};
	\draw (1.75,2.75) node [dot] (B) {};
	
	\draw [arrow] (-0.25,3) to (A);
	\draw [arrow] (A) to (-0.25,2.5);
	\draw [arrow] (A) to [bend left] node [above] {\small $\texttt{a} \:|\: \texttt{bb}$} (B);
	\draw [arrow] (B) to [bend left] node [below] {\small $\texttt{a} \:|\: \varepsilon$} (A);
\end{scope}
\begin{scope}[xshift=5cm,yshift=0.5cm,color=nicered]
	\draw (0,0.9) node (1) {$\texttt{a}$};
	\draw (0.3,0.9) node (2) {$\texttt{a}$};
	\draw (0.6,0.9) node (3) {$\texttt{a}$};
	\draw (0.9,0.9) node (4) {$\texttt{a}$};
	\draw (1.3,0.9) node (5) {$\dots$};
	\draw (1.7,0.9) node (6) {$\texttt{a}$};
	\draw (2,0.9) node (7) {$\texttt{a}$};

	\draw (0,0) node (1') {$\texttt{b}$};
	\draw (0.3,0) node (2') {$\texttt{b}$};
	\draw (0.6,0) node (3') {$\texttt{b}$};
	\draw (0.9,0) node (4') {$\texttt{b}$};
	\draw (1.3,0) node (5') {$\dots$};
	\draw (1.7,0) node (6') {$\texttt{b}$};
	\draw (2,0) node (7') {$\texttt{b}$};

	\draw [arrow] (1'.north) to ([xshift=-1pt] 1.south);
	\draw [arrow] (2'.north) to ([xshift=1pt] 1.south);
	\draw [arrow] (3'.north) to ([xshift=-1pt] 3.south);
	\draw [arrow] (4'.north) to ([xshift=1pt] 3.south);
	\draw [arrow] (6'.north) to ([xshift=-1pt] 6.south);
	\draw [arrow] (7'.north) to ([xshift=1pt] 6.south);
\end{scope}

\begin{scope}[xshift=10cm]
	\draw (0,3.5) node {$R$};

	\draw (0.5,2.75) node [dot] (A) {};
	\draw (1.5,3.25) node [dot] (B) {};
	\draw (2.5,2.75) node [dot] (C) {};
	\draw (1.5,2.25) node [dot] (D) {};

	\draw [arrow] (0,3) to (A);
	\draw [arrow] (A) to (0,2.5);
	\draw [arrow] (A) to node [above, sloped] 
	      {\small ${\color{nicecyan} \texttt{a}} \:|\: {\color{nicered} \texttt{a}}$} (B);
	\draw [arrow] (B) to node [above, sloped] 
	      {\small ${\color{nicecyan} \texttt{b}} \:|\: {\color{nicered} \texttt{b}}$} (C);
	\draw [arrow] (C) to node [below, sloped] 
	      {\small ${\color{nicecyan} \texttt{a}} \:|\: {\color{nicered} \texttt{b}}$} (D);
	\draw [arrow] (D) to node [below, sloped] 
	      {\small ${\color{nicecyan} \texttt{b}} \:|\: {\color{nicered} \texttt{a}}$} (A);
\end{scope}
\begin{scope}[xshift=10cm,yshift=0.5cm]
	\draw (1.25,0.9) node [nicecyan] {\texttt{abababab}~\dots~\texttt{abab}};
	\draw (1.25,0.4) node [rotate=-90] {$\mapsto$};
	\draw (1.25,0) node [nicered] {\texttt{abbaabba}~\dots~\texttt{abba}};
\end{scope}
\end{tikzpicture}
\caption{Two functional 1NFT $T_1,T_2$, their origin graphs, and a rational resynchronizer $R$.}
\label{fix:resynch}
\end{figure}
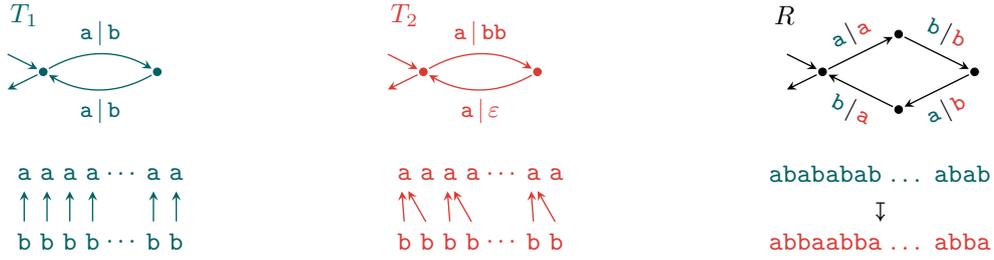

	

\begin{example}\label{ex:rational-resync} 
Consider the functional one-way transducers $T_1,T_2$ in Figure \ref{fix:resynch}. 
The domain of both transducers is $(aa)^*$. 
An origin graph of $T_1$ is a one-to-one mapping 
from the output to the input (each $a$ produces one $b$).
On the other hand, in an origin graph of $T_2$, every $a$ at 
input position $2i+1$ is the origin of two $b$'s at output positions 
$2i+1,2i+2$.   
The transducer $R$ depicted to the right of the figure transforms
synchronized words while preserving their input and output projections.
It is then a rational resynchronizer.
In particular, $R$ transforms origin graphs of $T_1$ to origin graphs of $T_2$. 
\end{example}

\paragraph*{Regular resynchronizers.}
While languages of synchronized words are a faithful representation of
rational relations, this notation does not capture regular relations,
so relations realized by two-way transducers.  An alternative
formalism for resynchronizations of relations defined by two-way
transducers was proposed in \cite{bmpp18} under the name of MSO
resynchronizer (here we call it simply `resynchronizer'). 
The formalism describes pairs $(\synch,\synch')$ of origin graphs by
means of two relations $\move_\otype$ and $\nxt_{\otype,\otype'}$ ($\otype,\otype'\in\G$) in the spirit of MSO graph
transductions. More precisely:
\begin{itemize}
	\item $\move_\otype$ describes how the origin $y$ of an output 
	position $x$ labeled by $\otype$ is redirected to a new origin $z$ 
	(for short, we call $y$ and $z$ the \emph{source} and \emph{target} origins of $x$).
	Formally, $\move_\otype$ is a relation contained in $\Sigma^*\times\Nat\times\Nat$ that
	induces resynchronization pairs $(\synch,\synch')$ such that,
	for all output positions $x$, 
	if $\Out(\synch)(x)=\otype$, $\Orig(\synch)(x)=y$, and $\Orig(\synch')(x)=z$, 
	then $(\In(\synch),y,z) \in \move_\otype$.
	\item $\nxt_{\otype,\otype'}$ constrains the target origins $z$ and $z'$ 
	of any two consecutive output positions $x$ and $x+1$ that are labelled 
	by $\otype$ and $\otype'$, respectively.
	Formally, $\nxt_{\otype,\otype'}$ is a relation contained in $\Sigma^*\times\Nat\times\Nat$ 
	that induces resynchronization pairs $(\synch,\synch')$ such that,
	for all output positions $x$ and $x+1$, 
	if $\Out(\synch)(x)=\otype$, $\Out(\synch)(x+1)=\otype'$, 
	$\Orig(\synch')(x)=z$, and $\Orig(\synch')(x+1)=z'$, 
	then $(\In(\synch),z,z') \in \nxt_{\otype,\otype'}$.
\end{itemize}
A \emph{resynchronizer} is  a tuple 
$\big((\move_\otype)_{\otype\in\G},(\nxt_{\otype,\otype'})_{\otype,\otype'\in\G}\big)$,
and defines the resynchronization $R$ with pairs $(\synch,\synch')$ 
induced by 
the relations $\move_\otype$ and $\nxt_{\otype,\otype'}$, where $\otype,\otype'\in\G$.

In order to obtain a well-behaved class of resynchronizations, that in particular 
preserves definability by two-way transducers, we need to enforce some restrictions.
First, we require that the relations $\move_\otype$ and $\nxt_{\otype,\otype'}$ are 
described by regular languages (or equally, definable in monadic second-order logic). 
By this we mean that we encode the input positions $y,z,z'$ with suitable annotations over 
the binary alphabet $\bbB=\{0,1\}$, so that we can identify the
relations $\move_\otype$ and $\nxt_{\otype,\otype'}$ with some
\emph{regular} languages over the expanded alphabet $\S\times\bbB^2$. 
We call \emph{regular resynchronizer} a resynchronizer 
where the relations $\move_\otype$ and $\nxt_{\otype,\otype'}$ are 
given by regular languages.
In addition, we also require that regular resynchronizers are \emph{$k$-bounded}, 
for some $k\in\bbN$, in the sense that for every input $u$, every output letter 
$\otype$, and every target origin $z$, there are at most $k$ positions $y$ 
such that $(u,y,z)\in\move_\otype$.

\begin{example}\label{ex:regular-resync} 
Consider the resynchronization $R$ that contains the pairs 
$(\synch,\synch')$,	where the origin graph $\synch$ (resp.~$\synch'$) 
maps every output position to the first (resp.~last) input position, 
as 
\mywraptextjustified{5cm}{%
		shown in the figure.
		Note that $R$ is `one-way', in the sense that it contains only origin
		graphs that are admissible outcomes of runs of one-way transducers. 
		However, $R$ is not 
		definable by any rational resynchronizer, since, in terms of 
		synchronized words, it should map $a\,v\,u$ to $a\,u\,v$, 
}%
\mywrapfig{5cm}{%
		\hspace{4mm}
\begin{tikzpicture}[scale=0.75,baseline]
\draw (0.5,1.7) node [hidden] (I1) {\small $a_1$};
\draw (1.5,1.7) node [hidden] (I2) {\small $a_2$};
\draw (2.5,1.7) node [hidden] (I3) {\small $a_3$};
\draw (3.5,1.7) node [hidden] (I4) {\small $a_4$};
\draw (4.5,1.7) node [hidden] (I5) {\small $a_5$};

\draw (0,0) node [hidden] (O0) {\small $b_1$};
\draw (1,0) node [hidden] (O1) {\small $b_2$};
\draw (2,0) node [hidden] (O2) {\small $b_3$};
\draw (3,0) node [hidden] (O3) {\small $b_4$};
\draw (4,0) node [hidden] (O4) {\small $b_5$};
\draw (5,0) node [hidden] (O5) {\small $b_5$};

\draw (O0) edge [arrow, nicered, dashed] ([xshift=-5mm] I5.south);
\draw (O1) edge [arrow, nicered, dashed] ([xshift=-3mm] I5.south);
\draw (O2) edge [arrow, nicered, dashed] ([xshift=-1mm] I5.south);
\draw (O3) edge [arrow, nicered, dashed] ([xshift=-0.25mm] I5.south);
\draw (O4) edge [arrow, nicered, dashed] ([xshift=2mm] I5.south);
\draw (O5) edge [arrow, nicered, dashed] ([xshift=4mm] I5.south);

\draw (O0) edge [arrow, nicecyan] ([xshift=-2mm] I1.south);
\draw (O1) edge [arrow, nicecyan] ([xshift=-1mm] I1.south);
\draw (O2) edge [arrow, nicecyan] ([xshift=0.25mm] I1.south);
\draw (O3) edge [arrow, nicecyan] ([xshift=2mm] I1.south);
\draw (O4) edge [arrow, nicecyan] ([xshift=4mm] I1.south);
\draw (O5) edge [arrow, nicecyan] ([xshift=6mm] I1.south);
\end{tikzpicture}


%
}
for every $a\in\S$, $u\in\S^*$, and $v\in\G^*$, 
which is clearly not a rational relation.
The resynchronization $R$ can however be defined by a $1$-bounded regular resynchronizer,
for example $\big((\move_\otype)_{\otype\in\G},(\nxt_{\otype,\otype'})_{\otype,\otype'\in\G}\big)$, 
where $\move_\otype = \{ (u,y,z) \:\mid\: u\in\S^*, \: y=1, \: z=|u| \}$ and 
$\nxt_{\otype,\otype'}=\S^* \times \Nat \times \Nat$.
\end{example}

One can observe that, in the previous example, $\nxt$ is not restricting
the resynchronization further. For other examples that use $\nxt$ in a 
non-trivial way see for instance \cite[Example~13]{bmpp18}.

The notion of resynchronizer can be slightly enhanced in order 
to allow some additional amount of non-determinism in the way origin graphs 
are transformed (this enhanced notion is indeed the one proposed in \cite{bmpp18}).
The principle is very similar to the idea of enhancing two-way transducers 
with common guess. 
More precisely, we allow additional monadic parameters that annotate 
the input and the output, thus obtaining words over expanded alphabets of the 
form $\S\times\S'$ and $\G\times\G'$.
A resynchronizer \emph{with parameters} is thus a tuple 
$\big(\ipar,\opar,(\move_\otype)_\otype,(\nxt_{\otype,\otype'})_{\otype,\otype'}\big)$,
where 
$\ipar \subseteq (\S\times\S')^*$ 
describes the possible annotations of the input,
$\opar \subseteq (\G\times\G')^*$ describes the possible annotations of the output,
and, for every $\otype,\otype'\in\G\times\G'$,
$\move_\otype \subseteq (\S\times\S'\times\bbB^2)^*$ describes a transformation from 
source to target origins of $\otype$-labelled output positions, and 
$\nxt_{\otype,\otype'} \subseteq (\S\times\S'\times\bbB^2)$ constraints the target origins
of consecutive output positions labelled by $\otype$ and $\otype'$. 
The resynchronization pairs $(\synch,\synch')$ in this case 
are induced by $\big((\move_\otype)_{\otype\in\G\times\G'},(\nxt_{\otype,\otype'})_{\otype,\otype'\in\G\times\G'}\big)$
and are obtained by projecting
the input and output over the original alphabets $\S$ and $\G$, 
under the assumption that the annotations satisfy $\ipar$ and $\opar$. 
A resynchronizer with parameters is called \emph{regular} 
if all its relations are regular. A regular resynchronizer 
is called \emph{bounded} if it is $k$-bounded, for some $k$.

In \cite{bmpp18} it was shown that, given a bounded regular resynchronizer 
$R$ with parameters and a two-way transducer $T$ with common guess, 
one can construct a two-way transducer $T'$ with common guess such that 
$T' \oeq R(T)$.  The notation $T' \oeq R(T)$ is used to represent the 
fact that $T'$ and $R(T)$ define the same relation in the origin semantics. 

\medskip

\emph{Unless otherwise stated, hereafter we assume that two-way transducers 
are enhanced with common guess, and regular resynchronizers are enhanced 
with parameters.}

\paragraph*{Rational vs regular resynchronizers.}

Our  first result  shows that bounded, regular resynchronizers are 
more expressive than rational resynchronizers.  Consider for instance
Example \ref{ex:rational-resync}: it can be captured by the regular
resynchronizer with $\opar$ annotating even/odd positions. The
resynchronizer shifts the origins of the even positions of the output
by one to the left and keeps the origins of the odd positions
unchanged. So here $\move_\otype$ can be described by a regular language.
On the other hand, Example \ref{ex:regular-resync} shows that there are bounded, 
regular resynchronizers that cannot be captured by rational resynchronizers.
	
\begin{restatable}{theorem}{RationalVsRegular}\label{thm:rational-vs-regular}
For every rational resynchronizer, there is an equivalent $1$-bounded regular resynchronizer.
\end{restatable}

\ifARXIV
The proof of the above result is rather technical, so we first provide a rough
idea.
\fi
Consider a rational resynchronizer $R$, that is, 
a one-way transducer that transforms synchronized 
\mywraptextjustified{7cm}{%
\ifARXIV
words while preserving the input and output
\fi
projections. For example, the figure to the right 
    represents
	a possible pair of synchronized words,
	denoted $w$ and $w'$, shown in blue and in red, respectively,
        such that $(w,w') \in R$.
	We assume that $\S=\{a\}$ and 
}
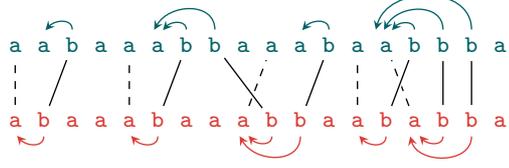
\mywrapfig{7cm}{%
	\vspace{-4.5mm}\hspace{4mm}%
\begin{tikzpicture}[xscale=1.25,yscale=0.8,baseline]
	
	\draw [nicecyan] (0,1.25) node [hidden] (1) {\small $\texttt{a}$};
	\draw [nicecyan] (0.3,1.25) node [hidden] (2) {\small $\texttt{a}$};
	\draw [nicecyan] (0.6,1.25) node [hidden] (3) {\small $\texttt{b}$};
	\draw [nicecyan] (0.9,1.25) node [hidden] (4) {\small $\texttt{a}$};
	\draw [nicecyan] (1.2,1.25) node [hidden] (5) {\small $\texttt{a}$};
	\draw [nicecyan] (1.5,1.25) node [hidden] (6) {\small $\texttt{a}$};
	\draw [nicecyan] (1.8,1.25) node [hidden] (7) {\small $\texttt{b}$};
	\draw [nicecyan] (2.1,1.25) node [hidden] (8) {\small $\texttt{b}$};
	\draw [nicecyan] (2.4,1.25) node [hidden] (9) {\small $\texttt{a}$};
	\draw [nicecyan] (2.7,1.25) node [hidden] (10) {\small $\texttt{a}$};
	\draw [nicecyan] (3.0,1.25) node [hidden] (11) {\small $\texttt{a}$};
	\draw [nicecyan] (3.3,1.25) node [hidden] (12) {\small $\texttt{b}$};
	\draw [nicecyan] (3.6,1.25) node [hidden] (13) {\small $\texttt{a}$};
	\draw [nicecyan] (3.9,1.25) node [hidden] (14) {\small $\texttt{a}$};
	\draw [nicecyan] (4.2,1.25) node [hidden] (15) {\small $\texttt{b}$};
	\draw [nicecyan] (4.5,1.25) node [hidden] (16) {\small $\texttt{b}$};
	\draw [nicecyan] (4.8,1.25) node [hidden] (17) {\small $\texttt{b}$};
	\draw [nicecyan] (5.1,1.25) node [hidden] (18) {\small $\texttt{a}$};
	
	\draw [nicered] (0,0) node [hidden] (1') {\small $\texttt{a}$};
	\draw [nicered] (0.3,0) node [hidden] (2') {\small $\texttt{b}$};
	\draw [nicered] (0.6,0) node [hidden] (3') {\small $\texttt{a}$};
	\draw [nicered] (0.9,0) node [hidden] (4') {\small $\texttt{a}$};
	\draw [nicered] (1.2,0) node [hidden] (5') {\small $\texttt{a}$};
	\draw [nicered] (1.5,0) node [hidden] (6') {\small $\texttt{b}$};
	\draw [nicered] (1.8,0) node [hidden] (7') {\small $\texttt{a}$};
	\draw [nicered] (2.1,0) node [hidden] (8') {\small $\texttt{a}$};
	\draw [nicered] (2.4,0) node [hidden] (9') {\small $\texttt{a}$};
	\draw [nicered] (2.7,0) node [hidden] (10') {\small $\texttt{b}$};
	\draw [nicered] (3.0,0) node [hidden] (11') {\small $\texttt{b}$};
	\draw [nicered] (3.3,0) node [hidden] (12') {\small $\texttt{a}$};
	\draw [nicered] (3.6,0) node [hidden] (13') {\small $\texttt{a}$};
	\draw [nicered] (3.9,0) node [hidden] (14') {\small $\texttt{b}$};
	\draw [nicered] (4.2,0) node [hidden] (15') {\small $\texttt{a}$};
	\draw [nicered] (4.5,0) node [hidden] (16') {\small $\texttt{b}$};
	\draw [nicered] (4.8,0) node [hidden] (17') {\small $\texttt{b}$};
	\draw [nicered] (5.1,0) node [hidden] (18') {\small $\texttt{a}$};

	\draw [dashed,link] (1') to (1);
	\draw [dashed,link] (5') to (5);
	\draw [dashed,link] (9') to (10);
	\draw [dashed,link] (13') to (13);
	\draw [dashed,link] (15') to (14);

	\draw [link] (2') to (3);
	\draw [link] (6') to (7);
	\draw [link] (10') to (8);
	\draw [link] (11') to (12);
	\draw [link] (14') to (15);
	\draw [link] (16') to (16);
	\draw [link] (17') to (17);
	
	\draw [arrow, nicecyan] (3.north) to [out=90,in=60,looseness=2] (2.north);
	\draw [arrow, nicecyan] (7.north) to [out=90,in=60,looseness=2] ([xshift=1pt]6.north);
	\draw [arrow, nicecyan] (8.north) to [out=90,in=70,looseness=2] ([xshift=-1.5pt]6.north);
	\draw [arrow, nicecyan] (12.north) to [out=90,in=60,looseness=2] (11.north);
	\draw [arrow, nicecyan] (15.north) to [out=90,in=60,looseness=2] ([xshift=1.25pt]14.north);
	\draw [arrow, nicecyan] (16.north) to [out=90,in=70,looseness=2] ([xshift=-1pt]14.north);
	\draw [arrow, nicecyan] (17.north) to [out=90,in=80,looseness=2] ([xshift=-3.25pt]14.north);

	\draw [arrow, nicered] (2'.south) to [out=-90,in=-60,looseness=2] (1'.south);
	\draw [arrow, nicered] (6'.south) to [out=-90,in=-60,looseness=2] (5'.south);
	\draw [arrow, nicered] (10'.south) to [out=-90,in=-60,looseness=2] ([xshift=1pt]9'.south);
	\draw [arrow, nicered] (11'.south) to [out=-90,in=-70,looseness=2] ([xshift=-1.5pt]9'.south);
	\draw [arrow, nicered] (14'.south) to [out=-90,in=-60,looseness=2] (13'.south);
	\draw [arrow, nicered] (16'.south) to [out=-90,in=-60,looseness=2] ([xshift=1pt]15'.south);
	\draw [arrow, nicered] (17'.south) to [out=-90,in=-80,looseness=2] ([xshift=-1.5pt]15'.south);
\end{tikzpicture}%
\hspace{-4mm}%
%
	\vspace{-3mm}%
}
$\G=\{b\}$. 

From the given rational resynchronizer $R$ we construct an 
equivalent $1$-bounded, regular resynchronizer
$R'$. The natural approach is to encode a successful run $\r$ of $R$ 
over a synchronized word $w$.
By measuring the differences between the partial inputs and the 
partial outputs that are consumed and produced along the run $\r$, 
we obtain a partial bijection on the input letters
that represents a mapping from source origins to target origins.
This mapping determines the relation $\move_\otype$ of $R'$, and
in fact depends on a suitable additional annotation $\otype$ of 
the underlying output position. The additional annotation is needed 
in order to distinguish output elements with the same origin in the 
source, but with different origins in the target.

For example, by referring again to the figure above, consider the first 
occurrence of $b$ in $w$. Its origin in $w$ is given by the closest
input letter to the left (follow the blue arrow). To find the 
origin in $w'$, one finds the same occurrence of $b$ in 
$w'$ (solid line), then moves to the closest input letter to the
left (red arrow), and finally maps the latter input position in $w'$
back to $w$ (dashed line). The resulting position determines
the new origin (w.r.t.~$w'$) of the considered output element.

The remaining components $\ipar$, $\opar$, and $\nxt_{\otype,\otype'}$
of $R'$ are used to guarantee the correctness of the various annotations
(notably, the correctness of the encoding of the run $\r$ and that of 
the output annotations).

\ifARXIV
\bigskip
The rest of the section is devoted to a formal proof of Theorem \ref{thm:rational-vs-regular}.

We fix a one-way transducer $R$ over $\S\uplus\G$ that defines a rational resynchronizer.
We assume without loss of generality that $R$ is \emph{letter-to-letter}, as well as \emph{trimmed}, 
namely, every state in $R$ occurs in some successful run.
Note that $R$ maps synchronized words to synchronized words.
With a slight abuse of terminology, we shall use the terms `source' (resp.~`target') 
to refer to a synchronized word that is an input (resp.~an output) of $R$. 
When depicting examples, we will often adopt the convention that source synchronized
words are shown in blue, while target synchronized words are shown in red.
On the other hand, we shall use the terms `input' and `output' to refer to 
the projections of a synchronized word over $\S$ and $\G$, respectively
(note that, in this case, it does not matter whether the synchronized word
is the source or the target, since these have the same projections over 
$\S$ and $\G$).
The goal is to construct a $1$-bounded, regular resynchronizer $R'$, with parameters,
that defines the same resynchronization as $R$. 

We begin by introducing the key concept of \emph{lag}, which represents the difference
between the number of input symbols consumed and number of input symbols produced along
a certain run (not necessarily successful) of $R$. Formally, given a run of $R$ of the 
form $\r = q_0 \act{c_1 \:|\: d_1} q_1 \act{c_2 \:|\: d_2} \dots \act{c_n \:|\: d_n} q_n$,
we define its lag $\il(\r)$ as $|\projS(c_1\ldots
c_n)|-|\projS(d_1\ldots d_n)|$, 
where $\projS$ denotes the operation of projection onto the alphabet $\S$.
Note that, because $R$ is letter-to-letter, one could have equally defined 
$\il(\r)$ by counting the difference between produced output symbols and consumed output symbols.
Further note that the lag of a successful run is always $0$, since $R$ preserves the input projection.
Notice that the lag of a run is a notion distinct of the \emph{delay} of a 
rational resynchronizer presented in \cite{FiliotJLW16}
which is the maximum distance between the target origin of an output position and its source origin.
The following lemma shows that the lag is in fact a property of the initial 
and final states of a run. 
				
\begin{lemma}\label{lem:rational-lag}
For every two runs $\r_1$ and $\r_2$ of $R$ that begin with the same state and end with the same state,
$\il(\r_1)=\il(\r_2)$. 
\end{lemma}
	
\begin{proof}
Since $R$ is trimmed, both runs $\r_1$ and $\r_2$ can be completed to some 
successful runs of the form $\r' \r_1 \r''$ and $\r' \r_2 \r''$.
From $\il(\r' \r_1 \r'') = 0 = \il(\r' \r_2 \r'')$, it immediately follows 
that $\il(\r_1) = 0-(\il(\r') + \il(\r'')) = \il(\r_2)$.
\end{proof}
		
In view of the above lemma, we can associate a lag $\il(q)$ with each state $q$ of $R$
as follows: we choose an arbitrary run $\r$ that starts with the initial state of $R$
and ends with $q$, and let $\il(q)=\il(\r)$. This is well-defined since 
$\il(q)$ does not depend on the particular choice of $\r$.
For instance, if we consider the letter-to-letter resynchronizer $R$ of Example \ref{ex:rational-resync},
the only state with non-zero lag is the bottom one, which has lag $1$.
Note that, because each transition of $R$ can only increase or decrease the lag
by $1$, all lags range over the finite set $\{-|Q|,\dots,+|Q|\}$, where
$Q$ is the state space of $R$.

Next, we consider a successful run of $R$, say 
$\r= q_0 \act{c_1 \:|\: d_1} q_1 \act{c_2 \:|\: d_2} \dots \act{c_n \:|\: d_n} q_n$,
and define relations $\om_\r$ and $\im_\r$ between positions of $\r$.
These relations are used later to define a bijection between source and target origins.
The relation $\om_\r$ consists of all pairs $(i,j)$ of positions of $\r$ such that
$c_i$ and $d_j$ are output letters and 
$c_1 c_2\dots c_i \eqG d_1 d_2 \dots d_j$
(the latter is a shorthand for $\projG(c_1 c_2\dots c_i)  = \projG(d_1 d_2 \dots d_j)$). 
Note that $\om_\r$ is in fact a partial bijection.
In a similar way, we define $\im_\r$ as the partial bijection that contains all pairs $(i,j)$ 
of positions of $\r$ such that $c_i$ and $d_j$ are input letters 
and $c_1 c_2 \dots c_i \eqS d_1 d_2 \dots d_j$. 
	
\begin{example}\label{ex:match}
\ifLIPICS
Consider the pair of source and target synchronized 
words over $\S\uplus\G$ shown
\fi
\ifARXIV
We consider again the pair of source and target synchronized 
words that we used beofre to give an intuitive account of 
the proof of Theorem \ref{thm:rational-vs-regular}.
These synchronized words are  
\fi
\mywraptextjustified{7cm}{%
\ifLIPICS
to the right, where $\S=\{a\}$ and $\G=\{b\}$, 
which could be realized by a successful run $\r$ of $R$. 
For the moment, we overlook the blue and red arrows.
Because $R$ is letter-to-letter, any position in any of the
two words corresponds precisely to a position in the run
$\r$, so we can represent the
\fi
\ifARXIV
depicted
to the right, and assumed to be realized by a 
successful run $\r$ of $R$.
For the moment, we overlook the blue and red arrows.
Because $R$ is letter-to-letter, any position in any of the
two words corresponds precisely to a position in the run
$\r$, so we can represent the
relations $\om_\r$ 
\fi
}%
\mywrapfig{7cm}{%
	\vspace{-3.5mm}\hspace{4mm}%
\begin{tikzpicture}[xscale=1.25,yscale=0.8,baseline]
	
	\draw [nicecyan] (0,1.25) node [hidden] (1) {\small $\texttt{a}$};
	\draw [nicecyan] (0.3,1.25) node [hidden] (2) {\small $\texttt{a}$};
	\draw [nicecyan] (0.6,1.25) node [hidden] (3) {\small $\texttt{b}$};
	\draw [nicecyan] (0.9,1.25) node [hidden] (4) {\small $\texttt{a}$};
	\draw [nicecyan] (1.2,1.25) node [hidden] (5) {\small $\texttt{a}$};
	\draw [nicecyan] (1.5,1.25) node [hidden] (6) {\small $\texttt{a}$};
	\draw [nicecyan] (1.8,1.25) node [hidden] (7) {\small $\texttt{b}$};
	\draw [nicecyan] (2.1,1.25) node [hidden] (8) {\small $\texttt{b}$};
	\draw [nicecyan] (2.4,1.25) node [hidden] (9) {\small $\texttt{a}$};
	\draw [nicecyan] (2.7,1.25) node [hidden] (10) {\small $\texttt{a}$};
	\draw [nicecyan] (3.0,1.25) node [hidden] (11) {\small $\texttt{a}$};
	\draw [nicecyan] (3.3,1.25) node [hidden] (12) {\small $\texttt{b}$};
	\draw [nicecyan] (3.6,1.25) node [hidden] (13) {\small $\texttt{a}$};
	\draw [nicecyan] (3.9,1.25) node [hidden] (14) {\small $\texttt{a}$};
	\draw [nicecyan] (4.2,1.25) node [hidden] (15) {\small $\texttt{b}$};
	\draw [nicecyan] (4.5,1.25) node [hidden] (16) {\small $\texttt{b}$};
	\draw [nicecyan] (4.8,1.25) node [hidden] (17) {\small $\texttt{b}$};
	\draw [nicecyan] (5.1,1.25) node [hidden] (18) {\small $\texttt{a}$};
	
	\draw [nicered] (0,0) node [hidden] (1') {\small $\texttt{a}$};
	\draw [nicered] (0.3,0) node [hidden] (2') {\small $\texttt{b}$};
	\draw [nicered] (0.6,0) node [hidden] (3') {\small $\texttt{a}$};
	\draw [nicered] (0.9,0) node [hidden] (4') {\small $\texttt{a}$};
	\draw [nicered] (1.2,0) node [hidden] (5') {\small $\texttt{a}$};
	\draw [nicered] (1.5,0) node [hidden] (6') {\small $\texttt{b}$};
	\draw [nicered] (1.8,0) node [hidden] (7') {\small $\texttt{a}$};
	\draw [nicered] (2.1,0) node [hidden] (8') {\small $\texttt{a}$};
	\draw [nicered] (2.4,0) node [hidden] (9') {\small $\texttt{a}$};
	\draw [nicered] (2.7,0) node [hidden] (10') {\small $\texttt{b}$};
	\draw [nicered] (3.0,0) node [hidden] (11') {\small $\texttt{b}$};
	\draw [nicered] (3.3,0) node [hidden] (12') {\small $\texttt{a}$};
	\draw [nicered] (3.6,0) node [hidden] (13') {\small $\texttt{a}$};
	\draw [nicered] (3.9,0) node [hidden] (14') {\small $\texttt{b}$};
	\draw [nicered] (4.2,0) node [hidden] (15') {\small $\texttt{a}$};
	\draw [nicered] (4.5,0) node [hidden] (16') {\small $\texttt{b}$};
	\draw [nicered] (4.8,0) node [hidden] (17') {\small $\texttt{b}$};
	\draw [nicered] (5.1,0) node [hidden] (18') {\small $\texttt{a}$};

	\draw [dashed,link] (1') to (1);
	\draw [dashed,link] (5') to (5);
	\draw [dashed,link] (9') to (10);
	\draw [dashed,link] (13') to (13);
	\draw [dashed,link] (15') to (14);

	\draw [link] (2') to (3);
	\draw [link] (6') to (7);
	\draw [link] (10') to (8);
	\draw [link] (11') to (12);
	\draw [link] (14') to (15);
	\draw [link] (16') to (16);
	\draw [link] (17') to (17);
	
	\draw [arrow, nicecyan] (3.north) to [out=90,in=60,looseness=2] (2.north);
	\draw [arrow, nicecyan] (7.north) to [out=90,in=60,looseness=2] ([xshift=1pt]6.north);
	\draw [arrow, nicecyan] (8.north) to [out=90,in=70,looseness=2] ([xshift=-1.5pt]6.north);
	\draw [arrow, nicecyan] (12.north) to [out=90,in=60,looseness=2] (11.north);
	\draw [arrow, nicecyan] (15.north) to [out=90,in=60,looseness=2] ([xshift=1.25pt]14.north);
	\draw [arrow, nicecyan] (16.north) to [out=90,in=70,looseness=2] ([xshift=-1pt]14.north);
	\draw [arrow, nicecyan] (17.north) to [out=90,in=80,looseness=2] ([xshift=-3.25pt]14.north);

	\draw [arrow, nicered] (2'.south) to [out=-90,in=-60,looseness=2] (1'.south);
	\draw [arrow, nicered] (6'.south) to [out=-90,in=-60,looseness=2] (5'.south);
	\draw [arrow, nicered] (10'.south) to [out=-90,in=-60,looseness=2] ([xshift=1pt]9'.south);
	\draw [arrow, nicered] (11'.south) to [out=-90,in=-70,looseness=2] ([xshift=-1.5pt]9'.south);
	\draw [arrow, nicered] (14'.south) to [out=-90,in=-60,looseness=2] (13'.south);
	\draw [arrow, nicered] (16'.south) to [out=-90,in=-60,looseness=2] ([xshift=1pt]15'.south);
	\draw [arrow, nicered] (17'.south) to [out=-90,in=-80,looseness=2] ([xshift=-1.5pt]15'.south);
\end{tikzpicture}%
\hspace{-4mm}%
%
	\vspace{-3mm}
}
\ifLIPICS
relations $\om_\r$ 
\fi
and $\im_\r$ by means of edges between source and target positions. 
In the figure, the solid edges represent pairs of $\om_\r$, while
the dashed edges represent some pairs of $\im_\r$ 
(precisely, those pairs $(i,j)$ such that
the transition at position $j$ produces an input letter, while the next
transition produces an output letter).
\end{example}

\paragraph*{Mapping the source to target origins.} 
We now explain how the relations $\im_\r$ and $\om_\r$ 
can be used to define a mapping from source to target origins. 
We do so by first using the figure of Example \ref{ex:match}.
Consider any output letter at position $i$ in the source synchronized 
word $w$ (e.g.~the first blue letter $b$). 
Let $j$ be the last $\S$-labelled position before $i$, as indicated
by the blue arrow. This position $j$ determines the source origin 
$y=|\projS(w[1,j])|$ of the output letter.
To find the corresponding target origin, we observe that the position 
$i$ is mapped via the relation $\om_\r$ (solid line) to some position 
$k$ in the target synchronized word.
Let $h$ be the last $\S$-labelled position before $k$ (red arrow),
and map $h$ back to a position $\ell$ in the source via the relation
$\im_\r$ (dashed line).
The position $\ell$ determines precisely the target origin
$z=|\projS(w[1,\ell])|$ of the considered output letter.
The above steps describe a correspondence between two positions 
$j$ and $\ell$ in $\r$, with labels over $\S$, that is precisely defined by
\begin{align*}
  \exists i,k,h ~~  
  \begin{cases}
  \r[j,i] \text{ consumes a word in } \S\G^+ \\
  (i,k) \in \om_\r \\
  \r[h,k] \text{ produces a word in } \S\G^+ \\
  (h,\ell)\in \im_\r.
  \end{cases}
  \tag{$\star$}
\end{align*}
In the above $\rho[j,i]$ represents the part of $\rho$ between positions $j, i$ (both $j,i$ included).  

We denote by $\match_\r$ the relation of all pairs $(j,\ell)$
that satisfy Equation ($\star$).
Note that $\match$ determines an analogous correspondence 
between source and target origins of the input projection.
However, $\match$ has two issues: 
it is not yet a partial bijection (since different output positions 
may have the same source origin), and it needs to be implemented 
by means of a regular relation $\move_\otype$ that only considers 
positions of the input, plus the label $\otype$ of a single position 
in the output. 
Below, we explain how to overcome those issues.

\paragraph*{The case of bounded output blocks.}
Hereafter, we call \emph{output block} any maximal factor of a synchronized 
word that is labelled over $\G$. Intuitively, this corresponds to a maximal 
factor of the output that originates at the same input position.
We first consider, as a simpler case, a rational resynchronizer $R$ that 
reads {\sl source} synchronized words where the lengths of the output blocks 
are uniformly bounded by some constant, say $\bound$ (a similar property holds 
for the blocks of the target synchronized words, using lag-based arguments).
In this case we can encode any successful run $\r$ of $R$ entirely on the input,
by annotating every $\S$-labelled position $y$ with a factor $\r_y$
of $\r$ that reads the input symbol at position $y$, followed by the sequence 
of output symbols up to the next input symbol. Note that every factor $\r_y$ 
has length at most $\bound+1$.
The correctness of this input annotation can be checked by the regular
language $\ipar$. Given a factor $\rho_y \in \Sigma \Gamma^+$, 
$\rho_y[1] \in \Sigma$ is the first position of the factor $\rho_y$. 
Likewise, $\rho_y[i,j]$ denotes the subfactor of $\rho_y$ 
consisting of positions $i, i+1, \dots, j$.

In addition, we also annotate the output word with indices 
from $\{1,\dots,\bound\}$, called \emph{offsets}, in such a way that 
an output position $x$ is annotated with an offset $o$ if and only if
it is the $o$-th output position with the same source origin.
Note that the correctness of the annotation cannot be checked by 
a regular language such as $\opar$ that refers only to the output.
The check will be done instead by a combined use of the relations
$\move_\otype$ and $\nxt_{\otype,\otype'}$.

We first check that, for every pair of consecutive output positions $x$ 
and $x+1$ annotated with the offsets $o$ and $o'$, respectively, 
it holds that $o'=o+1$ or $o'=1$, depending on whether the 
\emph{source} origins of $x$ and $x+1$ coincide or not. For this
we let $(u,z,z') \in \nxt_{\otype,\otype'}$, with $\otype=(a,o)$
and $\otype'=(a',o')$, if
\begin{enumerate}
  \item either $o'=o+1$ and there is $y=y'$ such that
        $(u,y,z)\in\move_\otype$ and $(u,y',z')\in\move_{\otype'}$,
  \item or $o'=1$ and there are $y<y'$ such that
        $(u,y,z)\in\move_\otype$ and $(u,y',z')\in\move_{\otype'}$.
\end{enumerate}
Recall that the relation $\nxt_{\otype,\otype'}$ must be
defined in terms of the \emph{target} origins of $x$ and $x+1$. 
So it needs to rely on the relation $\move_\otype$
in order to determine the source origins 
from the target origins. We assume that for every 
output type $\otype$ 
the relation $\move_\otype$, 
which will be defined later, determines a~\emph{partial bijection} 
between input positions (we will see that this is indeed the case). 
Based on these assumptions, the above definition of $\nxt_{\otype,\otype'}$ 
guarantees that the offsets annotating consecutive positions 
in the output are either incremented or reset, depending on whether
they have the same origin or not.

It remains to check that maximal offset occurring in an output block 
with origin $y$ coincides with number of output symbols produced by 
the corresponding factor $\rho_y$ of the run.
Thus, we modify slightly the definition of $\nxt_{\otype,\otype'}$
in case 2., as follows:
\begin{enumerate}
  \item[2'.] or $o'=1$ and there are $y<y'$ such that
        $(u,y,z)\in\move_\otype$ and $(u,y',z')\in\move_{\otype'}$,
        \emph{and $o=|\rho_y|-1$}.
\end{enumerate}
Note that the factor $\rho_y$ can be derived by inspecting the
annotation of the input position $y$. 
The modification suffices to guarantee that the output annotation is 
correct for all output blocks but the last one. 
The annotation for the last output block can be checked by
marking the last output position with a distinguished symbol
and by requiring that if $\otype$ witnesses the marked symbol
and the offset $o$, then $\move_\otype$ can only contain a triple 
of the form $(u,y,z)$, with $o=|\rho_y|-1$.
We omit the tedious definitions in this case.


\medskip
Now, having the input correctly annotated with the factors $\r_y$ 
of $\r$ and the output correctly annotated with the offsets, 
we can encode any position $i$ of $\r$ by a pair 
$(y,o)$ that consists of a position $y$ of the input and an offset 
$o\in\{0,1,\dots,\bound\}$. 
The encoding is defined in such a way that $i = \sum_{y'<y}|\r_{y'}| + o + 1$
(in particular, $o=0$ when the transition at position $i$ consumes an
input symbol, otherwise $o\ge1$).
We use this encoding to translate the relations $\om_\r$, $\im_\r$, 
and $\match_\r$, to equivalent finite unions of partial bijections between 
input positions.
We begin by explaining the translation of $\om_\r$.

\paragraph*{Translation of $\om_\r$.} 
Consider any pair $(i,j)\in\om_\r$. Since the transition at position
$i$ of $\r$ consumes an output symbol, it is encoded by a pair of the form 
$(y,o)$, with $o\ge1$.
On the other hand, the transition at position $j$ may consume 
either an input symbol or an output symbol (but does produce an output symbol).
In the former case, $j$ is encoded by a pair $(y',0)$;
in the latter case, it is encoded by a pair $(y',o')$, with $o'\ge1$.
As an example, in the figure below, $(7,4) \in \om_\r$. Position 7 of
the run is encoded as $(5,1)$ on the input.   
The transition at position 4 consumes an input symbol $a$, and 
produces the output symbol $b$, and is encoded as $(3,0)$.  


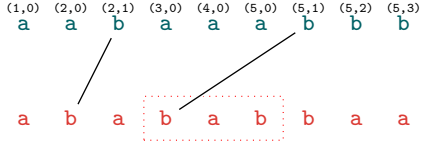
\begin{tikzpicture}[xscale=1.25,baseline]
	\draw [nicecyan] (0,1.25) node [hidden] (1) {\small $\texttt{a}$};
	\draw [black] (0,1.5) node [hidden] (1'') {\tiny $\texttt{(1,0)}$};
	\draw [nicecyan] (0.5,1.25) node [hidden] (2) {\small $\texttt{a}$};
	\draw [black] (0.5,1.5) node [hidden] (2'') {\tiny $\texttt{(2,0)}$};
	\draw [nicecyan] (1,1.25) node [hidden] (3) {\small $\texttt{b}$};
	\draw [black] (1,1.5) node [hidden] (3'') {\tiny $\texttt{(2,1)}$};

	\draw [nicecyan] (1.5,1.25) node [hidden] (4) {\small $\texttt{a}$};
	
	\draw [black] (1.5,1.5) node [hidden] (4'') {\tiny $\texttt{(3,0)}$};

	\draw [nicecyan] (2,1.25) node [hidden] (5) {\small $\texttt{a}$};
	\draw [black] (2,1.5) node [hidden] (5'') {\tiny $\texttt{(4,0)}$};

	\draw [nicecyan] (2.5,1.25) node [hidden] (6) {\small $\texttt{a}$};
	\draw [black] (2.5,1.5) node [hidden] (6'') {\tiny $\texttt{(5,0)}$};

	\draw [nicecyan] (3,1.25) node [hidden] (7) {\small $\texttt{b}$};
	\draw [black] (3,1.5) node [hidden] (7'') {\tiny $\texttt{(5,1)}$};
	
	\draw [nicecyan] (3.5,1.25) node [hidden] (8) {\small $\texttt{b}$};
	\draw [black] (3.5,1.5) node [hidden] (8'') {\tiny $\texttt{(5,2)}$};

	\draw [nicecyan] (4,1.25) node [hidden] (9) {\small $\texttt{b}$};
	\draw [black] (4,1.5) node [hidden] (9'') {\tiny $\texttt{(5,3)}$};

	\draw [nicered] (0,0) node [hidden] (1') {\small $\texttt{a}$};
	\draw [nicered] (0.5,0) node [hidden] (2') {\small $\texttt{b}$};
	
	\draw [nicered] (1,0) node [hidden] (3') {\small $\texttt{a}$};
	\draw [nicered] (1.5,0) node [hidden] (4') {\small $\texttt{b}$};

	\draw [nicered] (2,0) node [hidden] (5') {\small $\texttt{a}$};
	\draw [nicered] (2.5,0) node [hidden] (6') {\small $\texttt{b}$};
	\draw [nicered] (3,0) node [hidden] (7') {\small $\texttt{b}$};
	\draw [nicered] (3.5,0) node [hidden] (8') {\small $\texttt{a}$};
	\draw [nicered] (4,0) node [hidden] (9') {\small $\texttt{a}$};
\node[draw=red,dotted,fit=(4')(5')(6')] {};
	\draw [link] (2') to (3);
	\draw [link] (4') to (7);
\end{tikzpicture}
\medskip

\noindent
In general, we observe that the lag induced just after the 
$o$-th transition of $\r_y$ must be equal to the number 
of output symbols produced between the $(o'+1)$-th 
transition of $\r_{y'}$ and the $o$-th transition of $\r_y$,
both included
(when the lag is negative one follows the transitions
 in reverse order, counting negatively). As an illustration in the figure,   
 the lag after the first transition of $\rho_5$ is 2, which is the number 
 of output symbols in the dotted box. The dotted box 
 consists of the symbols produced between the first transition 
 of $\r_{3'}$ and the first transition of $\r_5$, and has two output symbols.

\paragraph*{Translation of $\im_{\r}$.} 
The translation of the relation $\im_\r$ is similar.
The only difference is that now the pairs $(i,j)\in\im_\r$ 
are encoded by tuples of the form 
$\big((y,o),(y',o')\big)$, with $o=0$ since the transition at $i$ 
consumes an input symbol. The transition at position $j$ as before, 
can consume an input symbol or an output symbol. 
Consider the figure below, where $(2,3) \in \im_{\r}$. 
Position $i=2$ is encoded as 
$(2,0)$.  The transition at position 3 
consumes an output symbol $b$ (and produces the input symbol $a$). 
Position 3 is encoded as $(2,1)$. 
	
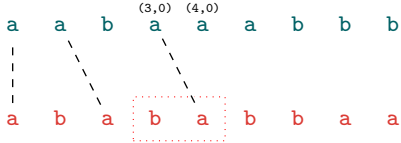
\begin{tikzpicture}[xscale=1.25,baseline]
	\draw [nicecyan] (6,1.25) node [hidden] (1) {\small $\texttt{a}$};
		\draw [nicecyan] (6.5,1.25) node [hidden] (2) {\small $\texttt{a}$};
	\draw [nicecyan] (7,1.25) node [hidden] (3) {\small $\texttt{b}$};
	\draw [nicecyan] (7.5,1.25) node [hidden] (4) {\small $\texttt{a}$};
	\draw [black] (7.5,1.5) node [hidden] (4'') {\tiny $\texttt{(3,0)}$};
	
	\draw [nicecyan] (8,1.25) node [hidden] (5) {\small $\texttt{a}$};
	\draw [black] (8,1.5) node [hidden] (5'') {\tiny $\texttt{(4,0)}$};
	
	\draw [nicecyan] (8.5,1.25) node [hidden] (6) {\small $\texttt{a}$};
	\draw [nicecyan] (9,1.25) node [hidden] (7) {\small $\texttt{b}$};
	
	\draw [nicecyan] (9.5,1.25) node [hidden] (8) {\small $\texttt{b}$};
	\draw [nicecyan] (10,1.25) node [hidden] (9) {\small $\texttt{b}$};
	\draw [nicered] (6,0) node [hidden] (1') {\small $\texttt{a}$};
	\draw [nicered] (6.5,0) node [hidden] (2') {\small $\texttt{b}$};
	\draw [nicered] (7,0) node [hidden] (3') {\small $\texttt{a}$};
	\draw [nicered] (7.5,0) node [hidden] (4') {\small $\texttt{b}$};
	\draw [nicered] (8,0) node [hidden] (5') {\small $\texttt{a}$};
	\draw [nicered] (8.5,0) node [hidden] (6') {\small $\texttt{b}$};
	\draw [nicered] (9,0) node [hidden] (7') {\small $\texttt{b}$};
	\draw [nicered] (9.5,0) node [hidden] (8') {\small $\texttt{a}$};
	\draw [nicered] (10,0) node [hidden] (9') {\small $\texttt{a}$};
\node[draw=red,dotted,fit=(4') (5')] {};
	\draw [dashed,link] (1') to (1);
	\draw [dashed,link] (3') to (2);
	\draw [dashed,link] (4) to (5');
\end{tikzpicture}
\medskip

\noindent
The only difference here is that 
one has to relate the lag with the 
number of {\sl input} letters produced between (both positions included) the first transition 
of $\r_y$ and the $o'$-th transition 
of $\r_{y'}$. 
Again, in the figure, the lag after the first transition of $\r_3$ is 
1, which is the number of input symbols in the 
dotted box. The dotted box contains the symbols 
produced between the first transition of 
$\r_3$ and the first transition of $\rho_{4'}$, and has one input 
symbol.

\paragraph*{Relations encoding $\om_\r$ and $\im_\r$.} 
So we can represent $\om_\r$ as a finite union of
relations $O_{o,o'}\subseteq(\S\times\S')^*\times\Nat\times\Nat$,
each describing a regular property of annotated inputs with two 
distinguished positions in it, in such a way that the positions
are bijectively related to one another.

Likewise, we can represent $\im_\r$ as a finite union of relations 
$I_{0,o'}$, each describing a regular property of annotated inputs 
with two distinguished positions encoded as $(y,0)$ and $(y', o')$ in it, 
which are bijectively related to one another.

\paragraph*{Translation of $\match_\r$.} 
We finally turn to the translation of the relation $\match_\r$,
which will eventually determine the relations $\move_\otype$ 
of the desired regular resynchronizer $R'$.
This is done by mimicking Equation ($\star$) via the encoding
of  positions in the run $\rho$ using pairs of input positions and offsets,
and more precisely, by replacing the variables $j,i,k,h,\ell$ 
of Equation (*) with the pairs $(y,0)$, $(y,o)$, $(y',o')$, $(y'',o'')$, $(z,0)$.

Formally, for every offset $o\in\{1,\dots,\bound\}$, we define 
the set $M_o$ of all triples $(u,y,z)$, where $u$ is an annotated 
input and $y,z$ are positions in it that satisfy the following 
property:

\begin{align*}
  \exists y',y''  
  \bigvee_{0\le o',o''\le\bound} ~~ 
  \begin{cases}
  \r_y[1,o+1] \text{ consumes a word in $\S\G^+$} \\
  (u,y,y') \in O_{o,o'} \\
  \r_{y''}[o''+1,|\r_{y''}|] ~ \r_{y''+1} \dots \r_{y'-1} ~ \r_{y'}[1,o'+1] 
    \text{ produces a word in $\S\G^+$} \\
  (u,z,y'') \in I_{0,o''}.
  \tag{$\star\star$}
  \end{cases}
\end{align*}

Note that the first condition holds trivially by definition of $\r_y$, 
while the third condition is easily implemented by accessing the
factors $\r_{y''},\dots, \r_{y'}$
of $\r$ that are encoded by the input parameters. 
For simplicitly, here we assumed that $(y'',o'')$ is lexicographically 
before $(y',o')$; to treat the symmetric case, one has to interpret 
the definition by considering the sequence of transitions in reverse.
The intended meaning of $(u,y,z)\in M_o$ is as follows. 
Suppose that the input is correctly annotated with the factors $\r_y$
of a successful run $\r$ of $R$, and that the output 
position $x$ of $\r$ is correctly annotated with an offset $o$. 
Assuming that $x$ is the $o$-th output position with source origin $y$, 
then $z$ is its target origin in $\r$.

Continuing with our running example, we determine the target origin 
for the point $b$ annotated (5,1), whose source origin is 
(5,0). We will find the target origin of this $b$ annotated (5,1). As seen in the computation of $\om_\r$, we know that 
$(u,5,3) \in O_{1,0}$.  The factor  $\rho_5=abbb$, and 
$\rho_5[1,2]=ab \in \Sigma \Gamma^+$, and as we have seen, 
$(u,5,3) \in O_{1,0}$. Now, consider the part of the source $u$ annotated with 
$(2,1)(3,0)$. This produces the output $ab \in \Sigma \Gamma^+$. That is, for $y''=2, o''=1$, and $y'=3, o'=0$, we have
$\rho_{y''}[o''+1, 2] \:  \rho_{y'}[1,o'+1]=\rho_2[2,2]\rho_3[1,1]=ba$ produces the output 
$ab \in \Sigma \Gamma^+$. 

Consider $(z,0)=(2,0)$. The lag after the $a$ at $i=2$ annotated $(2,0)$ is 1. Also, 
$(2,3) \in \im_\r$. The position 3 consumes an 
output and produces an input $a$. 
Indeed, the lag after the first transition of $\rho_2$ is 
1, which is the number of input symbols between the first transition 
of $\rho_2$ and the second transition ($(o'+1)$th transition) of  
$\rho_2$. That is, $(u,2,2) \in I_{0,1}$. Thus, starting with the $b$ annotated $(y,o)=(5,1)$ such that 
$\rho_5[1,2] \in \Sigma \Gamma^+$,  
we first obtain $(y',o')=(3,0)$ with $(u,5,3) \in O_{1,0}$. Further, 
 $\rho_2[2,2]\rho_3[1,1]$ produces a word in $\Sigma \Gamma^+$. 
 Finally, we have $(u,2,2) \in I_{0,1}$, obtaining $(u,5,2) \in M_1$.

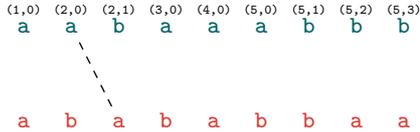
\begin{tikzpicture}[xscale=1.25,baseline]
	\draw [nicecyan] (6,1.25) node [hidden] (1) {\small $\texttt{a}$};
	\draw [black] (6,1.5) node [hidden] (1') {\tiny $\texttt{(1,0)}$};
			\draw [nicecyan] (6.5,1.25) node [hidden] (2) {\small $\texttt{a}$};
	\draw [black] (6.5,1.5) node [hidden] (2') {\tiny $\texttt{(2,0)}$};
		\draw [nicecyan] (7,1.25) node [hidden] (3) {\small $\texttt{b}$};
	\draw [black] (7,1.5) node [hidden] (3') {\tiny $\texttt{(2,1)}$};
			\draw [nicecyan] (7.5,1.25) node [hidden] (4) {\small $\texttt{a}$};
	\draw [black] (7.5,1.5) node [hidden] (4') {\tiny $\texttt{(3,0)}$};
		\draw [nicecyan] (8,1.25) node [hidden] (5) {\small $\texttt{a}$};
	\draw [black] (8,1.5) node [hidden] (5') {\tiny $\texttt{(4,0)}$};
		\draw [nicecyan] (8.5,1.25) node [hidden] (6) {\small $\texttt{a}$};
	\draw [black] (8.5,1.5) node [hidden] (6'') {\tiny $\texttt{(5,0)}$};
	\draw [nicecyan] (9,1.25) node [hidden] (7) {\small $\texttt{b}$};
	\draw [black] (9,1.5) node [hidden] (7'') {\tiny $\texttt{(5,1)}$};
	\draw [nicecyan] (9.5,1.25) node [hidden] (8) {\small $\texttt{b}$};
	\draw [black] (9.5,1.5) node [hidden] (7'') {\tiny $\texttt{(5,2)}$};
	\draw [nicecyan] (10,1.25) node [hidden] (9) {\small $\texttt{b}$};
	\draw [black] (10,1.5) node [hidden] (7'') {\tiny $\texttt{(5,3)}$};

	\draw [nicered] (6,0) node [hidden] (1') {\small $\texttt{a}$};
	\draw [nicered] (6.5,0) node [hidden] (2') {\small $\texttt{b}$};
	\draw [nicered] (7,0) node [hidden] (3') {\small $\texttt{a}$};
	\draw [nicered] (7.5,0) node [hidden] (4') {\small $\texttt{b}$};
	\draw [nicered] (8,0) node [hidden] (5') {\small $\texttt{a}$};
	\draw [nicered] (8.5,0) node [hidden] (6') {\small $\texttt{b}$};
	\draw [nicered] (9,0) node [hidden] (7') {\small $\texttt{b}$};
	\draw [nicered] (9.5,0) node [hidden] (8') {\small $\texttt{a}$};
	\draw [nicered] (10,0) node [hidden] (9') {\small $\texttt{a}$};
	\draw [dashed,link] (3') to (2);
	
\end{tikzpicture}

\smallskip
\paragraph*{Definition of $\move_\otype$.} 
It is tempting to define $\move_\otype$ just as $M_o$, 
for every $\otype=(a,o)\in\G\times\{1,\dots,\bound\}$. 
However, we recall that the correctness of the output annotation is guaranteed
only once we are sure that every relation $\move_\otype$ defines a partial 
bijection between input positions $y$ and $z$ (hereafter we say for short that
the relation is bijective), which is not known a priori.
Bijectiveness must then be enforced syntactically, without relying on annotations: 
for this it suffices to define $\move_\otype$ as 
$\{ (u,y,z)\in M_o \:\mid\: 
    \forall (u,y',z')\in M_o ~ (y=y')\leftrightarrow (z=z') \}$,
and observe that either $M_o$ is bijective, and hence
$\move_\otype=M_o$, 
or it is not, and in this case $\move_\otype$ is a subrelation of $M_o$
that is still bijective.
Note that, in the case where $\move_\otype$ is a subrelation of $M_o$,
there will be no induced pair of synchronized words, since the origins
of some output elements could not be redirected. 
This is fine, and actually needed, in order to avoid generating with $R'$
spurious pairs of synchronized words that are not also generated by $R$.
On the other hand, observe that the relation $\move_\otype$ does generate, 
for appropriate choices of the output annotations, all the pairs of 
synchronized words that are generated by $R$.
We finally observe that the relations $\move_\otype$ and $\nxt_{\otype,\otype'}$
are regular. We obtain in this way,  a $1$-bounded, regular resynchronizer $R'$
equivalent to $R$.

\paragraph*{The general case.} 
We now aim at generalizing the previous ideas to capture a
rational resynchronizer $R$ with source output blocks of possibly
unbounded length. One additional difficulty is that we cannot 
anymore encode a successful run $\r$ of $R$ entirely on the input, 
as $\r$ may have arbitrarily long factors on outputs blocks.
Another difficulty is that we cannot uniquely identify the
positions in an output block using offsets ranging over a fixed 
finite set. We will see that a solution to both problems comes 
from covering most of the output by factors in which the positions 
behave similarly in terms of the source-to-target origin transformation. 
Intuitively, each of these factors can be thought of as 
a `pseudo-position', and accordingly the output blocks can be thought 
of as having boundedly many pseudo-positions.
This will make it possible to apply the same ideas as before. 
We now state the key lemma that 
identifies the aforesaid factors.
By a slight abuse of terminology, we call output blocks 
also the maximal $\G$-labelled factors of a synchronized word.

\begin{lemma}\label{lem:blocks}
Let $\r$ 
be a successful run of $R$, and let $w$ and $w'$ 
be the source and target synchronized words induced by $\r$.
\begin{itemize}
\item Every output block $v$ of $w$ can be factorized into $\cO(|Q|^2)$ 
      sub-blocks $v_1,\dots,v_n$ such that
	  if $|v_i|>1$ and $\r_i$ is the factor of $\r$ that corresponds to $v_i$,
	  then all states in $\r_i$ have the same lag, say $\ell_i$,
	  and the factor obtained by extending $\r_i$ to the left and 
	  to the right by exactly $|\ell_i|$ transitions forms a loop of $R$.
\item Moreover, for every factorization $v=v_1\dots v_n$ as above,
      each sub-block $v_i$ is also a factor of $w'$, and hence all 
      positions in $v_i$ have the same target origin.
\end{itemize}
\end{lemma}

\begin{proof}
%
\begin{figure}
\centering  
\begin{tikzpicture}

\draw (0,0) edge [|-|] (10,0);

\draw [dashed,rounded corners=5] (1,-0.2) rectangle (4,0.2);
\draw [dashed,rounded corners=5] (5,-0.2) rectangle (9,0.2);

\draw [draw=none,pattern=north east lines] (1.75,-0.2) rectangle (3.25,0.2);
\draw [draw=none,pattern=north east lines] (5.75,-0.2) rectangle (8.25,0.2);

\draw (7.5,2) node [above] {singleton sub-blocks};
\draw [arrow,dotted] (7,2) to [out=-10,in=110] (10,0.3);
\draw [arrow,dotted] (7,2) to [out=-20,in=100] (9.25,0.3);
\draw [arrow,dotted] (7,2) to [out=-30,in=90] (8.5,0.3);
\draw [arrow,dotted] (7,2) to [out=-144,in=80] (5.5,0.3);
\draw [arrow,dotted] (7,2) to [out=-150,in=70] (4.5,0.3);
\draw [arrow,dotted] (7,2) to [out=-155,in=60] (3.5,0.3);
\draw [arrow,dotted] (7,2) to [out=-160,in=60] (1.5,0.3);
\draw [arrow,dotted] (7,2) to [out=-165,in=55] (0.75,0.3);
\draw [arrow,dotted] (7,2) to [out=-170,in=50] (0,0.3);

\draw (1.25,2) node [above] {non-overlapping maximal loops};
\draw [arrow,thick] (1,2) to (1.25,0.3);
\draw [arrow,thick] (1.75,2) to (5.25,0.3);

\draw (4.75,-1) node [below] {loops shrinked by lag};
\draw [arrow,thick] (4.25,-1) to (2.5,-0.3);
\draw [arrow,thick] (5,-1) to (7,-0.3);
\end{tikzpicture}
\caption{Factorization of an output block.}
\label{fig:loop}
\end{figure}


%
%
We prove the first claim of the lemma (Figure \ref{fig:loop} provides an
intuitive account of the constructions).
Let $v$ be an output block of the source synchronized word $w$ and let
$\r'$ be the factor of the run $\r$ aligned with $v$. As a preliminary
step, we fix a maximal set of pairwise non-overlapping maximal loops inside $\r'$,
say $\r'_1,\dots,\r'_m$. A simple counting argument shows that $m\le|Q|$ and that
there are at most $|Q|$ positions in $\r'$ that are not covered by the loops 
$\r'_1,\dots,\r'_m$. The latter positions determine some sub-blocks of $v$
of length $1$. The remaining sub-blocks of $v$ will be obtained by factorizing
the loops $\r'_1,\dots,\r'_m$, as follows. Consider any loop $\r'_j$. 
By construction, all letters consumed by $\r'_j$ occur in $v$, so they must be 
output letters.
Similarly, all letters produced by $\r'_j$ are also output letters, since otherwise,
by considering repetitions of the loop $\r'_j$, one could get different lags,
violating Lemma \ref{lem:rational-lag}. 
This means that the lag associated with the states along $\r'_j$ is constant, 
say $\ell_j$ ($\le|Q|$).
If $\r'_j$ has length at most $2|\ell_j|$, then we simply 
decompose it into $2|\ell_j|$ factors of length $1$.
Otherwise, we cover a prefix of $\r'_j$ with $|\ell_j|$ factors of length $1$,
and a suffix of $\r'_j$ with $|\ell_j|$ other factors of length $1$.
The remaining part of $\r'_j$ is covered by a last factor of length
$|\r'_j| - 2|\ell_j|$. Overall, this induces a factorization of $v$ 
into at most $|Q|$ (the sub-blocks not covered by a loop)  + $|Q|\cdot (2|Q| + 1)$ (Each $\rho'_j$ 
is decomposed into $(2\ell_j+1) \leq (2|Q|+1)$ sub-blocks). This gives 
$\cO(|Q|^2)$ sub-blocks $v_1,\dots,v_n$.
Moreover, by construction, if $|v_i|>1$, then in the corresponding factor
$\r_i$ of $\r$, all states have the same lag, say $\ell_i$, and if we
extend $\r_i$ to the left and to the right by exactly $|\ell_i|$ transitions,
we get back one of the loops $\r'_j$ (recall that each loop $\r'_j$ 
of length $>2|\ell_j|$ is decomposed into  $|\ell_j|$ blocks of length 1, then 
a block of length $|\r'_j|-|\ell_j|$, and finally, $|\ell_j|$ blocks of length 1.
Clearly, if we extend the middle block on either side by 
blocks of length $|\ell_j|$, then we get back $\r'_j$. 
 This proves the first claim of the lemma.

As for the second claim, suppose that $v_1,\dots,v_n$ is a factorization 
of an output block $v$ of $w$ satisfying the first claim. Clearly, every
sub-block $v_i$ of length $1$ is also a factor of the target synchronized
word $w'$. The interesting case is when a sub-block $v_i$ has length 
larger than $1$. In this case, by the previous claim, we know that in
the corresponding factor $\r_i$ of $\r$, all states have the same lag $\ell_i$,
and the factor $\r'_i$ of $\r$ that is obtained by expanding $\r_i$ to the 
left and to the right by $|\ell_i|$ transition is a loop.
In fact, since $\r'_i$ is a loop, we also know that all states in it
have lag $\ell_i$.
Now, to prove that $v_i$ is a factor of the target synchronized word $w'$,
it suffices to show that every two consecutive positions of $\r_i$
are mapped to consecutive positions via the relation $\om_\r$.
This follows almost by construction, since for every 
pair $(i',k')\in\om_\r$, if $i'$ occurs inside the factor $\r_i$, 
then $k'$ occurs inside the loop $\r'_i$ (recall that $\r'_i$ consumes
and produces only output symbols), and hence $k'=i' - \ell_i$.
In addition, if $i'+1$ also occurs inside $\r_i$, then clearly
$(i'+1,k'+1)\in\om_\r$.
This proves that $v_i$ is a factor of the target synchronized word $w'$,
and hence all positions in it have the same target origin.
\end{proof}

In view of the above lemma we can guess a suitable factorization of 
the output into sub-blocks that refine the output blocks, and treat 
each sub-block as if it were a single position. 
In particular, we can annotate every sub-block with a unique offset from 
a finite set of quadratic size w.r.t.~$|Q|$. The role of the offsets
will be the same as in the previous proof, where blocks had bounded length,
namely, determine some partial bijections $O_{o,o'}$, $I_{0,o'}$, and $M_o$
between positions of the input.
In addition, we annotate every sub-block with the pair consisting 
of the first and last states of the factor of the successful run 
that consumes that sub-block. We call such a pair of states a 
\emph{pseudo-transition}, as it plays the same role of a transition
associated with a single output position.
Finally, we annotate every input position $y$ with a sequence
of bounded length that represents a single transition on $y$
followed by the pseudo-transitions on the subblocks with source
origin $y$. The resulting input annotation provides an abstraction
of a successful run of $R$.

The correctness of the above annotations can be enforced by defining suitable
relations $\ipar$, $\opar$, $\nxt_{\otype,\otype'}$ for the regular resynchronizer $R'$.
We omit the tedious details concerning these relations, and only observe that, 
as before, the definition $\nxt_{\otype,\otype'}$ relies on the fact that 
$\move_\otype$ and $\move_{\otype'}$ define partial bijections between
input positions.

Finally, we turn to describing the relation $\move_\otype$ that 
maps source to target origins for $\otype$-labelled output positions.
The definition is basically the same as before, based on some
auxiliary relations $O_{o,o'}$ and $I_{0,o''}$ that implement
$\om_\r$ and $\im_\r$ at the level of input positions.
As before, we guarantee, by means of a syntactical trick, 
that $\move_\otype$ determines a partial bijection between 
input positions.
In conclusion, we get a regular resynchronizer $R'$, with input and output parameters,
that is equivalent to the rational resynchronizer $R$.



\fi



\section{Synthesis of Resynchronizers}\label{sec:synthesis}

Recall that containment between transducers depends on the adopted 
semantics. More precisely, according to the classical semantics, 
$T_1$ is contained in $T_2$ (denoted $T_1\subseteq T_2$) if all 
input-output pairs realized by $T_1$ are also realized by $T_2$; 
according to the origin semantics, $T_1$ is contained in $T_2$ 
(denoted $T_1\ocont T_2$) if all origin graphs realized by $T_1$ 
are also realized by $T_2$. 
In this section, we study the following variant of the containment problem:

\begin{labeling}{\sffamily\bfseries Question:}
\item[\sffamily\bfseries Resynchronizer synthesis problem.]
\item[\sffamily\bfseries Input:] 
      two transducers $T_1,T_2$.
\item[\sffamily\bfseries Question:] 
       does there exist some resynchronization $R$ such that $T_1 \ocont R(T_2)$.
\end{labeling}

In fact, the above problem comes in several variants, depending on the model
of transducers considered (one-way or two-way) and the class of admissible
resynchronizations $R$ (rational or bounded regular).
Moreover, for the positive instances of the above problem, we usually 
ask to compute a witnessing resynchronization $R$ from the given $T_1$ and $T_2$
(this is the reason for calling the problem a \emph{synthesis problem}).

Clearly, the synthesis problem for unrestricted 
resynchronizers is equivalent to a classical containment, that is, 
$T_1 \subseteq T_2$ if and only if $T_1 \ocont R(T_2)$ for some  
resynchronizer $R$. 
Therefore, the synthesis problem for unrestricted resynchronizers
is undecidable. Thus we will consider the synthesis problem of rational 
(resp.~bounded regular) resynchronizers for one-way (resp.~two-way) transducers.

We also recall that rational resynchronizers preserve definability of 
relations by one-way transducers~\cite{FiliotJLW16}, while bounded regular resynchronizers 
(which, by Theorem \ref{thm:rational-vs-regular}, are strictly more expressive 
than rational resynchronizers) preserve definability by two-way
transducers~\cite{bmpp18}.
For the sake of presentation, we shall first consider the synthesis of
rational resynchronizers in the functional one-way setting, that is,
for instances given by functional one-way transducers.
We show that in this setting the problem collapses again to the classical 
containment problem, which is however decidable now, that is: 
$T_1 \subseteq T_2$ if and only if $T_1 \ocont R(T_2)$ for some rational 
resynchronizer $R$. 
The decidability result can be slightly extended to some non-functional 
transducers. More precisely, we will show that synthesis of rational 
resynchronizers for finite-valued one-way transducers 
is still decidable. When moving to the relational case, however, 
the problem becomes undecidable.

The decidability status in the one-way setting could be also contrasted 
with the two-way setting. In this respect, we observe that, in the functional case,
the synthesis problem does not collapse anymore to classical containment,
as there are functional two-way transducers $T_1,T_2$ such that $T_1\subseteq T_2$, 
but for which no bounded regular resynchronizer $R$ satisfies $T_1 \ocont R(T_2)$
(an example can be found at the beginning of Section \ref{subsec:synthesis-twoway}).
We are able to prove decidability of synthesis of bounded, regular resynchronizers 
for unambiguous two-way transducers. The decidability status, however, remains 
open in the functional two-way case, as well as in the unrestricted 
(non-functional) two-way case. 

\subsection{Resynchronizing functional, one-way transducers}\label{subsec:one-way-equivalence-up}

Recall that it can be decided in \pspace whether a transducer (be it one-way or two-way) is functional \cite{bcps03},
and that the classical containment problem for functional (one-way/two-way) transducers is also in \pspace \cite{BH77}.
The following result shows that, for functional one-way transducers, classical containment and 
rational resynchronizer synthesis are inter-reducible.

\begin{restatable}{theorem}{oneWayEquivalenceUp}\label{th:one-way-equivalence-up}
Let $T_1,T_2$ be two functional one-way transducers. The following conditions are equivalent, and decidable:
\begin{enumerate}
  \item $T_1 \subseteq T_2$,
  \item $T_1 \ocont R(T_2)$ for some resynchronization $R$,
  \item $T_1 \oeq R(T_2)$ for some {\sl rational} resynchronizer $R$.
\end{enumerate}
\end{restatable}

\ifARXIV
\begin{proof}
One implication, from 2.~to~1., is trivial, 
since origin containment implies classical containment,
and since applying an arbitrary resynchronization $R$ to $T_2$ 
cannot result in having more input-output pairs (it can however
modify the origin, as well as discard some input-output pairs).
The implication from 3.~to 2.~is also trivial.

For the remaining implication, from 1.~to~3., suppose that $T_1,T_2$ are 
functional one-way transducers such that $T_1 \subseteq T_2$. 
We construct a rational resynchronizer $R$ over the disjoint union $\Sigma\uplus\Gamma$ 
of the input and output alphabets of $T_1,T_2$, using a variant of the 
direct product of $T_1$ and $T_2$. 
More precisely, let $T_1=(Q_1,q_1,\Delta_1,F_1)$,
$T_2=(Q_2,q_2,\Delta_2,F_2)$, and $R=(Q,q,\Delta,F)$, where
$Q = Q_1\times Q_2$, $q=(q_1,q_2)$, $F=F_1\times F_2$, $\Delta$ contains all transitions
of the form $(s_1,s_2) \act{a w_2 \,\mid\, a w_1} (t_1,t_2)$, with $s_i \act{a \,\mid\, w_i} t_i$ 
in $\Delta_i$ for both $i=1$ and $i=2$.
Intuitively, the transducer $R$ simulates a run of $T_1$ and a run of $T_2$ in parallel, by
repeatedly consuming an input symbol $a$ and the corresponding output $w_2$ produced by $T_2$, 
and producing the same input symbol $a$ and the corresponding output $w_1$ of $T_1$.
Since $T_1$ and $T_2$ are functional and classically contained one in the other, 
we have that $R$ maps strings over $\Sigma\uplus\Gamma$ to strings over $\Sigma\uplus\Gamma$ 
while preserving the projections on the input and on the output alphabets.
This means that $R$ is indeed a resynchronizer. 
Finally, $T_1$ is clearly origin equivalent to $R(T_2)$.
\end{proof}


\fi

A natural question arises: can a characterization similar to 
Theorem \ref{th:one-way-equivalence-up} 
be obtained for transducers that compute arbitrary relations,
rather than just functions? The example below provides a negative 
answer to this question. Later in Section \ref{subsec:synthesis-oneway-relational},
we will see that synthesis of rational resynchronizers 
for unrestricted one-way transducers is an undecidable problem.

\begin{example}
Consider a one-way transducer $T_1$ that checks that the input is from
$(aa)^*$ and produces a single 
output letter $b$ for each consumed input letter $a$,
and another transducer $T_2$ that works in two phases: during 
the first phase, it produces two $b$'s for each consumed $a$, 
and 
\mywraptext{2.5cm}{%
during the second phase consumes 
the remaining part of the input without producing any output. 
The origin graphs of
$T_1$ and $T_2$ are shown to the right. 
We have $T_1 \subseteq T_2$, but  $T_1 \nocont T_2$. 
The only resynchronization $R$ that satisfies $T_1 \ocont R(T_2)$ 
must map synchronized words from $(ab)^*$ to $(abb)^*(a)^*$,
while preserving the number of $a$'s and $b$'s. 
Such a transformation cannot be defined by any rational resynchronizer, 
nor by a bounded regular resynchronizer.
}%
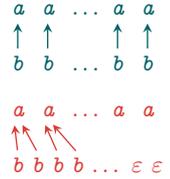
\mywrapfig{2.5cm}{%
\vspace{-2mm}%
\begin{tikzpicture}[xscale=0.75,yscale=0.6,baseline]
\begin{scope}[xscale=1.75,color=nicecyan]
\draw (0,1.25) node [hidden] (I1) {\small $\texttt{a}$};
\draw (0.3,1.25) node [hidden] (I2) {\small $\texttt{a}$};
\draw (0.7,1.25) node [hidden] (I3) {\small $\dots$};
\draw (1,1.25) node [hidden] (I4) {\small $\texttt{a}$};
\draw (1.3,1.25) node [hidden] (I5) {\small $\texttt{a}$};

\draw (0,0) node [hidden] (O1) {\small $\texttt{b}$};
\draw (0.3,0) node [hidden] (O2) {\small $\texttt{b}$};
\draw (0.7,0) node [hidden] (O3) {\small $\dots$};
\draw (1,0) node [hidden] (O4) {\small $\texttt{b}$};
\draw (1.3,0) node [hidden] (O5) {\small $\texttt{b}$};

\draw (O1.north) edge [arrow] (I1.south);
\draw (O2.north) edge [arrow] (I2.south);
\draw (O4.north) edge [arrow] (I4.south);
\draw (O5.north) edge [arrow] (I5.south);
\end{scope}
\begin{scope}[xscale=1.75,color=nicered,yshift=-2.25cm]
\draw (0,1.25) node [hidden] (I1) {\small $\texttt{a}$};
\draw (0.3,1.25) node [hidden] (I2) {\small $\texttt{a}$};
\draw (0.7,1.25) node [hidden] (I3) {\small $\dots$};
\draw (1,1.25) node [hidden] (I4) {\small $\texttt{a}$};
\draw (1.3,1.25) node [hidden] (I5) {\small $\texttt{a}$};

\draw (0,0) node [hidden] (O1) {\small $\texttt{b}$};
\draw (0.2,0) node [hidden] (O2) {\small $\texttt{b}$};
\draw (0.4,0) node [hidden] (O3) {\small $\texttt{b}$};
\draw (0.6,0) node [hidden] (O4) {\small $\texttt{b}$};
\draw (0.9,0) node [hidden] (O5) {\small $\dots$};
\draw (1.2,0) node [hidden] (O6) {\small $\emptystr$};
\draw (1.4,0) node [hidden] (O7) {\small $\emptystr$};

\draw (O1.north) edge [arrow] ([xshift=-1pt] I1.south);
\draw (O2.north) edge [arrow] ([xshift=1pt] I1.south);
\draw (O3.north) edge [arrow] ([xshift=-1pt] I2.south);
\draw (O4.north) edge [arrow] ([xshift=1pt] I2.south);
\end{scope}
\end{tikzpicture}%
\hspace{-7mm}
%
\vspace{-3mm}%
}	
\end{example}

There is however an intermediate case, between the functional and the full 
relational case, for which a generalization of Theorem \ref{th:one-way-equivalence-up} 
is possible. This is the case of \emph{finite-valued} one-way transducers, 
that is, transducers that realize finite unions of partial functions. 
The generalization exploits a result from \cite{FiliotJLW16}, 
stated just below, that concerns synthesis of bounded-delay resynchronizers. 
Formally, given two origin graphs $\synch$ and $\synch'$ with the same input and 
output projections, 
and given an input position $y$, we denote by $\delay_{\synch,\synch'}(y)$
the difference between the largest $x\in\dom(\Out(\synch))$ such that 
$\Orig(\synch)(x)=y$ and the largest $x'\in\dom(\Out(\synch'))$ such that 
$\Orig(\synch')(x')=y$.
Given $d\in\bbN$, we define the \emph{$d$-delay resynchronizer} as the 
resynchronization that contains all pairs $(\synch,\synch')$ with the
same input and output projections and such that 
$\delay_{\synch,\synch'}(y) \in [-d,+d]$ for all input positions $y$.
It is easy to see that the $d$-delay resynchronizer is a special case of
a rational resynchronizer.

\begin{theorem}[Theorem~13 in \cite{FiliotJLW16}]\label{th:k-val equivalence}
Let $T_1,T_2$ be one-way transducers, where $T_2$ is $k$-ambiguous.%
\footnote{A transducer is $k$-ambiguous if each input admits at most $k$ successful runs.}
One can compute a $d$-delay resynchronizer $R_d$, for some $d\in\bbN$, 
such that $T_1 \subseteq T_2$ implies $T_1 \ocont R_d(T_2)$.
\end{theorem}

As a corollary we can generalize Theorem \ref{th:one-way-equivalence-up}
to $k$-valued one-way transducers, with the only difference that the witnessing 
rational resynchronizer now satisfies $T_1 \ocont R(T_2)$ rather than $T_1 \oeq R(T_2)$. We also recall that classical containment remains decidable
for $k$-valued one-way transducers, thanks to the fact that these can be
effectively transformed to finite unions of functional transducers~\cite{Weber96}:

\begin{corollary}\label{cor:k-val equivalence}
Let $T_1,T_2$ be $k$-valued one-way transducers. The following conditions are equivalent, and decidable:
\begin{enumerate}
  \item $T_1 \subseteq T_2$,
  \item $T_1 \ocont R(T_2)$ for some resynchronization $R$,
  \item $T_1 \ocont R(T_2)$ for some {\sl rational} resynchronizer $R$.
\end{enumerate}
\end{corollary}

\begin{proof}
We prove the only interesting implication from 1.~to~3.
Suppose that $T_1,T_2$ are $k$-valued one-way transducers such that $T_1 \subseteq T_2$. 
Using the decomposition theorem from \cite{Weber96}, 
we can construct a $k$-ambiguous
one-way transducer $T'_2$ that is classically equivalent to $T_2$ 
and such that $T'_2\ocont T_2$. 
Since $T_1\subseteq T'_2$, 
by Theorem \ref{th:k-val equivalence} we can compute a $d$-delay (in particular, rational)
resynchronizer $R_d$ such that $T_1 \ocont R_d(T'_2)$. 
Finally, since $T'_2\ocont T_2$, $T_1\ocont R_d(T'_2)$, and $R_d(T'_2) \ocont R_d(T_2)$,
we get $T_1\ocont R_d(T_2)$. 
\end{proof}

\subsection{Resynchronizing arbitrary one-way transducers}\label{subsec:synthesis-oneway-relational}

In the previous section we saw how to synthesize a rational
resynchronizer for functional, or even finite-valued, one-way 
transducers. One may ask if finite-valuedness is necessary.
 We  already know
that classical containment $T_1 \subseteq T_2$ is
undecidable~\cite{FischerR68,gri68} for arbitrary 
one-way transducers, whereas origin-containment 
$T_1 \ocont T_2$ is decidable~\cite{bmpp18}. Synthesis of a
rational resynchronizer $R$ such that $T_1 \ocont R(T_2)$ is a
question that lies between the two questions above. 
We  show in this section that in the case of {\sl real-time} transducers with 
unary output alphabet, the latter question is equivalent to 
language-boundedness of one-counter automata, a problem that we 
define below.

A transducer is said to be \emph{real-time} if it produces 
bounded outputs for each consumed input symbol.
A \emph{one-counter automaton} (OCA) is a non-deterministic pushdown 
automaton with a single stack symbol, besides the bottom stack symbol. 
In the definition of the language-boundedness problem, we assume that 
the OCA recognizes a universal language; this assumption is used in the 
reduction to the synthesis problem.

\begin{labeling}{\sffamily\bfseries Question:}
\item[\sffamily\bfseries Language-boundedness of OCA.]
\item[\sffamily\bfseries Input:] 
      An OCA $A$ over alphabet $\Omega$ that recognizes 
      the {\sl universal} language $L(A)=\Omega^*$. 
\item[\sffamily\bfseries Question:] 
      Does there exist some bound $k$ such that every word over 
      $\Omega$ can be accepted by $A$ with a run where the counter 
      never exceeds $k$?
\end{labeling}

\noindent Our reductions between language-boundedness of OCA and 
synthesis of rational resynchronizers rely on the following
result from \cite{FiliotJLW16}, that implies that bounded-delay 
resynchronizers are enough for synthesizing resynchronizers of real-time
transducers:

\begin{theorem}[Theorem~11 in \cite{FiliotJLW16}]\label{thm:bounded-delay}
Let $T_1,T_2$ be real-time, one-way transducers and 
$R$ a rational resynchronizer such that $T_1 \ocont R(T_2)$. 
One can compute a $d$-delay resynchronizer 
$R_d$ such that $T_1 \ocont R_d(T_2)$.
\end{theorem}


\begin{restatable}{proposition}{synthesisUnaryOneWay}\label{prop:synthesis-unary-oneway}
Synthesis of rational resynchronizers 
for real-time one-way transducers with unary output alphabet 
and language-boundedness of OCA are inter-reducible problems.
Moreover, in the reductions, one can assume that the left
hand-side transducer is functional.
\end{restatable}

\ifARXIV
\begin{proof}
We first prove the reduction from synthesis of rational resynchronizers 
to language-boundedness of OCA, and then prove the reduction in 
the opposite direction.

\paragraph*{From synthesis to language-boundedness.}
Let $T_1,T_2$ be real-time, one-way transducers with unary output
alphabet.  We suppose in addition that $T_1$ is trimmed.
We construct an OCA $A$ that reads encodings 
of successful runs of $T_1$.
If the input is not a successful run of $T_1$, 
then, as soon as an error is detected, $A$ resets its counter
and accepts any continuation of the input.
In particular, thanks to this behaviour and to $T_1$ being trimmed, badly-formed encodings 
of runs will not cause the counter of $A$ to be unbounded.

Consider now an input for $A$ that is a correct 
encoding of a successful run of $T_1$, say $\r_1$. 
In this case, $A$ guesses and simulates a successful run $\r_2$ of $T_2$ 
having the same input as $\r_1$. 
The counter of $A$ is used as expected: it is incremented according 
to the outputs produced using the transitions of $\r_1$, and decremented according to the outputs produced using the transitions
of $\r_2$, or vice versa when one needs to represent a negative value
(recall that OCA work with counter over natural numbers). 
The detail regarding which among $T_1, T_2$ is ``leading'', resulting 
 in the non-negative counter value can be stored in the finite control 
 of the OCA.

Intuitively, a configuration of $A$ determines how ahead or behind 
is the partial output produced by the encoded run of $T_1$ compared to 
the partial output produced by the simulated run $T_2$.
The OCA $A$ accepts with empty counter. 
Note that this construction is close to the direct product of 
$T_1$ and $T_2$, the main difference being the treatment of the
badly formed encodings and the role played by the counter.

Let us now prove that the OCA $A$ is language-bounded if and only
if $T_1 \ocont R(T_2)$ for some rational resynchronizer $R$.

Suppose first that the OCA $A$ is language-bounded, namely, that 
there is some $k\in\Nat$ such that every word is accepted by $A$ 
with a counter that never exceeds $k$. 
We can think of the successful runs of $A$ that maintain the counter
between $0$ and $k$ as runs of a $k$-delay resynchronizer $R$. 
More precisely, we can define a letter-to-letter resynchronizer $R$,
the states of which are  the configurations of $A$ with the value of the
counter inside $\{0,\dots,k\}$.
On consuming an input letter, $R$ produces the same input letter; 
on consuming a sequence of $j$ output letters, depending 
on the simulated transition of $A$, $R$ produces an output 
of length $j+h$ if the counter is incremented by $h$. Likewise,
 if the simulated transition of $A$ decrements the counter by $h$, then on reading a sequence of $j$ 
 output symbols, $R$ produces an output of length $j-h$. 
The run of $R$ is successful if an only if the simulated run of $A$ is so. 
The fact that $A$ accepts every word with a counter that never 
exceeds $k$, immediately implies that $T_1 \ocont R(T_2)$.  

Conversely, suppose that $T_1 \ocont R(T_2)$ for some
rational resynchronizer $R$. By Theorem \ref{thm:bounded-delay}, 
we can assume without loss of generality that $R$ is a $k$-delay resynchronizer,
for some $k$ (that can be even computed from $T_1$, $T_2$, and $R$, 
but this is immaterial here). 
From this it is easy to see that $A$ is language-bounded, and precisely, 
that $A$ accepts every word with a counter that never exceeds $k$,
as when reading a run $\rho$ of $T_1$, it can guess 
a run $\rho'$ of $T_2$ such that $R(\rho') = \rho$.

\paragraph*{From language-boundedness to synthesis.}
Let $A$ be an OCA. We construct two real-time, one-way transducers 
$T_1,T_2$ that have the same input alphabet as $A$, say $\S$, and 
a singleton output alphabet, say $\G=\{c\}$.  
The transducer $T_1$ reads any word $a_1 \dots a_n \in \S^*$ 
and outputs one letter $c$ for each consumed input symbol. 
In particular, the synchronization language of $T_1$ is 
$\{a_1 c \dots a_n c : a_i \in \S, n \ge 0\}$. 
Note that $T_1$ is real-time and functional.
The transducer $T_2$ does the following: 
upon reading  $a_1 \dots a_n$, it guesses 
a successful run of the OCA $A$. 
Whenever the counter is incremented along the guessed run of $A$,
$T_2$ outputs $cc$; whenever the counter is decremented, $T_2$ outputs
$\emptystr$; whenever the counter is unchanged, $T_2$ outputs $c$.
Note that $T_2$ is also real-time, but not necessarily functional.

Let us now prove that $A$ is language-bounded if and only
if $T_1 \ocont R(T_2)$ for some rational resynchronizer $R$.

Suppose first that $A$ is language-bounded, with bound $k$. 
We obtain from this a  $k$-delay resynchronizer
$R$ that reads a synchronized word 
$a_1 c^{i_1} \dots a_n c^{i_n}$ of $T_2$, where $i_j\in\set{0,1,2}$
for all $j$.
The resynchronizer $R$ simulates a counter taking values in $[-k,k]$, 
and outputs $a_1 c \dots a_n c$, accepting if and only if 
the counter is $0$. Each time an $a_ic^2a_j$ is encountered, it corresponds 
to an increment in the OCA; then $R$ outputs $a_ic$, and 
the simulated counter decreases by 1 in $R$; likewise, 
each time an $a_ica_j$ is encountered, $R$ outputs $a_ic$ with 
no change in the simulated counter value, and finally, 
when two consecutive input symbols $a_ia_j$ are read by $R$, 
$R$ outputs $a_ic$ and 
 the simulated counter value increases by 1. 
 Since the counter value is bounded by $k$ in the OCA, the simulated counter 
 in $R$ is within $[-k,k]$. 
Clearly, $T_1 \ocont R(T_2)$. 

Conversely, suppose that $T_1 \ocont R(T_2)$ for some rational resynchronizer $R$.
We argue as before, using Theorem \ref{thm:bounded-delay}:
we assume without loss of generality that $R$ is a $k$-delay resynchronizer,
for some $k$, and derive from this that $A$ is language-bounded.
\end{proof}



\fi

The status of the problem of language-boundedness of OCA was open, 
to the best of our knowledge. Piotr Hofman communicated to us the
following unpublished result, which can be obtained by a reduction
from the undecidable boundedness problem for Minsky machines%
%
:

\begin{restatable}[\cite{hofman19}]{theorem}{undecidableOCABoundedness}\label{th:ph}
The language-boundedness problem for OCA is undecidable.
\end{restatable}

\ifARXIV
\begin{proof}
 The reduction is from the
  boundedness problem for multi-counter (Minksy) machines. Such a
  machine $M$ can increment, decrement and test for zero. The question is
  whether there exists some bound $k$ such that all computations of
  $M$ (not necessarily accepting) from the initial configuration with
  all counters zero, have all counters stay below $k$. One can assume
  w.l.o.g.~that if $M$ is not bounded then for every $k$ there is some
  initial run of $M$ where \emph{all} counters exceed $k$.

  The OCA $A$ reads sequences of transitions of $M$. At the beginning, 
  $A$ guesses a counter index $j$ of $M$ and starts simulating the
  sequence of transitions on counter $j$. If the sequence of
  transitions is incorrect because of counter $j$, the OCA accepts and
  stops after emptying counter $j$. Note that there are two types of
  error: either the counter is zero but should be decremented, or the
  counter is tested for zero, but is not zero. Both kinds of error can
  be checked by the OCA. Otherwise, if the simulation goes through for
  counter $j$, then the OCA accepts with empty counter at the end.

  Assume that $M$ is bounded, with bound $k$. If a sequence $\r$ of
  transitions is a run of $M$, then all simulations on any counter
  will be bounded by $k$. If $\r$ is not a run, then there is a first
  position of $\r$ where an error occurs, for instance because of
  counter $j$. Then the run of $A$ simulating counter $j$ will accept
  $\r$ within bound $k$.

  If $M$ is unbounded then for every $k$ there is a run $\r$ where \emph{all}
  counters exceed $k$. In this case all runs of $A$ on $\r$ exceed
  $k$, so $A$ is not language-bounded.
\end{proof}


%
\fi

\begin{corollary}
Synthesis of rational resynchronizers for (real-time) one-way 
transducers is undecidable,
and this holds even when the left hand-side transducer is functional.
\end{corollary}




\subsection{Resynchronizing unambiguous, two-way transducers}\label{subsec:synthesis-twoway}
	
We now focus on the resynchronizer synthesis problem for two-way transducers. 
Here the appropriate class of resynchronizations is that of regular resynchronizers,
since, differently from rational resynchronizer, they can handle origin graphs induced 
by two-way transducers.
The situation is more delicate, as the synthesis problem does not reduce anymore
to classical containment. As an example, consider the transducer $T_1$ that
consumes an input of the form $a^*$ from left to right, while copying the letters to the output,
and a two-way transducer $T_2$ that realizes the same function but while consuming
the input in reverse. We have that $T_1\subseteq T_2$, but there is no resynchronizer
$R$ that satisfies $T_1 \ocont R(T_2)$ and that is bounded and regular at the same time.
As we will see, extending Theorem \ref{th:one-way-equivalence-up} to two-way 
transducers is possible if we move beyond the class of regular 
resynchronizers and consider bounded resynchronizers defined by Parikh automata. 
The existence of bounded regular resynchronizers between functional two-way transducers
can thus be seen as a strengthening of the classical containment relation.
Unfortunately, we are only able to solve the synthesis problem of bounded regular 
resynchronizers for {\sl unambiguous} two-way transducers, so the problem remains
open for functional two-way transducers.

First we introduce resynchronizers definable by Parikh automata. 
Formally, a \emph{Parikh automaton} is a finite automaton
$A=(\Sigma,Q,I,E,F,Z,S)$  
equipped with a function $Z:E\rightarrow \bbZ^k$ that associates vectors of integers 
to transitions and a semi-linear set $S \subseteq \bbZ^k$.
A successful run of $A$ is a run starting in $I$, ending in $F$ and
such as the sum of the weights of its transitions belongs to $S$. 
%
We say that $A$ is \emph{unambiguous} if the underlying finite
automaton is. 
In this case, we can associate with each input $u$ the vector $A(u) \in \bbZ^k$ 
associated with the unique accepting run of the underlying automaton of
$A$ on $u$, if this exists, 
otherwise $A(u)$ is undefined. 
By taking products, one can easily prove that unambiguous Parikh automata 
are closed under pointwise sum and difference, that is, given $A_1$ and $A_2$, 
there are $A_+$ and $A_-$ such that $A_+(u)=A_1(u)+A_2(u)$ and $A_-(u)=A_1(u)-A_2(u)$ 
for all possible inputs $u$.
Hereafter, we will only consider languages recognized by {\sl unambiguous} Parikh automata with the
trivial semilinear set $S=\{0^k\}$.

By a slight abuse of terminology, we call \emph{Parikh resynchronizer} 
any resynchronizer with parameters whose relations $\move_\otype$ and 
$\nxt_{\otype,\otype'}$ are recognizable by unambiguous Parikh automata,
and $\ipar$ and $\opar$ are regular.
We naturally inherit from regular resynchronizers the notion of boundedness. 
Moreover, we introduce another technical notion, that will be helpful later.
Given a resynchronizer $R$, we define its \emph{target set} as the set of all pairs $(u,z)$
where $u$ is an input, $z$ is a position in it, and $(w,y,z)\in\move_\otype$ for 
some annotation $w$ of $u$ with input parameters, some input position $y$, and some output 
type $\otype$.
Similarly, we define the \emph{target set} of a two-way transducer $T$ as the set of all 
pairs $(u,z)$, where $u=\In(\synch)$ and $z\in\Orig(\synch)(x)$ for some $x\in\dom(\Out(\synch))$
and some origin graph $\synch$ realized by $T$.

\begin{restatable}{theorem}{twoWayEquivalenceUp}\label{th:two-way-equivalence-up}
Let $T_1,T_2$ be two unambiguous two-way transducers. The following conditions are equivalent:
\begin{enumerate}
	\item $T_1 \subseteq T_2$,
	\item $T_1 \ocont R(T_2)$ for some resynchronization $R$,
	\item $T_1 \oeq R(T_2)$ for some $1$-bounded 
	      Parikh resynchronizer $R$ 
	      whose target set coincides with that of $T_1$
	      and where, each relation $\nxt_{\otype,\otype'}$ is regular
	      if $\move_{\otype}$ and $\move_{\otype'}$ are regular.
\end{enumerate}
\end{restatable}

\ifARXIV
\begin{proof}
The implications from 2.~to~1. and from 3.~to 2.~are as in the proof
of Theorem \ref{th:one-way-equivalence-up}. The only interesting implication
is from 1.~to~3, where we suppose that $T_1 \subseteq T_2$ and we aim at 
constructing a $1$-bounded 
Parikh resynchronizer $R$ such that $T_1 \oeq R(T_2)$,
and with the same target set as $T_1$. 
The proof exploits some constructions based on crossing sequences, 
which are classically used to translate two-way 
automata to equivalent one-way automata \cite{Shepherdson59}, 
as well as to reduce containment of functional two-way 
transducers to emptiness of languages recognized by Parikh 
automata \cite{STACS19}. We briefly recall the key notions here, by adapting
them in a way that is convenient for the presentation (notably, considering transitions instead of states).

A \emph{crossing sequence} of a two-way automaton or a functional two-way transducer is a tuple 
$\bar t=(t_1,\dots,t_n)$ of transitions such that the source states of $t_1,t_3,\dots$ are 
right-reading and the source states of $t_2,t_4,\dots$ are left-reading.
The tuple is meant to describe the transitions along a successful run that 
depart from configurations at a certain position $y$. 
Formally, given a run $\rho$, the crossing sequence of $\rho$ at input position $y$, denoted $\rho[y]$,
consists of the quadruples $(q,a,v,q')$ such that $(q,y) \act{a \:|\: v} (q',y')$ is a transition
of $\rho$, where the occurrence order on transitions induces a corresponding order on the
quadruples of the crossing sequence.
Without loss of generality, for two-way automata, as well as for functional two-way transducers, 
one can restrict to successful runs that never visit the same state twice at the same position. 
Accordingly, we can assume that the length of a crossing sequence never exceeds the total number 
of states of the device. 
Moreover, when the two-way automaton or transducer is unambiguous, the crossing sequences are 
uniquely determined by the input and the specific position in it. More precisely, there are
regular languages $L_{\bar t}$, one for each possible crossing sequence, that contains precisely those
inputs $u$ with a specific position $y$ marked on it (for short, we denote such words by $\angled{u,y}$),
such that the crossing sequence at $y$ of the unique successful run on $u$ is precisely $\bar t$.

\smallskip
We now turn to the main proof, which is divided into several steps.

\paragraph*{Encoding output positions.}
We begin by describing a natural encoding of arbitrary output positions by means of their origins.
Of course, the encoding depends on the given input, denoted $u$, and on the transducer we consider, 
either $T_1$ or $T_2$, which here is generically denoted by $T$.
Now, let $\rho$ be the unique successful run of $T$ on $u$, and 
let $\synch$ be the induced origin graph.
To simplify the notations, hereafter we tacitly assume that $T$
produces at most one letter at each transition --- the 
assumption is without loss of generality, since long outputs 
originating at the same input position can be produced incrementally
by exploiting two-way head motions. 
Let $n$ be the number of states of $T$. Since $T$ is unambiguous, $\synch$ contains at most $n$ 
output positions with the same origin (otherwise, the same configuration would be visited at least 
twice along the successful run $\rho$, which could then be used to contradict the assumption of 
unambiguity). 
This means that every position $x$ in $\Out(\synch)$ can be encoded by its origin 
$y_x=\Orig(\synch)(x)$ together with a suitable index $i_x\in\{1,\dots,n\}$, 
describing the number of output positions $x' \le x$ with the same origin $y_x$ as $x$.
Moreover, we recall that $y_x$ can be represented as an annotated input of the form
$\angled{u,y_x}$.

\paragraph*{Decoding by Parikh automata.}
We now show that there are Parikh automata that compute the inverse of 
the encoding $x \mapsto (y_x,i_x)$ described above. More precisely, there are unambiguous 
Parikh automata $A_1,\dots,A_n$ such that each $A_i$ receives as input a word $\angled{u,y}$
having a special position marked on it, and outputs the unique output position $x$ such that 
$(y,i)=(y_x,i_x)$, if this exists, otherwise the output is undefined.
Each automaton $A_i$ can be constructed from $T$ and $i$ by unambiguously guessing 
the crossing sequences of the unique run of $T$ on $u$, and by counting the number of 
output symbols emitted until a \emph{productive} transition at the marked position 
$y$ is executed for the $i$-th time --- a productive transition is a transition 
that produces non-empty output.

\paragraph*{Redirecting origins.}
We now apply the constructions outlined above in order to obtain the desired Parikh resynchronizer $R$
from $T_1$ and $T_2$.
Let $u$ be some input and $\synch_1,\synch_2$ be the origin graphs induced by the 
unique successful runs of $T_1,T_2$ on $u$.
Since $T_1\subseteq T_2$, we can further let 
$v=\Out(\synch_1)=\Out(\synch_2)$.
Consider any output position $x\in\dom(v)$. 
According to $\synch_2$, $x$ is encoded by an input position $y_x$ and an index $i_x\in\{1,\dots,n_2\}$,
where $n_2$ is the number of states of $T_2$. In a similar way, according to $\synch_1$, the same position 
$x$ is encoded by some input position $z_x$ and an index $j_x\in\{1,\dots,n_1\}$, where $n_1$ is the number
of states of $T_1$. 
Moreover, based on the previous constructions, there are unambiguous Parikh automata 
$A_{2,i}$ and $A_{1,j}$ such that
\begin{itemize}
	\item $A_{2,i}(\angled{u,y})=x$ if and only $(y,i) = (y_x,i_x)$,
	\item $A_{1,j}(\angled{u,z})=x$ if and only $(z,j) = (z_x,j_x)$.
\end{itemize}
Since unambiguous Parikh automata are closed under pointwise difference, there is a 
unambiguous Parikh automaton $A_{i,j}$ that recognizes precisely the language of 
annotated words $\angled{u,y,z}$ such that
\begin{align*}
	A_{2,i}(\angled{u,y}) - A_{1,j}(\angled{u,z}) = 0
	\tag{$\star$}
\end{align*}
Note that the above language defines 
a partial bijection between pairs of positions $y,z$ in the input $u$ in 
such a way that $y$ and $z$ are the origins of the same output position $x$ 
according to the unique origin graphs $\synch_1,\synch_2$ of $T_1,T_2$ 
such that $\In(\synch_1)=\In(\synch_2)=u$.
This property can be used to define the component $\move_\otype$ 
of the desired resynchronizer $R$, by simply letting
\[
  \move_\otype = \{ (u,y,z) ~\mid~ A_{i,j}(\angled{u,y,z}) = 0 \}
\]
where $\otype=(a,i,j) \in \G\times\{1,\dots,n_2\}\times\{1,\dots,n_1\}$.

For the correctness of the above definition we rely on guessing the 
correct pairs of indices $(i,j)$ as annotations of output positions.
More precisely, we have that:
\begin{itemize}
  \item for every output position $x$ with source origin 
        $y=\Orig(\synch_2)(x)$ and with label $\otype=(a,i_x,j)$,
        there is at most one input position $z$ such that $(u,y,z)\in\move_\otype$; 
        in addition, if we also have $j=j_x$, then $z=\Orig(\synch_1)(x)$
        is the target origin of $x$; symmetrically,
  \item for every output position $x$ with target origin 
        $z=\Orig(\synch_1)(x)$ and with label $\otype=(a,i,j_x)$,
        there is at most one input position $y$ such that $(u,y,z)\in\move_\otype$; 
        in addition, if we also have $i=i_x$, then $y=\Orig(\synch_2)(x)$
        is the source origin of $x$.
\end{itemize}
Based on the above properties, we need to guess suitable output parameters 
that associate with each position $x$, a correct pair $(i_x,j_x)$. 
We explain below how this is done using the components $\opar$ and 
$\nxt_{\otype,\otype'}$ of the resynchronizer. 

\paragraph*{Constraining output parameters.}
We first focus on the indices $j_x$ related to $T_1$; 
we will later explain how to adapt the constructions to check the indices $i_x$ 
related to $T_2$. As usual, we fix an input $u$ and the unique successful run $\rho_1$
of $T_1$ on $u$.
The idea is that each index $j_x$ corresponds to a certain element of the crossing 
sequence of $\rho_1$ at the target origin $z_x$, and knowing the correct index for 
$x$ determines the correct index for the next output position $x+1$. 
Based on this, correctness can be verified inductively using the guessed 
crossing sequences and the relation $\nxt_{\otype,\otype'}$ of the resynchronizer, 
as follows.
For the base case, we check that the first output position is correctly 
annotated with the index $j=1$: this is readily done by a regular language $\opar$.

For the inductive step, we consider an output position $x$ and assume 
that it is correctly annotated with $j=j_x$.
Let $j'$ be the annotation of the next position $x+1$. 
To check that $j'$ is also correct, we consider pairs of productive transitions
in the crossing sequences associated with the target origins of $x$ and $x+1$, 
and verify that they are connected by a non-productive run.
More precisely, let $z$ and $z'$ be the target origins of $x$ and $x+1$,
respectively, and let $\bar t_z$ and $\bar t_{z'}$ be the crossing sequences 
of $\rho_1$ at those positions.
We have that $j'=j_{x+1}$ if and only if the $j$-th productive transition 
of $\bar t_z$ and the $j'$-th productive transition of $\bar t_{z'}$ 
are connected by a factor of the run that consists only of non-productive transitions.
The latter property can be translated to a regular property $\nxt_{\otype,\otype'}$
concerning the input annotated with two specific positions, $z$ and $z'$, 
assuming that $\otype=(a,i,j)$ and $\otype=(a',i',j')$ are the letters of 
the output positions $x$ and $x+1$.

It now remains to check the correctness of the output annotations
w.r.t.~the indices $i$ for the second transducer $T_2$. We follow
a principle similar to the one described above for $T_1$. The only
difference is that now, in the inductive step, we have work with the 
source origins $y$ and $y'$ of consecutive output positions $x$ and $x+1$. 
The additional difficulty is that, by definition, the relation $\nxt_{\otype,\otype'}$ 
can only refer to target origins. We overcome this problem by exploiting the partial
bijection between target and source origins, as defined by the relations 
$\move_\otype$ and $\move_{\otype'}$. Formally, we first define a relation 
$\nxt_{\otype,\otype'}^{\text{source}}$ as before, that constrain the indices 
$i$ and $i'$ associated with two consecutive output positions $x$ and $x+1$ 
labeled by $\otype$ and $\otype'$, respectively. We do this as if 
$\nxt_{\otype,\otype'}^{\text{source}}$ were able to speak about source origins. 
We then intersect the following relation with the previously 
defined relation $\nxt_{\otype,\otype'}$: 
\[
  \big\{ (u,z,z') ~\mid~ 
     \exists y,y'~ (u,y,y')\in\nxt_{\otype,\otype'}^{\text{source}}, ~
                   (u,y,z)\in\move_\otype, ~
                   (u,y',z')\in\move_{\otype'} \big\}.
\]
Since in the inductive step we assume that $x$ is correctly annotated with the
pair $(i,j)$ and $x+1$ is annotated with $(i',j')$, where $j'=j_x$ is correct 
by the previous arguments, there are unique $y,y'$ that satisfy 
$(u,y,z)\in\move_\otype$ and $(u,y',z')\in\move_\otype$ 
in the above definition, and these must be the source origins of $x$ and $x+1$.
This means that the above relation, which is definable by a unambiguous Parikh 
automaton, correctly verifies the correctness of the index $i'$ associated with $x+1$.

We conclude by observing a few properties of the defined Parikh resynchronizer $R$.
As already explained, the relation $\move_\otype$ defines a bijection between
pairs of input positions, so $R$ is a $1$-bounded 
Parikh resynchronizer.
As concerns its target set, that is the set of pairs $(u,z)$ 
such that $(u,y,z)\in\move_\otype$ for some $z\in\dom(u)$ and some 
$\otype\in\G\times\{1,\dots,n_2\}\times\{1,\dots,n_1\}$,
it coincides by construction with the target set of $T_1$.
Finally, since the relation $\nxt_{\otype,\otype'}$ is defined
by conjoining a regular property with the properties defined by
the relations $\move_\otype$ and $\move_{\otype'}$, we have that
$\nxt_{\otype,\otype'}$ is regular if $\move_\otype$ and $\move_{\otype'}$
are regular.
\end{proof}


\fi

We now explain how to exploit the above characterization to decide 
bounded regular resynchronizer synthesis problem. We provide the following 
characterization, whose proof follows from the previous theorem:

\begin{restatable}{theorem}{twoWayEquivalenceUpBis}\label{th:two-way-equivalence-up-bis}
Let $T_1,T_2$ be two unambiguous two-way transducers such that $T_1\subseteq T_2$,
and let $\hat R$ be the bounded 
Parikh resynchronizer obtained 
from Theorem \ref{th:two-way-equivalence-up}. 
The following conditions are equivalent:
\begin{enumerate}
	\item $\hat R$ is a regular resynchronizer,
	\item $T_1 \ocont R(T_2)$ for some bounded regular resynchronizer $R$,
	\item $T_1 \ocont R(T_2)$ for some $1$-bounded regular resynchronizer $R$,
	\item $T_1 \oeq R(T_2)$ for some $1$-bounded 
	      regular resynchronizer $R$ 
	      with the same target set as $T_1$.
\end{enumerate}
\end{restatable}

\ifARXIV
\begin{proof}
We prove the following implications in the order: 
1.~$\rightarrow$ 2.~$\rightarrow$ 3.~$\rightarrow$ 4.~$\rightarrow$ 1.

\paragraph*{From 1.~to~2.}
This is trivial since $\hat R$ is bounded 
and satisfies $T_1 \oeq \hat R(T_2)$, and hence $T_1\ocont \hat R(T_2)$.

\paragraph*{From 2.~to~3.}
Let $R$ be a $k$-bounded regular resynchronizer.
The goal is to construct an equivalent $1$-bounded regular resynchronizer $R'$
(note that this part of the proof does not depend on $T_1$ and $T_2$).
For this, we introduce a parameter $i_x\in\{1,\dots,k\}$ 
associated with each output position $x$, and require that 
for all output positions $x,x'$ having the same label $\otype$, 
and for all input positions $y,z$ 
such that $(u,y,z)\in \move_\otype$, if $i_x=i_{x'}$, then $x=x'$.
The existence of such a mapping $x\mapsto i_x$ follows easily 
from the assumption that $R$ is $k$-bounded. 
The relation $\move'_{(\otype,i)}$ of the new resynchronizer $R'$
redirects origins of output positions based on their annotations 
$(\otype,i)\in\G\times\{1,\dots,k\}$, as follows:
\[
\move'_{(\otype,i)} = 
\big\{ (w,y,z) ~\mid~
(w,y,z)\in \move_\otype, ~ 
\exists^{!i} y'~~ y'<y \wedge (w,y',z)\in \move_\otype \big\}
\]
where $\exists^{!i} y'$ is an abbreviation for ``there exist exactly $i$ positions $y'$ such that\dots''.
As for the relation $\nxt'_{(\otype,i),(\otype',j)}$, this coincides with
$\nxt_{\otype,\otype'}$, so it does not take into account the new annotations.
Thus, the defined resynchronizer $R'$ is $1$-bounded, regular, and defines the 
same resynchronization as $R$.

\paragraph*{From 3.~to~4.}
Suppose that $R$ is a $1$-bounded regular resynchronizer with input
alphabet $\S$ and output alphabet $\G$, such that $T_1 \ocont R(T_2)$. 
The goal is to construct a $1$-bounded regular resynchronizer $R'$
with the same target set as $T_1$ and such that $T_1 \oeq R'(T_2)$.
For the sake of simplicity, we assume 
that $R$ has no input parameters, and similarly
$T_1$ has no common guess (the more general cases can be dealt with
by annotating the considered inputs with the possible parameters and
the common guess).
The idea for defining the desired resynchronizer $R'$ is as follows. 
We first restrict each relation $\move_\otype$ so as to make it a 
partial bijection, that is, for every input $u$, and every source origin $y\in\dom(u)$, 
there is an annotation $w$ of the input and at most one target origin $z$ 
that corresponds to $y$ in $u\otimes w$ (and conversely, since $R$ is $1$-bounded, 
for every target origin $z$ there is a unique source origin $y$ that 
corresponds to $z$). This step requires the use of appropriate input
parameters that determine a unique target origin $z$ from any given 
source origin $y$. Then, we restrict further the relation $\move_{\otype}$
so that every target origin $z$ is witnessed by $T_1$.
Formally, we introduce input parameters ranging over $\bbB^\G$
and work with annotated inputs of the form $u\otimes w$,
with $u\in\S^*$ and $w\in(\bbB^\G)^*$.
Given $u\in\S^*$, we define $O_u$ as the set of all positions $z=\Orig(\synch)(x)$
where $\synch$ is an origin graph of $T_1$, $x\in\dom(\Out(\synch))$,
and $\In(\synch)=u$.
The new relation $\move'_{\otype}$ that redirects source origins 
to target origins is defined as the following restriction of 
$\move_{\otype}$:
\[
\move'_{\otype} = 
\big\{ (u\otimes w,y,z) ~\mid~
(u,y,z)\in \move_\otype, ~ w(z)(\otype)=1, ~ z\in O_u ~
\big\}.
\]
Clearly, the above relation is regular and contained in $\move_{\otype}$. 
However, it is still possible that $\move'_{\otype}$ associates multiple target
origins with the same source origin.

To get a partial bijection from $\move'_{\otype}$ we need
to constrain the possible annotated input $u\otimes w$. 
We do so by requiring that, for every output letter $\otype\in\G$ 
and every position $y$ in $u\otimes w$, 
if there is $z$ satisfying $(u,y,z)\in \move_\otype$,
then there is {\sl exactly one} $z'$ satisfying $(u,y,z')\in \move_\otype$
and $w(\otype)(z)=1$.
Note that the latter property is again regular, and thus could be 
conjoined with the original relation $\ipar$ to form the new
relation $\ipar'$.
Accordingly, the relation $\nxt'_{\otype,\otype'}$ of the 
desired resynchronizer $R'$ defines the same language as 
$\nxt_{\otype,\otype'}$, but expanded with arbitrary
input annotations over $\bbB^\G$.

It is now easy to see that the the resulting resynchronizer $R'$ is $1$-bounded,
and in fact, on each input, defines a partial bijection between source and target
origins in such a way that the target set coincides with that of $T_1$.
By pairing this with the containments $R'(T_2)\ocont R(T_2)$ and $T_1 \ocont R(T_2)$, 
we obtain $T_1 \oeq R'(T_2)$.

\paragraph*{From 4.~to~1.}
Knowing that $\hat R(T_2) \oeq T_1 \oeq R(T_2)$ for two $1$-bounded resynchronizers
$R,\hat R$ with the same target sets as $T_1$ implies that the relations 
$\move_\otype$ and $\move'_\otype$, from $R$ and $\hat R$ respectively, coincide. 
Moreover, since the relation $\move_\otype$ of $R$ is assumed regular, this means 
that $\move'_\otype$ is regular too. Finally, we recall that  $\hat R$
is such that $\nxt'_{\otype,\otype'}$ is regular whenever $\move_\otype$ and $\move_{\otype'}$
are. We can then conclude that the relations $\nxt'_{\otype,\otype'}$ from $\hat R$
are also regular, and hence $\hat R$ is a regular resynchronizer.
\end{proof}


\fi

\medskip
Theorems \ref{th:two-way-equivalence-up} and
\ref{th:two-way-equivalence-up-bis} together provide a
characterization of those pairs of unambiguous two-way transducers $T_1,T_2$ for
which there is a bounded regular resynchronizer $R$ such that
$T_1 \ocont R(T_2)$.  
The effectiveness of this characterization stems
from the decidability of regularity of languages recognized by
unambiguous Parikh automata~\cite{CFM13}.
This result requires unambiguity and uses Presburger arithmetics to
determine for each (simple) loop a threshold such that iterating the
loop more than the threshold always satisfies the Parikh
constraint. The language of the Parikh automaton is regular if and
only if every  (simple) loop has such a threshold.
We thus conclude:

\begin{corollary}\label{cor:two-way-equivalence-up}
Given two unambiguous two-way transducers $T_1,T_2,$ one can decide whether
there is a regular resynchronizer $R$ such that $T_1 \ocont R(T_2)$.
\end{corollary}






\section{Conclusions}\label{sec:conclusions}

We studied two notions of resynchronization for transducers with origin,
called rational resynchronizer and regular resynchronizer. Rational
resynchronizers are suited for transforming origin graphs of one-way
transducers, while regular resynchronizers can be applied also to 
origin graphs of two-way transducers.
We showed that the former are strictly included in the latter, even
when restricting the origin graphs to be one-way.
We then studied the following variant of containment problem for transducers:
given two transducers $T_1,T_2$, decide whether $T_1 \ocont R(T_2)$ for 
some (rational or regular) resynchronizer $R$.  That is, if all origin graphs
of $T_1$ can be seen as some origin graph of $T_2$ transformed according
to $R$, then compute such a resynchronizer $R$. 
This problem can be seen as a synthesis problem of resynchronizers.
It is shown that the synthesis problem is decidable when
$T_1,T_2$ are finite-valued one-way transducers and the
resynchronizer is constrained to be rational, as well as 
when $T_1,T_2$ are unambiguous two-way transducers and
the resynchronizer is allowed to be regular (and bounded). 
In the one-way setting, the problem turns out to be undecidable 
already for unrestricted (non-functional) transducers and
rational resynchronizers. 
In the two-way setting, the decidability status remains open
already when the transducers are not unambiguous (be them functional or not). 
Concerning this last point, however, we recall that the synthesis 
problem becomes undecidable as soon as we consider regular resynchronizers
that are unbounded, as in this case the problem is at least as hard as
classical containment.



\ifARXIV

\fi

\ifLIPICSappendix
\ifLIPICS
\newpage
\appendix

\section{Proofs of Section \ref{sec:expressiveness}}

\RationalVsRegular*

We fix a one-way transducer $R$ over $\S\uplus\G$ that defines a rational resynchronizer.
We assume without loss of generality that $R$ is \emph{letter-to-letter}, as well as \emph{trimmed}, 
namely, every state in $R$ occurs in some successful run.
Note that $R$ maps synchronized words to synchronized words.
With a slight abuse of terminology, we shall use the terms `source' (resp.~`target') 
to refer to a synchronized word that is an input (resp.~an output) of $R$. 
When depicting examples, we will often adopt the convention that source synchronized
words are shown in blue, while target synchronized words are shown in red.
On the other hand, we shall use the terms `input' and `output' to refer to 
the projections of a synchronized word over $\S$ and $\G$, respectively
(note that, in this case, it does not matter whether the synchronized word
is the source or the target, since these have the same projections over 
$\S$ and $\G$).
The goal is to construct a $1$-bounded, regular resynchronizer $R'$, with parameters,
that defines the same resynchronization as $R$. 

We begin by introducing the key concept of \emph{lag}, which represents the difference
between the number of input symbols consumed and number of input symbols produced along
a certain run (not necessarily successful) of $R$. Formally, given a run of $R$ of the 
form $\r = q_0 \act{c_1 \:|\: d_1} q_1 \act{c_2 \:|\: d_2} \dots \act{c_n \:|\: d_n} q_n$,
we define its lag $\il(\r)$ as $|\projS(c_1\ldots
c_n)|-|\projS(d_1\ldots d_n)|$, 
where $\projS$ denotes the operation of projection onto the alphabet $\S$.
Note that, because $R$ is letter-to-letter, one could have equally defined 
$\il(\r)$ by counting the difference between produced output symbols and consumed output symbols.
Further note that the lag of a successful run is always $0$, since $R$ preserves the input projection.
Notice that the lag of a run is a notion distinct of the \emph{delay} of a 
rational resynchronizer presented in \cite{FiliotJLW16}
which is the maximum distance between the target origin of an output position and its source origin.
The following lemma shows that the lag is in fact a property of the initial 
and final states of a run. 
				
\begin{lemma}\label{lem:rational-lag}
For every two runs $\r_1$ and $\r_2$ of $R$ that begin with the same state and end with the same state,
$\il(\r_1)=\il(\r_2)$. 
\end{lemma}
	
\begin{proof}
Since $R$ is trimmed, both runs $\r_1$ and $\r_2$ can be completed to some 
successful runs of the form $\r' \r_1 \r''$ and $\r' \r_2 \r''$.
From $\il(\r' \r_1 \r'') = 0 = \il(\r' \r_2 \r'')$, it immediately follows 
that $\il(\r_1) = 0-(\il(\r') + \il(\r'')) = \il(\r_2)$.
\end{proof}
		
In view of the above lemma, we can associate a lag $\il(q)$ with each state $q$ of $R$
as follows: we choose an arbitrary run $\r$ that starts with the initial state of $R$
and ends with $q$, and let $\il(q)=\il(\r)$. This is well-defined since 
$\il(q)$ does not depend on the particular choice of $\r$.
For instance, if we consider the letter-to-letter resynchronizer $R$ of Example \ref{ex:rational-resync},
the only state with non-zero lag is the bottom one, which has lag $1$.
Note that, because each transition of $R$ can only increase or decrease the lag
by $1$, all lags range over the finite set $\{-|Q|,\dots,+|Q|\}$, where
$Q$ is the state space of $R$.

Next, we consider a successful run of $R$, say 
$\r= q_0 \act{c_1 \:|\: d_1} q_1 \act{c_2 \:|\: d_2} \dots \act{c_n \:|\: d_n} q_n$,
and define relations $\om_\r$ and $\im_\r$ between positions of $\r$.
These relations are used later to define a bijection between source and target origins.
The relation $\om_\r$ consists of all pairs $(i,j)$ of positions of $\r$ such that
$c_i$ and $d_j$ are output letters and 
$c_1 c_2\dots c_i \eqG d_1 d_2 \dots d_j$
(the latter is a shorthand for $\projG(c_1 c_2\dots c_i)  = \projG(d_1 d_2 \dots d_j)$). 
Note that $\om_\r$ is in fact a partial bijection.
In a similar way, we define $\im_\r$ as the partial bijection that contains all pairs $(i,j)$ 
of positions of $\r$ such that $c_i$ and $d_j$ are input letters 
and $c_1 c_2 \dots c_i \eqS d_1 d_2 \dots d_j$. 
	
\begin{example}\label{ex:match}
\ifLIPICS
Consider the pair of source and target synchronized 
words over $\S\uplus\G$ shown
\fi
\ifARXIV
We consider again the pair of source and target synchronized 
words that we used beofre to give an intuitive account of 
the proof of Theorem \ref{thm:rational-vs-regular}.
These synchronized words are  
\fi
\mywraptextjustified{7cm}{%
\ifLIPICS
to the right, where $\S=\{a\}$ and $\G=\{b\}$, 
which could be realized by a successful run $\r$ of $R$. 
For the moment, we overlook the blue and red arrows.
Because $R$ is letter-to-letter, any position in any of the
two words corresponds precisely to a position in the run
$\r$, so we can represent the
\fi
\ifARXIV
depicted
to the right, and assumed to be realized by a 
successful run $\r$ of $R$.
For the moment, we overlook the blue and red arrows.
Because $R$ is letter-to-letter, any position in any of the
two words corresponds precisely to a position in the run
$\r$, so we can represent the
relations $\om_\r$ 
\fi
}%
\mywrapfig{7cm}{%
	\vspace{-3.5mm}\hspace{4mm}%
\begin{tikzpicture}[xscale=1.25,yscale=0.8,baseline]
	
	\draw [nicecyan] (0,1.25) node [hidden] (1) {\small $\texttt{a}$};
	\draw [nicecyan] (0.3,1.25) node [hidden] (2) {\small $\texttt{a}$};
	\draw [nicecyan] (0.6,1.25) node [hidden] (3) {\small $\texttt{b}$};
	\draw [nicecyan] (0.9,1.25) node [hidden] (4) {\small $\texttt{a}$};
	\draw [nicecyan] (1.2,1.25) node [hidden] (5) {\small $\texttt{a}$};
	\draw [nicecyan] (1.5,1.25) node [hidden] (6) {\small $\texttt{a}$};
	\draw [nicecyan] (1.8,1.25) node [hidden] (7) {\small $\texttt{b}$};
	\draw [nicecyan] (2.1,1.25) node [hidden] (8) {\small $\texttt{b}$};
	\draw [nicecyan] (2.4,1.25) node [hidden] (9) {\small $\texttt{a}$};
	\draw [nicecyan] (2.7,1.25) node [hidden] (10) {\small $\texttt{a}$};
	\draw [nicecyan] (3.0,1.25) node [hidden] (11) {\small $\texttt{a}$};
	\draw [nicecyan] (3.3,1.25) node [hidden] (12) {\small $\texttt{b}$};
	\draw [nicecyan] (3.6,1.25) node [hidden] (13) {\small $\texttt{a}$};
	\draw [nicecyan] (3.9,1.25) node [hidden] (14) {\small $\texttt{a}$};
	\draw [nicecyan] (4.2,1.25) node [hidden] (15) {\small $\texttt{b}$};
	\draw [nicecyan] (4.5,1.25) node [hidden] (16) {\small $\texttt{b}$};
	\draw [nicecyan] (4.8,1.25) node [hidden] (17) {\small $\texttt{b}$};
	\draw [nicecyan] (5.1,1.25) node [hidden] (18) {\small $\texttt{a}$};
	
	\draw [nicered] (0,0) node [hidden] (1') {\small $\texttt{a}$};
	\draw [nicered] (0.3,0) node [hidden] (2') {\small $\texttt{b}$};
	\draw [nicered] (0.6,0) node [hidden] (3') {\small $\texttt{a}$};
	\draw [nicered] (0.9,0) node [hidden] (4') {\small $\texttt{a}$};
	\draw [nicered] (1.2,0) node [hidden] (5') {\small $\texttt{a}$};
	\draw [nicered] (1.5,0) node [hidden] (6') {\small $\texttt{b}$};
	\draw [nicered] (1.8,0) node [hidden] (7') {\small $\texttt{a}$};
	\draw [nicered] (2.1,0) node [hidden] (8') {\small $\texttt{a}$};
	\draw [nicered] (2.4,0) node [hidden] (9') {\small $\texttt{a}$};
	\draw [nicered] (2.7,0) node [hidden] (10') {\small $\texttt{b}$};
	\draw [nicered] (3.0,0) node [hidden] (11') {\small $\texttt{b}$};
	\draw [nicered] (3.3,0) node [hidden] (12') {\small $\texttt{a}$};
	\draw [nicered] (3.6,0) node [hidden] (13') {\small $\texttt{a}$};
	\draw [nicered] (3.9,0) node [hidden] (14') {\small $\texttt{b}$};
	\draw [nicered] (4.2,0) node [hidden] (15') {\small $\texttt{a}$};
	\draw [nicered] (4.5,0) node [hidden] (16') {\small $\texttt{b}$};
	\draw [nicered] (4.8,0) node [hidden] (17') {\small $\texttt{b}$};
	\draw [nicered] (5.1,0) node [hidden] (18') {\small $\texttt{a}$};

	\draw [dashed,link] (1') to (1);
	\draw [dashed,link] (5') to (5);
	\draw [dashed,link] (9') to (10);
	\draw [dashed,link] (13') to (13);
	\draw [dashed,link] (15') to (14);

	\draw [link] (2') to (3);
	\draw [link] (6') to (7);
	\draw [link] (10') to (8);
	\draw [link] (11') to (12);
	\draw [link] (14') to (15);
	\draw [link] (16') to (16);
	\draw [link] (17') to (17);
	
	\draw [arrow, nicecyan] (3.north) to [out=90,in=60,looseness=2] (2.north);
	\draw [arrow, nicecyan] (7.north) to [out=90,in=60,looseness=2] ([xshift=1pt]6.north);
	\draw [arrow, nicecyan] (8.north) to [out=90,in=70,looseness=2] ([xshift=-1.5pt]6.north);
	\draw [arrow, nicecyan] (12.north) to [out=90,in=60,looseness=2] (11.north);
	\draw [arrow, nicecyan] (15.north) to [out=90,in=60,looseness=2] ([xshift=1.25pt]14.north);
	\draw [arrow, nicecyan] (16.north) to [out=90,in=70,looseness=2] ([xshift=-1pt]14.north);
	\draw [arrow, nicecyan] (17.north) to [out=90,in=80,looseness=2] ([xshift=-3.25pt]14.north);

	\draw [arrow, nicered] (2'.south) to [out=-90,in=-60,looseness=2] (1'.south);
	\draw [arrow, nicered] (6'.south) to [out=-90,in=-60,looseness=2] (5'.south);
	\draw [arrow, nicered] (10'.south) to [out=-90,in=-60,looseness=2] ([xshift=1pt]9'.south);
	\draw [arrow, nicered] (11'.south) to [out=-90,in=-70,looseness=2] ([xshift=-1.5pt]9'.south);
	\draw [arrow, nicered] (14'.south) to [out=-90,in=-60,looseness=2] (13'.south);
	\draw [arrow, nicered] (16'.south) to [out=-90,in=-60,looseness=2] ([xshift=1pt]15'.south);
	\draw [arrow, nicered] (17'.south) to [out=-90,in=-80,looseness=2] ([xshift=-1.5pt]15'.south);
\end{tikzpicture}%
\hspace{-4mm}%
%
	\vspace{-3mm}
}
\ifLIPICS
relations $\om_\r$ 
\fi
and $\im_\r$ by means of edges between source and target positions. 
In the figure, the solid edges represent pairs of $\om_\r$, while
the dashed edges represent some pairs of $\im_\r$ 
(precisely, those pairs $(i,j)$ such that
the transition at position $j$ produces an input letter, while the next
transition produces an output letter).
\end{example}

\paragraph*{Mapping the source to target origins.} 
We now explain how the relations $\im_\r$ and $\om_\r$ 
can be used to define a mapping from source to target origins. 
We do so by first using the figure of Example \ref{ex:match}.
Consider any output letter at position $i$ in the source synchronized 
word $w$ (e.g.~the first blue letter $b$). 
Let $j$ be the last $\S$-labelled position before $i$, as indicated
by the blue arrow. This position $j$ determines the source origin 
$y=|\projS(w[1,j])|$ of the output letter.
To find the corresponding target origin, we observe that the position 
$i$ is mapped via the relation $\om_\r$ (solid line) to some position 
$k$ in the target synchronized word.
Let $h$ be the last $\S$-labelled position before $k$ (red arrow),
and map $h$ back to a position $\ell$ in the source via the relation
$\im_\r$ (dashed line).
The position $\ell$ determines precisely the target origin
$z=|\projS(w[1,\ell])|$ of the considered output letter.
The above steps describe a correspondence between two positions 
$j$ and $\ell$ in $\r$, with labels over $\S$, that is precisely defined by
\begin{align*}
  \exists i,k,h ~~  
  \begin{cases}
  \r[j,i] \text{ consumes a word in } \S\G^+ \\
  (i,k) \in \om_\r \\
  \r[h,k] \text{ produces a word in } \S\G^+ \\
  (h,\ell)\in \im_\r.
  \end{cases}
  \tag{$\star$}
\end{align*}
In the above $\rho[j,i]$ represents the part of $\rho$ between positions $j, i$ (both $j,i$ included).  

We denote by $\match_\r$ the relation of all pairs $(j,\ell)$
that satisfy Equation ($\star$).
Note that $\match$ determines an analogous correspondence 
between source and target origins of the input projection.
However, $\match$ has two issues: 
it is not yet a partial bijection (since different output positions 
may have the same source origin), and it needs to be implemented 
by means of a regular relation $\move_\otype$ that only considers 
positions of the input, plus the label $\otype$ of a single position 
in the output. 
Below, we explain how to overcome those issues.

\paragraph*{The case of bounded output blocks.}
Hereafter, we call \emph{output block} any maximal factor of a synchronized 
word that is labelled over $\G$. Intuitively, this corresponds to a maximal 
factor of the output that originates at the same input position.
We first consider, as a simpler case, a rational resynchronizer $R$ that 
reads {\sl source} synchronized words where the lengths of the output blocks 
are uniformly bounded by some constant, say $\bound$ (a similar property holds 
for the blocks of the target synchronized words, using lag-based arguments).
In this case we can encode any successful run $\r$ of $R$ entirely on the input,
by annotating every $\S$-labelled position $y$ with a factor $\r_y$
of $\r$ that reads the input symbol at position $y$, followed by the sequence 
of output symbols up to the next input symbol. Note that every factor $\r_y$ 
has length at most $\bound+1$.
The correctness of this input annotation can be checked by the regular
language $\ipar$. Given a factor $\rho_y \in \Sigma \Gamma^+$, 
$\rho_y[1] \in \Sigma$ is the first position of the factor $\rho_y$. 
Likewise, $\rho_y[i,j]$ denotes the subfactor of $\rho_y$ 
consisting of positions $i, i+1, \dots, j$.

In addition, we also annotate the output word with indices 
from $\{1,\dots,\bound\}$, called \emph{offsets}, in such a way that 
an output position $x$ is annotated with an offset $o$ if and only if
it is the $o$-th output position with the same source origin.
Note that the correctness of the annotation cannot be checked by 
a regular language such as $\opar$ that refers only to the output.
The check will be done instead by a combined use of the relations
$\move_\otype$ and $\nxt_{\otype,\otype'}$.

We first check that, for every pair of consecutive output positions $x$ 
and $x+1$ annotated with the offsets $o$ and $o'$, respectively, 
it holds that $o'=o+1$ or $o'=1$, depending on whether the 
\emph{source} origins of $x$ and $x+1$ coincide or not. For this
we let $(u,z,z') \in \nxt_{\otype,\otype'}$, with $\otype=(a,o)$
and $\otype'=(a',o')$, if
\begin{enumerate}
  \item either $o'=o+1$ and there is $y=y'$ such that
        $(u,y,z)\in\move_\otype$ and $(u,y',z')\in\move_{\otype'}$,
  \item or $o'=1$ and there are $y<y'$ such that
        $(u,y,z)\in\move_\otype$ and $(u,y',z')\in\move_{\otype'}$.
\end{enumerate}
Recall that the relation $\nxt_{\otype,\otype'}$ must be
defined in terms of the \emph{target} origins of $x$ and $x+1$. 
So it needs to rely on the relation $\move_\otype$
in order to determine the source origins 
from the target origins. We assume that for every 
output type $\otype$ 
the relation $\move_\otype$, 
which will be defined later, determines a~\emph{partial bijection} 
between input positions (we will see that this is indeed the case). 
Based on these assumptions, the above definition of $\nxt_{\otype,\otype'}$ 
guarantees that the offsets annotating consecutive positions 
in the output are either incremented or reset, depending on whether
they have the same origin or not.

It remains to check that maximal offset occurring in an output block 
with origin $y$ coincides with number of output symbols produced by 
the corresponding factor $\rho_y$ of the run.
Thus, we modify slightly the definition of $\nxt_{\otype,\otype'}$
in case 2., as follows:
\begin{enumerate}
  \item[2'.] or $o'=1$ and there are $y<y'$ such that
        $(u,y,z)\in\move_\otype$ and $(u,y',z')\in\move_{\otype'}$,
        \emph{and $o=|\rho_y|-1$}.
\end{enumerate}
Note that the factor $\rho_y$ can be derived by inspecting the
annotation of the input position $y$. 
The modification suffices to guarantee that the output annotation is 
correct for all output blocks but the last one. 
The annotation for the last output block can be checked by
marking the last output position with a distinguished symbol
and by requiring that if $\otype$ witnesses the marked symbol
and the offset $o$, then $\move_\otype$ can only contain a triple 
of the form $(u,y,z)$, with $o=|\rho_y|-1$.
We omit the tedious definitions in this case.


\medskip
Now, having the input correctly annotated with the factors $\r_y$ 
of $\r$ and the output correctly annotated with the offsets, 
we can encode any position $i$ of $\r$ by a pair 
$(y,o)$ that consists of a position $y$ of the input and an offset 
$o\in\{0,1,\dots,\bound\}$. 
The encoding is defined in such a way that $i = \sum_{y'<y}|\r_{y'}| + o + 1$
(in particular, $o=0$ when the transition at position $i$ consumes an
input symbol, otherwise $o\ge1$).
We use this encoding to translate the relations $\om_\r$, $\im_\r$, 
and $\match_\r$, to equivalent finite unions of partial bijections between 
input positions.
We begin by explaining the translation of $\om_\r$.

\paragraph*{Translation of $\om_\r$.} 
Consider any pair $(i,j)\in\om_\r$. Since the transition at position
$i$ of $\r$ consumes an output symbol, it is encoded by a pair of the form 
$(y,o)$, with $o\ge1$.
On the other hand, the transition at position $j$ may consume 
either an input symbol or an output symbol (but does produce an output symbol).
In the former case, $j$ is encoded by a pair $(y',0)$;
in the latter case, it is encoded by a pair $(y',o')$, with $o'\ge1$.
As an example, in the figure below, $(7,4) \in \om_\r$. Position 7 of
the run is encoded as $(5,1)$ on the input.   
The transition at position 4 consumes an input symbol $a$, and 
produces the output symbol $b$, and is encoded as $(3,0)$.  


\begin{tikzpicture}[xscale=1.25,baseline]
	\draw [nicecyan] (0,1.25) node [hidden] (1) {\small $\texttt{a}$};
	\draw [black] (0,1.5) node [hidden] (1'') {\tiny $\texttt{(1,0)}$};
	\draw [nicecyan] (0.5,1.25) node [hidden] (2) {\small $\texttt{a}$};
	\draw [black] (0.5,1.5) node [hidden] (2'') {\tiny $\texttt{(2,0)}$};
	\draw [nicecyan] (1,1.25) node [hidden] (3) {\small $\texttt{b}$};
	\draw [black] (1,1.5) node [hidden] (3'') {\tiny $\texttt{(2,1)}$};

	\draw [nicecyan] (1.5,1.25) node [hidden] (4) {\small $\texttt{a}$};
	
	\draw [black] (1.5,1.5) node [hidden] (4'') {\tiny $\texttt{(3,0)}$};

	\draw [nicecyan] (2,1.25) node [hidden] (5) {\small $\texttt{a}$};
	\draw [black] (2,1.5) node [hidden] (5'') {\tiny $\texttt{(4,0)}$};

	\draw [nicecyan] (2.5,1.25) node [hidden] (6) {\small $\texttt{a}$};
	\draw [black] (2.5,1.5) node [hidden] (6'') {\tiny $\texttt{(5,0)}$};

	\draw [nicecyan] (3,1.25) node [hidden] (7) {\small $\texttt{b}$};
	\draw [black] (3,1.5) node [hidden] (7'') {\tiny $\texttt{(5,1)}$};
	
	\draw [nicecyan] (3.5,1.25) node [hidden] (8) {\small $\texttt{b}$};
	\draw [black] (3.5,1.5) node [hidden] (8'') {\tiny $\texttt{(5,2)}$};

	\draw [nicecyan] (4,1.25) node [hidden] (9) {\small $\texttt{b}$};
	\draw [black] (4,1.5) node [hidden] (9'') {\tiny $\texttt{(5,3)}$};

	\draw [nicered] (0,0) node [hidden] (1') {\small $\texttt{a}$};
	\draw [nicered] (0.5,0) node [hidden] (2') {\small $\texttt{b}$};
	
	\draw [nicered] (1,0) node [hidden] (3') {\small $\texttt{a}$};
	\draw [nicered] (1.5,0) node [hidden] (4') {\small $\texttt{b}$};

	\draw [nicered] (2,0) node [hidden] (5') {\small $\texttt{a}$};
	\draw [nicered] (2.5,0) node [hidden] (6') {\small $\texttt{b}$};
	\draw [nicered] (3,0) node [hidden] (7') {\small $\texttt{b}$};
	\draw [nicered] (3.5,0) node [hidden] (8') {\small $\texttt{a}$};
	\draw [nicered] (4,0) node [hidden] (9') {\small $\texttt{a}$};
\node[draw=red,dotted,fit=(4')(5')(6')] {};
	\draw [link] (2') to (3);
	\draw [link] (4') to (7);
\end{tikzpicture}
\medskip

\noindent
In general, we observe that the lag induced just after the 
$o$-th transition of $\r_y$ must be equal to the number 
of output symbols produced between the $(o'+1)$-th 
transition of $\r_{y'}$ and the $o$-th transition of $\r_y$,
both included
(when the lag is negative one follows the transitions
 in reverse order, counting negatively). As an illustration in the figure,   
 the lag after the first transition of $\rho_5$ is 2, which is the number 
 of output symbols in the dotted box. The dotted box 
 consists of the symbols produced between the first transition 
 of $\r_{3'}$ and the first transition of $\r_5$, and has two output symbols.

\paragraph*{Translation of $\im_{\r}$.} 
The translation of the relation $\im_\r$ is similar.
The only difference is that now the pairs $(i,j)\in\im_\r$ 
are encoded by tuples of the form 
$\big((y,o),(y',o')\big)$, with $o=0$ since the transition at $i$ 
consumes an input symbol. The transition at position $j$ as before, 
can consume an input symbol or an output symbol. 
Consider the figure below, where $(2,3) \in \im_{\r}$. 
Position $i=2$ is encoded as 
$(2,0)$.  The transition at position 3 
consumes an output symbol $b$ (and produces the input symbol $a$). 
Position 3 is encoded as $(2,1)$. 
	
\begin{tikzpicture}[xscale=1.25,baseline]
	\draw [nicecyan] (6,1.25) node [hidden] (1) {\small $\texttt{a}$};
		\draw [nicecyan] (6.5,1.25) node [hidden] (2) {\small $\texttt{a}$};
	\draw [nicecyan] (7,1.25) node [hidden] (3) {\small $\texttt{b}$};
	\draw [nicecyan] (7.5,1.25) node [hidden] (4) {\small $\texttt{a}$};
	\draw [black] (7.5,1.5) node [hidden] (4'') {\tiny $\texttt{(3,0)}$};
	
	\draw [nicecyan] (8,1.25) node [hidden] (5) {\small $\texttt{a}$};
	\draw [black] (8,1.5) node [hidden] (5'') {\tiny $\texttt{(4,0)}$};
	
	\draw [nicecyan] (8.5,1.25) node [hidden] (6) {\small $\texttt{a}$};
	\draw [nicecyan] (9,1.25) node [hidden] (7) {\small $\texttt{b}$};
	
	\draw [nicecyan] (9.5,1.25) node [hidden] (8) {\small $\texttt{b}$};
	\draw [nicecyan] (10,1.25) node [hidden] (9) {\small $\texttt{b}$};
	\draw [nicered] (6,0) node [hidden] (1') {\small $\texttt{a}$};
	\draw [nicered] (6.5,0) node [hidden] (2') {\small $\texttt{b}$};
	\draw [nicered] (7,0) node [hidden] (3') {\small $\texttt{a}$};
	\draw [nicered] (7.5,0) node [hidden] (4') {\small $\texttt{b}$};
	\draw [nicered] (8,0) node [hidden] (5') {\small $\texttt{a}$};
	\draw [nicered] (8.5,0) node [hidden] (6') {\small $\texttt{b}$};
	\draw [nicered] (9,0) node [hidden] (7') {\small $\texttt{b}$};
	\draw [nicered] (9.5,0) node [hidden] (8') {\small $\texttt{a}$};
	\draw [nicered] (10,0) node [hidden] (9') {\small $\texttt{a}$};
\node[draw=red,dotted,fit=(4') (5')] {};
	\draw [dashed,link] (1') to (1);
	\draw [dashed,link] (3') to (2);
	\draw [dashed,link] (4) to (5');
\end{tikzpicture}
\medskip

\noindent
The only difference here is that 
one has to relate the lag with the 
number of {\sl input} letters produced between (both positions included) the first transition 
of $\r_y$ and the $o'$-th transition 
of $\r_{y'}$. 
Again, in the figure, the lag after the first transition of $\r_3$ is 
1, which is the number of input symbols in the 
dotted box. The dotted box contains the symbols 
produced between the first transition of 
$\r_3$ and the first transition of $\rho_{4'}$, and has one input 
symbol.

\paragraph*{Relations encoding $\om_\r$ and $\im_\r$.} 
So we can represent $\om_\r$ as a finite union of
relations $O_{o,o'}\subseteq(\S\times\S')^*\times\Nat\times\Nat$,
each describing a regular property of annotated inputs with two 
distinguished positions in it, in such a way that the positions
are bijectively related to one another.

Likewise, we can represent $\im_\r$ as a finite union of relations 
$I_{0,o'}$, each describing a regular property of annotated inputs 
with two distinguished positions encoded as $(y,0)$ and $(y', o')$ in it, 
which are bijectively related to one another.

\paragraph*{Translation of $\match_\r$.} 
We finally turn to the translation of the relation $\match_\r$,
which will eventually determine the relations $\move_\otype$ 
of the desired regular resynchronizer $R'$.
This is done by mimicking Equation ($\star$) via the encoding
of  positions in the run $\rho$ using pairs of input positions and offsets,
and more precisely, by replacing the variables $j,i,k,h,\ell$ 
of Equation (*) with the pairs $(y,0)$, $(y,o)$, $(y',o')$, $(y'',o'')$, $(z,0)$.

Formally, for every offset $o\in\{1,\dots,\bound\}$, we define 
the set $M_o$ of all triples $(u,y,z)$, where $u$ is an annotated 
input and $y,z$ are positions in it that satisfy the following 
property:

\begin{align*}
  \exists y',y''  
  \bigvee_{0\le o',o''\le\bound} ~~ 
  \begin{cases}
  \r_y[1,o+1] \text{ consumes a word in $\S\G^+$} \\
  (u,y,y') \in O_{o,o'} \\
  \r_{y''}[o''+1,|\r_{y''}|] ~ \r_{y''+1} \dots \r_{y'-1} ~ \r_{y'}[1,o'+1] 
    \text{ produces a word in $\S\G^+$} \\
  (u,z,y'') \in I_{0,o''}.
  \tag{$\star\star$}
  \end{cases}
\end{align*}

Note that the first condition holds trivially by definition of $\r_y$, 
while the third condition is easily implemented by accessing the
factors $\r_{y''},\dots, \r_{y'}$
of $\r$ that are encoded by the input parameters. 
For simplicitly, here we assumed that $(y'',o'')$ is lexicographically 
before $(y',o')$; to treat the symmetric case, one has to interpret 
the definition by considering the sequence of transitions in reverse.
The intended meaning of $(u,y,z)\in M_o$ is as follows. 
Suppose that the input is correctly annotated with the factors $\r_y$
of a successful run $\r$ of $R$, and that the output 
position $x$ of $\r$ is correctly annotated with an offset $o$. 
Assuming that $x$ is the $o$-th output position with source origin $y$, 
then $z$ is its target origin in $\r$.

Continuing with our running example, we determine the target origin 
for the point $b$ annotated (5,1), whose source origin is 
(5,0). We will find the target origin of this $b$ annotated (5,1). As seen in the computation of $\om_\r$, we know that 
$(u,5,3) \in O_{1,0}$.  The factor  $\rho_5=abbb$, and 
$\rho_5[1,2]=ab \in \Sigma \Gamma^+$, and as we have seen, 
$(u,5,3) \in O_{1,0}$. Now, consider the part of the source $u$ annotated with 
$(2,1)(3,0)$. This produces the output $ab \in \Sigma \Gamma^+$. That is, for $y''=2, o''=1$, and $y'=3, o'=0$, we have
$\rho_{y''}[o''+1, 2] \:  \rho_{y'}[1,o'+1]=\rho_2[2,2]\rho_3[1,1]=ba$ produces the output 
$ab \in \Sigma \Gamma^+$. 

Consider $(z,0)=(2,0)$. The lag after the $a$ at $i=2$ annotated $(2,0)$ is 1. Also, 
$(2,3) \in \im_\r$. The position 3 consumes an 
output and produces an input $a$. 
Indeed, the lag after the first transition of $\rho_2$ is 
1, which is the number of input symbols between the first transition 
of $\rho_2$ and the second transition ($(o'+1)$th transition) of  
$\rho_2$. That is, $(u,2,2) \in I_{0,1}$. Thus, starting with the $b$ annotated $(y,o)=(5,1)$ such that 
$\rho_5[1,2] \in \Sigma \Gamma^+$,  
we first obtain $(y',o')=(3,0)$ with $(u,5,3) \in O_{1,0}$. Further, 
 $\rho_2[2,2]\rho_3[1,1]$ produces a word in $\Sigma \Gamma^+$. 
 Finally, we have $(u,2,2) \in I_{0,1}$, obtaining $(u,5,2) \in M_1$.

\begin{tikzpicture}[xscale=1.25,baseline]
	\draw [nicecyan] (6,1.25) node [hidden] (1) {\small $\texttt{a}$};
	\draw [black] (6,1.5) node [hidden] (1') {\tiny $\texttt{(1,0)}$};
			\draw [nicecyan] (6.5,1.25) node [hidden] (2) {\small $\texttt{a}$};
	\draw [black] (6.5,1.5) node [hidden] (2') {\tiny $\texttt{(2,0)}$};
		\draw [nicecyan] (7,1.25) node [hidden] (3) {\small $\texttt{b}$};
	\draw [black] (7,1.5) node [hidden] (3') {\tiny $\texttt{(2,1)}$};
			\draw [nicecyan] (7.5,1.25) node [hidden] (4) {\small $\texttt{a}$};
	\draw [black] (7.5,1.5) node [hidden] (4') {\tiny $\texttt{(3,0)}$};
		\draw [nicecyan] (8,1.25) node [hidden] (5) {\small $\texttt{a}$};
	\draw [black] (8,1.5) node [hidden] (5') {\tiny $\texttt{(4,0)}$};
		\draw [nicecyan] (8.5,1.25) node [hidden] (6) {\small $\texttt{a}$};
	\draw [black] (8.5,1.5) node [hidden] (6'') {\tiny $\texttt{(5,0)}$};
	\draw [nicecyan] (9,1.25) node [hidden] (7) {\small $\texttt{b}$};
	\draw [black] (9,1.5) node [hidden] (7'') {\tiny $\texttt{(5,1)}$};
	\draw [nicecyan] (9.5,1.25) node [hidden] (8) {\small $\texttt{b}$};
	\draw [black] (9.5,1.5) node [hidden] (7'') {\tiny $\texttt{(5,2)}$};
	\draw [nicecyan] (10,1.25) node [hidden] (9) {\small $\texttt{b}$};
	\draw [black] (10,1.5) node [hidden] (7'') {\tiny $\texttt{(5,3)}$};

	\draw [nicered] (6,0) node [hidden] (1') {\small $\texttt{a}$};
	\draw [nicered] (6.5,0) node [hidden] (2') {\small $\texttt{b}$};
	\draw [nicered] (7,0) node [hidden] (3') {\small $\texttt{a}$};
	\draw [nicered] (7.5,0) node [hidden] (4') {\small $\texttt{b}$};
	\draw [nicered] (8,0) node [hidden] (5') {\small $\texttt{a}$};
	\draw [nicered] (8.5,0) node [hidden] (6') {\small $\texttt{b}$};
	\draw [nicered] (9,0) node [hidden] (7') {\small $\texttt{b}$};
	\draw [nicered] (9.5,0) node [hidden] (8') {\small $\texttt{a}$};
	\draw [nicered] (10,0) node [hidden] (9') {\small $\texttt{a}$};
	\draw [dashed,link] (3') to (2);
	
\end{tikzpicture}

\smallskip
\paragraph*{Definition of $\move_\otype$.} 
It is tempting to define $\move_\otype$ just as $M_o$, 
for every $\otype=(a,o)\in\G\times\{1,\dots,\bound\}$. 
However, we recall that the correctness of the output annotation is guaranteed
only once we are sure that every relation $\move_\otype$ defines a partial 
bijection between input positions $y$ and $z$ (hereafter we say for short that
the relation is bijective), which is not known a priori.
Bijectiveness must then be enforced syntactically, without relying on annotations: 
for this it suffices to define $\move_\otype$ as 
$\{ (u,y,z)\in M_o \:\mid\: 
    \forall (u,y',z')\in M_o ~ (y=y')\leftrightarrow (z=z') \}$,
and observe that either $M_o$ is bijective, and hence
$\move_\otype=M_o$, 
or it is not, and in this case $\move_\otype$ is a subrelation of $M_o$
that is still bijective.
Note that, in the case where $\move_\otype$ is a subrelation of $M_o$,
there will be no induced pair of synchronized words, since the origins
of some output elements could not be redirected. 
This is fine, and actually needed, in order to avoid generating with $R'$
spurious pairs of synchronized words that are not also generated by $R$.
On the other hand, observe that the relation $\move_\otype$ does generate, 
for appropriate choices of the output annotations, all the pairs of 
synchronized words that are generated by $R$.
We finally observe that the relations $\move_\otype$ and $\nxt_{\otype,\otype'}$
are regular. We obtain in this way,  a $1$-bounded, regular resynchronizer $R'$
equivalent to $R$.

\paragraph*{The general case.} 
We now aim at generalizing the previous ideas to capture a
rational resynchronizer $R$ with source output blocks of possibly
unbounded length. One additional difficulty is that we cannot 
anymore encode a successful run $\r$ of $R$ entirely on the input, 
as $\r$ may have arbitrarily long factors on outputs blocks.
Another difficulty is that we cannot uniquely identify the
positions in an output block using offsets ranging over a fixed 
finite set. We will see that a solution to both problems comes 
from covering most of the output by factors in which the positions 
behave similarly in terms of the source-to-target origin transformation. 
Intuitively, each of these factors can be thought of as 
a `pseudo-position', and accordingly the output blocks can be thought 
of as having boundedly many pseudo-positions.
This will make it possible to apply the same ideas as before. 
We now state the key lemma that 
identifies the aforesaid factors.
By a slight abuse of terminology, we call output blocks 
also the maximal $\G$-labelled factors of a synchronized word.

\begin{lemma}\label{lem:blocks}
Let $\r$ 
be a successful run of $R$, and let $w$ and $w'$ 
be the source and target synchronized words induced by $\r$.
\begin{itemize}
\item Every output block $v$ of $w$ can be factorized into $\cO(|Q|^2)$ 
      sub-blocks $v_1,\dots,v_n$ such that
	  if $|v_i|>1$ and $\r_i$ is the factor of $\r$ that corresponds to $v_i$,
	  then all states in $\r_i$ have the same lag, say $\ell_i$,
	  and the factor obtained by extending $\r_i$ to the left and 
	  to the right by exactly $|\ell_i|$ transitions forms a loop of $R$.
\item Moreover, for every factorization $v=v_1\dots v_n$ as above,
      each sub-block $v_i$ is also a factor of $w'$, and hence all 
      positions in $v_i$ have the same target origin.
\end{itemize}
\end{lemma}

\begin{proof}
%
\begin{figure}
\centering  
\begin{tikzpicture}

\draw (0,0) edge [|-|] (10,0);

\draw [dashed,rounded corners=5] (1,-0.2) rectangle (4,0.2);
\draw [dashed,rounded corners=5] (5,-0.2) rectangle (9,0.2);

\draw [draw=none,pattern=north east lines] (1.75,-0.2) rectangle (3.25,0.2);
\draw [draw=none,pattern=north east lines] (5.75,-0.2) rectangle (8.25,0.2);

\draw (7.5,2) node [above] {singleton sub-blocks};
\draw [arrow,dotted] (7,2) to [out=-10,in=110] (10,0.3);
\draw [arrow,dotted] (7,2) to [out=-20,in=100] (9.25,0.3);
\draw [arrow,dotted] (7,2) to [out=-30,in=90] (8.5,0.3);
\draw [arrow,dotted] (7,2) to [out=-144,in=80] (5.5,0.3);
\draw [arrow,dotted] (7,2) to [out=-150,in=70] (4.5,0.3);
\draw [arrow,dotted] (7,2) to [out=-155,in=60] (3.5,0.3);
\draw [arrow,dotted] (7,2) to [out=-160,in=60] (1.5,0.3);
\draw [arrow,dotted] (7,2) to [out=-165,in=55] (0.75,0.3);
\draw [arrow,dotted] (7,2) to [out=-170,in=50] (0,0.3);

\draw (1.25,2) node [above] {non-overlapping maximal loops};
\draw [arrow,thick] (1,2) to (1.25,0.3);
\draw [arrow,thick] (1.75,2) to (5.25,0.3);

\draw (4.75,-1) node [below] {loops shrinked by lag};
\draw [arrow,thick] (4.25,-1) to (2.5,-0.3);
\draw [arrow,thick] (5,-1) to (7,-0.3);
\end{tikzpicture}
\caption{Factorization of an output block.}
\label{fig:loop}
\end{figure}


%
%
We prove the first claim of the lemma (Figure \ref{fig:loop} provides an
intuitive account of the constructions).
Let $v$ be an output block of the source synchronized word $w$ and let
$\r'$ be the factor of the run $\r$ aligned with $v$. As a preliminary
step, we fix a maximal set of pairwise non-overlapping maximal loops inside $\r'$,
say $\r'_1,\dots,\r'_m$. A simple counting argument shows that $m\le|Q|$ and that
there are at most $|Q|$ positions in $\r'$ that are not covered by the loops 
$\r'_1,\dots,\r'_m$. The latter positions determine some sub-blocks of $v$
of length $1$. The remaining sub-blocks of $v$ will be obtained by factorizing
the loops $\r'_1,\dots,\r'_m$, as follows. Consider any loop $\r'_j$. 
By construction, all letters consumed by $\r'_j$ occur in $v$, so they must be 
output letters.
Similarly, all letters produced by $\r'_j$ are also output letters, since otherwise,
by considering repetitions of the loop $\r'_j$, one could get different lags,
violating Lemma \ref{lem:rational-lag}. 
This means that the lag associated with the states along $\r'_j$ is constant, 
say $\ell_j$ ($\le|Q|$).
If $\r'_j$ has length at most $2|\ell_j|$, then we simply 
decompose it into $2|\ell_j|$ factors of length $1$.
Otherwise, we cover a prefix of $\r'_j$ with $|\ell_j|$ factors of length $1$,
and a suffix of $\r'_j$ with $|\ell_j|$ other factors of length $1$.
The remaining part of $\r'_j$ is covered by a last factor of length
$|\r'_j| - 2|\ell_j|$. Overall, this induces a factorization of $v$ 
into at most $|Q|$ (the sub-blocks not covered by a loop)  + $|Q|\cdot (2|Q| + 1)$ (Each $\rho'_j$ 
is decomposed into $(2\ell_j+1) \leq (2|Q|+1)$ sub-blocks). This gives 
$\cO(|Q|^2)$ sub-blocks $v_1,\dots,v_n$.
Moreover, by construction, if $|v_i|>1$, then in the corresponding factor
$\r_i$ of $\r$, all states have the same lag, say $\ell_i$, and if we
extend $\r_i$ to the left and to the right by exactly $|\ell_i|$ transitions,
we get back one of the loops $\r'_j$ (recall that each loop $\r'_j$ 
of length $>2|\ell_j|$ is decomposed into  $|\ell_j|$ blocks of length 1, then 
a block of length $|\r'_j|-|\ell_j|$, and finally, $|\ell_j|$ blocks of length 1.
Clearly, if we extend the middle block on either side by 
blocks of length $|\ell_j|$, then we get back $\r'_j$. 
 This proves the first claim of the lemma.

As for the second claim, suppose that $v_1,\dots,v_n$ is a factorization 
of an output block $v$ of $w$ satisfying the first claim. Clearly, every
sub-block $v_i$ of length $1$ is also a factor of the target synchronized
word $w'$. The interesting case is when a sub-block $v_i$ has length 
larger than $1$. In this case, by the previous claim, we know that in
the corresponding factor $\r_i$ of $\r$, all states have the same lag $\ell_i$,
and the factor $\r'_i$ of $\r$ that is obtained by expanding $\r_i$ to the 
left and to the right by $|\ell_i|$ transition is a loop.
In fact, since $\r'_i$ is a loop, we also know that all states in it
have lag $\ell_i$.
Now, to prove that $v_i$ is a factor of the target synchronized word $w'$,
it suffices to show that every two consecutive positions of $\r_i$
are mapped to consecutive positions via the relation $\om_\r$.
This follows almost by construction, since for every 
pair $(i',k')\in\om_\r$, if $i'$ occurs inside the factor $\r_i$, 
then $k'$ occurs inside the loop $\r'_i$ (recall that $\r'_i$ consumes
and produces only output symbols), and hence $k'=i' - \ell_i$.
In addition, if $i'+1$ also occurs inside $\r_i$, then clearly
$(i'+1,k'+1)\in\om_\r$.
This proves that $v_i$ is a factor of the target synchronized word $w'$,
and hence all positions in it have the same target origin.
\end{proof}

In view of the above lemma we can guess a suitable factorization of 
the output into sub-blocks that refine the output blocks, and treat 
each sub-block as if it were a single position. 
In particular, we can annotate every sub-block with a unique offset from 
a finite set of quadratic size w.r.t.~$|Q|$. The role of the offsets
will be the same as in the previous proof, where blocks had bounded length,
namely, determine some partial bijections $O_{o,o'}$, $I_{0,o'}$, and $M_o$
between positions of the input.
In addition, we annotate every sub-block with the pair consisting 
of the first and last states of the factor of the successful run 
that consumes that sub-block. We call such a pair of states a 
\emph{pseudo-transition}, as it plays the same role of a transition
associated with a single output position.
Finally, we annotate every input position $y$ with a sequence
of bounded length that represents a single transition on $y$
followed by the pseudo-transitions on the subblocks with source
origin $y$. The resulting input annotation provides an abstraction
of a successful run of $R$.

The correctness of the above annotations can be enforced by defining suitable
relations $\ipar$, $\opar$, $\nxt_{\otype,\otype'}$ for the regular resynchronizer $R'$.
We omit the tedious details concerning these relations, and only observe that, 
as before, the definition $\nxt_{\otype,\otype'}$ relies on the fact that 
$\move_\otype$ and $\move_{\otype'}$ define partial bijections between
input positions.

Finally, we turn to describing the relation $\move_\otype$ that 
maps source to target origins for $\otype$-labelled output positions.
The definition is basically the same as before, based on some
auxiliary relations $O_{o,o'}$ and $I_{0,o''}$ that implement
$\om_\r$ and $\im_\r$ at the level of input positions.
As before, we guarantee, by means of a syntactical trick, 
that $\move_\otype$ determines a partial bijection between 
input positions.
In conclusion, we get a regular resynchronizer $R'$, with input and output parameters,
that is equivalent to the rational resynchronizer $R$.




\section{Proofs of Section \ref{sec:synthesis}}

\oneWayEquivalenceUp*

\begin{proof}
One implication, from 2.~to~1., is trivial, 
since origin containment implies classical containment,
and since applying an arbitrary resynchronization $R$ to $T_2$ 
cannot result in having more input-output pairs (it can however
modify the origin, as well as discard some input-output pairs).
The implication from 3.~to 2.~is also trivial.

For the remaining implication, from 1.~to~3., suppose that $T_1,T_2$ are 
functional one-way transducers such that $T_1 \subseteq T_2$. 
We construct a rational resynchronizer $R$ over the disjoint union $\Sigma\uplus\Gamma$ 
of the input and output alphabets of $T_1,T_2$, using a variant of the 
direct product of $T_1$ and $T_2$. 
More precisely, let $T_1=(Q_1,q_1,\Delta_1,F_1)$,
$T_2=(Q_2,q_2,\Delta_2,F_2)$, and $R=(Q,q,\Delta,F)$, where
$Q = Q_1\times Q_2$, $q=(q_1,q_2)$, $F=F_1\times F_2$, $\Delta$ contains all transitions
of the form $(s_1,s_2) \act{a w_2 \,\mid\, a w_1} (t_1,t_2)$, with $s_i \act{a \,\mid\, w_i} t_i$ 
in $\Delta_i$ for both $i=1$ and $i=2$.
Intuitively, the transducer $R$ simulates a run of $T_1$ and a run of $T_2$ in parallel, by
repeatedly consuming an input symbol $a$ and the corresponding output $w_2$ produced by $T_2$, 
and producing the same input symbol $a$ and the corresponding output $w_1$ of $T_1$.
Since $T_1$ and $T_2$ are functional and classically contained one in the other, 
we have that $R$ maps strings over $\Sigma\uplus\Gamma$ to strings over $\Sigma\uplus\Gamma$ 
while preserving the projections on the input and on the output alphabets.
This means that $R$ is indeed a resynchronizer. 
Finally, $T_1$ is clearly origin equivalent to $R(T_2)$.
\end{proof}



\synthesisUnaryOneWay*

\begin{proof}
We first prove the reduction from synthesis of rational resynchronizers 
to language-boundedness of OCA, and then prove the reduction in 
the opposite direction.

\paragraph*{From synthesis to language-boundedness.}
Let $T_1,T_2$ be real-time, one-way transducers with unary output
alphabet.  We suppose in addition that $T_1$ is trimmed.
We construct an OCA $A$ that reads encodings 
of successful runs of $T_1$.
If the input is not a successful run of $T_1$, 
then, as soon as an error is detected, $A$ resets its counter
and accepts any continuation of the input.
In particular, thanks to this behaviour and to $T_1$ being trimmed, badly-formed encodings 
of runs will not cause the counter of $A$ to be unbounded.

Consider now an input for $A$ that is a correct 
encoding of a successful run of $T_1$, say $\r_1$. 
In this case, $A$ guesses and simulates a successful run $\r_2$ of $T_2$ 
having the same input as $\r_1$. 
The counter of $A$ is used as expected: it is incremented according 
to the outputs produced using the transitions of $\r_1$, and decremented according to the outputs produced using the transitions
of $\r_2$, or vice versa when one needs to represent a negative value
(recall that OCA work with counter over natural numbers). 
The detail regarding which among $T_1, T_2$ is ``leading'', resulting 
 in the non-negative counter value can be stored in the finite control 
 of the OCA.

Intuitively, a configuration of $A$ determines how ahead or behind 
is the partial output produced by the encoded run of $T_1$ compared to 
the partial output produced by the simulated run $T_2$.
The OCA $A$ accepts with empty counter. 
Note that this construction is close to the direct product of 
$T_1$ and $T_2$, the main difference being the treatment of the
badly formed encodings and the role played by the counter.

Let us now prove that the OCA $A$ is language-bounded if and only
if $T_1 \ocont R(T_2)$ for some rational resynchronizer $R$.

Suppose first that the OCA $A$ is language-bounded, namely, that 
there is some $k\in\Nat$ such that every word is accepted by $A$ 
with a counter that never exceeds $k$. 
We can think of the successful runs of $A$ that maintain the counter
between $0$ and $k$ as runs of a $k$-delay resynchronizer $R$. 
More precisely, we can define a letter-to-letter resynchronizer $R$,
the states of which are  the configurations of $A$ with the value of the
counter inside $\{0,\dots,k\}$.
On consuming an input letter, $R$ produces the same input letter; 
on consuming a sequence of $j$ output letters, depending 
on the simulated transition of $A$, $R$ produces an output 
of length $j+h$ if the counter is incremented by $h$. Likewise,
 if the simulated transition of $A$ decrements the counter by $h$, then on reading a sequence of $j$ 
 output symbols, $R$ produces an output of length $j-h$. 
The run of $R$ is successful if an only if the simulated run of $A$ is so. 
The fact that $A$ accepts every word with a counter that never 
exceeds $k$, immediately implies that $T_1 \ocont R(T_2)$.  

Conversely, suppose that $T_1 \ocont R(T_2)$ for some
rational resynchronizer $R$. By Theorem \ref{thm:bounded-delay}, 
we can assume without loss of generality that $R$ is a $k$-delay resynchronizer,
for some $k$ (that can be even computed from $T_1$, $T_2$, and $R$, 
but this is immaterial here). 
From this it is easy to see that $A$ is language-bounded, and precisely, 
that $A$ accepts every word with a counter that never exceeds $k$,
as when reading a run $\rho$ of $T_1$, it can guess 
a run $\rho'$ of $T_2$ such that $R(\rho') = \rho$.

\paragraph*{From language-boundedness to synthesis.}
Let $A$ be an OCA. We construct two real-time, one-way transducers 
$T_1,T_2$ that have the same input alphabet as $A$, say $\S$, and 
a singleton output alphabet, say $\G=\{c\}$.  
The transducer $T_1$ reads any word $a_1 \dots a_n \in \S^*$ 
and outputs one letter $c$ for each consumed input symbol. 
In particular, the synchronization language of $T_1$ is 
$\{a_1 c \dots a_n c : a_i \in \S, n \ge 0\}$. 
Note that $T_1$ is real-time and functional.
The transducer $T_2$ does the following: 
upon reading  $a_1 \dots a_n$, it guesses 
a successful run of the OCA $A$. 
Whenever the counter is incremented along the guessed run of $A$,
$T_2$ outputs $cc$; whenever the counter is decremented, $T_2$ outputs
$\emptystr$; whenever the counter is unchanged, $T_2$ outputs $c$.
Note that $T_2$ is also real-time, but not necessarily functional.

Let us now prove that $A$ is language-bounded if and only
if $T_1 \ocont R(T_2)$ for some rational resynchronizer $R$.

Suppose first that $A$ is language-bounded, with bound $k$. 
We obtain from this a  $k$-delay resynchronizer
$R$ that reads a synchronized word 
$a_1 c^{i_1} \dots a_n c^{i_n}$ of $T_2$, where $i_j\in\set{0,1,2}$
for all $j$.
The resynchronizer $R$ simulates a counter taking values in $[-k,k]$, 
and outputs $a_1 c \dots a_n c$, accepting if and only if 
the counter is $0$. Each time an $a_ic^2a_j$ is encountered, it corresponds 
to an increment in the OCA; then $R$ outputs $a_ic$, and 
the simulated counter decreases by 1 in $R$; likewise, 
each time an $a_ica_j$ is encountered, $R$ outputs $a_ic$ with 
no change in the simulated counter value, and finally, 
when two consecutive input symbols $a_ia_j$ are read by $R$, 
$R$ outputs $a_ic$ and 
 the simulated counter value increases by 1. 
 Since the counter value is bounded by $k$ in the OCA, the simulated counter 
 in $R$ is within $[-k,k]$. 
Clearly, $T_1 \ocont R(T_2)$. 

Conversely, suppose that $T_1 \ocont R(T_2)$ for some rational resynchronizer $R$.
We argue as before, using Theorem \ref{thm:bounded-delay}:
we assume without loss of generality that $R$ is a $k$-delay resynchronizer,
for some $k$, and derive from this that $A$ is language-bounded.
\end{proof}




\undecidableOCABoundedness*

\begin{proof}
 The reduction is from the
  boundedness problem for multi-counter (Minksy) machines. Such a
  machine $M$ can increment, decrement and test for zero. The question is
  whether there exists some bound $k$ such that all computations of
  $M$ (not necessarily accepting) from the initial configuration with
  all counters zero, have all counters stay below $k$. One can assume
  w.l.o.g.~that if $M$ is not bounded then for every $k$ there is some
  initial run of $M$ where \emph{all} counters exceed $k$.

  The OCA $A$ reads sequences of transitions of $M$. At the beginning, 
  $A$ guesses a counter index $j$ of $M$ and starts simulating the
  sequence of transitions on counter $j$. If the sequence of
  transitions is incorrect because of counter $j$, the OCA accepts and
  stops after emptying counter $j$. Note that there are two types of
  error: either the counter is zero but should be decremented, or the
  counter is tested for zero, but is not zero. Both kinds of error can
  be checked by the OCA. Otherwise, if the simulation goes through for
  counter $j$, then the OCA accepts with empty counter at the end.

  Assume that $M$ is bounded, with bound $k$. If a sequence $\r$ of
  transitions is a run of $M$, then all simulations on any counter
  will be bounded by $k$. If $\r$ is not a run, then there is a first
  position of $\r$ where an error occurs, for instance because of
  counter $j$. Then the run of $A$ simulating counter $j$ will accept
  $\r$ within bound $k$.

  If $M$ is unbounded then for every $k$ there is a run $\r$ where \emph{all}
  counters exceed $k$. In this case all runs of $A$ on $\r$ exceed
  $k$, so $A$ is not language-bounded.
\end{proof}



\twoWayEquivalenceUp*

\begin{proof}
The implications from 2.~to~1. and from 3.~to 2.~are as in the proof
of Theorem \ref{th:one-way-equivalence-up}. The only interesting implication
is from 1.~to~3, where we suppose that $T_1 \subseteq T_2$ and we aim at 
constructing a $1$-bounded 
Parikh resynchronizer $R$ such that $T_1 \oeq R(T_2)$,
and with the same target set as $T_1$. 
The proof exploits some constructions based on crossing sequences, 
which are classically used to translate two-way 
automata to equivalent one-way automata \cite{Shepherdson59}, 
as well as to reduce containment of functional two-way 
transducers to emptiness of languages recognized by Parikh 
automata \cite{STACS19}. We briefly recall the key notions here, by adapting
them in a way that is convenient for the presentation (notably, considering transitions instead of states).

A \emph{crossing sequence} of a two-way automaton or a functional two-way transducer is a tuple 
$\bar t=(t_1,\dots,t_n)$ of transitions such that the source states of $t_1,t_3,\dots$ are 
right-reading and the source states of $t_2,t_4,\dots$ are left-reading.
The tuple is meant to describe the transitions along a successful run that 
depart from configurations at a certain position $y$. 
Formally, given a run $\rho$, the crossing sequence of $\rho$ at input position $y$, denoted $\rho[y]$,
consists of the quadruples $(q,a,v,q')$ such that $(q,y) \act{a \:|\: v} (q',y')$ is a transition
of $\rho$, where the occurrence order on transitions induces a corresponding order on the
quadruples of the crossing sequence.
Without loss of generality, for two-way automata, as well as for functional two-way transducers, 
one can restrict to successful runs that never visit the same state twice at the same position. 
Accordingly, we can assume that the length of a crossing sequence never exceeds the total number 
of states of the device. 
Moreover, when the two-way automaton or transducer is unambiguous, the crossing sequences are 
uniquely determined by the input and the specific position in it. More precisely, there are
regular languages $L_{\bar t}$, one for each possible crossing sequence, that contains precisely those
inputs $u$ with a specific position $y$ marked on it (for short, we denote such words by $\angled{u,y}$),
such that the crossing sequence at $y$ of the unique successful run on $u$ is precisely $\bar t$.

\smallskip
We now turn to the main proof, which is divided into several steps.

\paragraph*{Encoding output positions.}
We begin by describing a natural encoding of arbitrary output positions by means of their origins.
Of course, the encoding depends on the given input, denoted $u$, and on the transducer we consider, 
either $T_1$ or $T_2$, which here is generically denoted by $T$.
Now, let $\rho$ be the unique successful run of $T$ on $u$, and 
let $\synch$ be the induced origin graph.
To simplify the notations, hereafter we tacitly assume that $T$
produces at most one letter at each transition --- the 
assumption is without loss of generality, since long outputs 
originating at the same input position can be produced incrementally
by exploiting two-way head motions. 
Let $n$ be the number of states of $T$. Since $T$ is unambiguous, $\synch$ contains at most $n$ 
output positions with the same origin (otherwise, the same configuration would be visited at least 
twice along the successful run $\rho$, which could then be used to contradict the assumption of 
unambiguity). 
This means that every position $x$ in $\Out(\synch)$ can be encoded by its origin 
$y_x=\Orig(\synch)(x)$ together with a suitable index $i_x\in\{1,\dots,n\}$, 
describing the number of output positions $x' \le x$ with the same origin $y_x$ as $x$.
Moreover, we recall that $y_x$ can be represented as an annotated input of the form
$\angled{u,y_x}$.

\paragraph*{Decoding by Parikh automata.}
We now show that there are Parikh automata that compute the inverse of 
the encoding $x \mapsto (y_x,i_x)$ described above. More precisely, there are unambiguous 
Parikh automata $A_1,\dots,A_n$ such that each $A_i$ receives as input a word $\angled{u,y}$
having a special position marked on it, and outputs the unique output position $x$ such that 
$(y,i)=(y_x,i_x)$, if this exists, otherwise the output is undefined.
Each automaton $A_i$ can be constructed from $T$ and $i$ by unambiguously guessing 
the crossing sequences of the unique run of $T$ on $u$, and by counting the number of 
output symbols emitted until a \emph{productive} transition at the marked position 
$y$ is executed for the $i$-th time --- a productive transition is a transition 
that produces non-empty output.

\paragraph*{Redirecting origins.}
We now apply the constructions outlined above in order to obtain the desired Parikh resynchronizer $R$
from $T_1$ and $T_2$.
Let $u$ be some input and $\synch_1,\synch_2$ be the origin graphs induced by the 
unique successful runs of $T_1,T_2$ on $u$.
Since $T_1\subseteq T_2$, we can further let 
$v=\Out(\synch_1)=\Out(\synch_2)$.
Consider any output position $x\in\dom(v)$. 
According to $\synch_2$, $x$ is encoded by an input position $y_x$ and an index $i_x\in\{1,\dots,n_2\}$,
where $n_2$ is the number of states of $T_2$. In a similar way, according to $\synch_1$, the same position 
$x$ is encoded by some input position $z_x$ and an index $j_x\in\{1,\dots,n_1\}$, where $n_1$ is the number
of states of $T_1$. 
Moreover, based on the previous constructions, there are unambiguous Parikh automata 
$A_{2,i}$ and $A_{1,j}$ such that
\begin{itemize}
	\item $A_{2,i}(\angled{u,y})=x$ if and only $(y,i) = (y_x,i_x)$,
	\item $A_{1,j}(\angled{u,z})=x$ if and only $(z,j) = (z_x,j_x)$.
\end{itemize}
Since unambiguous Parikh automata are closed under pointwise difference, there is a 
unambiguous Parikh automaton $A_{i,j}$ that recognizes precisely the language of 
annotated words $\angled{u,y,z}$ such that
\begin{align*}
	A_{2,i}(\angled{u,y}) - A_{1,j}(\angled{u,z}) = 0
	\tag{$\star$}
\end{align*}
Note that the above language defines 
a partial bijection between pairs of positions $y,z$ in the input $u$ in 
such a way that $y$ and $z$ are the origins of the same output position $x$ 
according to the unique origin graphs $\synch_1,\synch_2$ of $T_1,T_2$ 
such that $\In(\synch_1)=\In(\synch_2)=u$.
This property can be used to define the component $\move_\otype$ 
of the desired resynchronizer $R$, by simply letting
\[
  \move_\otype = \{ (u,y,z) ~\mid~ A_{i,j}(\angled{u,y,z}) = 0 \}
\]
where $\otype=(a,i,j) \in \G\times\{1,\dots,n_2\}\times\{1,\dots,n_1\}$.

For the correctness of the above definition we rely on guessing the 
correct pairs of indices $(i,j)$ as annotations of output positions.
More precisely, we have that:
\begin{itemize}
  \item for every output position $x$ with source origin 
        $y=\Orig(\synch_2)(x)$ and with label $\otype=(a,i_x,j)$,
        there is at most one input position $z$ such that $(u,y,z)\in\move_\otype$; 
        in addition, if we also have $j=j_x$, then $z=\Orig(\synch_1)(x)$
        is the target origin of $x$; symmetrically,
  \item for every output position $x$ with target origin 
        $z=\Orig(\synch_1)(x)$ and with label $\otype=(a,i,j_x)$,
        there is at most one input position $y$ such that $(u,y,z)\in\move_\otype$; 
        in addition, if we also have $i=i_x$, then $y=\Orig(\synch_2)(x)$
        is the source origin of $x$.
\end{itemize}
Based on the above properties, we need to guess suitable output parameters 
that associate with each position $x$, a correct pair $(i_x,j_x)$. 
We explain below how this is done using the components $\opar$ and 
$\nxt_{\otype,\otype'}$ of the resynchronizer. 

\paragraph*{Constraining output parameters.}
We first focus on the indices $j_x$ related to $T_1$; 
we will later explain how to adapt the constructions to check the indices $i_x$ 
related to $T_2$. As usual, we fix an input $u$ and the unique successful run $\rho_1$
of $T_1$ on $u$.
The idea is that each index $j_x$ corresponds to a certain element of the crossing 
sequence of $\rho_1$ at the target origin $z_x$, and knowing the correct index for 
$x$ determines the correct index for the next output position $x+1$. 
Based on this, correctness can be verified inductively using the guessed 
crossing sequences and the relation $\nxt_{\otype,\otype'}$ of the resynchronizer, 
as follows.
For the base case, we check that the first output position is correctly 
annotated with the index $j=1$: this is readily done by a regular language $\opar$.

For the inductive step, we consider an output position $x$ and assume 
that it is correctly annotated with $j=j_x$.
Let $j'$ be the annotation of the next position $x+1$. 
To check that $j'$ is also correct, we consider pairs of productive transitions
in the crossing sequences associated with the target origins of $x$ and $x+1$, 
and verify that they are connected by a non-productive run.
More precisely, let $z$ and $z'$ be the target origins of $x$ and $x+1$,
respectively, and let $\bar t_z$ and $\bar t_{z'}$ be the crossing sequences 
of $\rho_1$ at those positions.
We have that $j'=j_{x+1}$ if and only if the $j$-th productive transition 
of $\bar t_z$ and the $j'$-th productive transition of $\bar t_{z'}$ 
are connected by a factor of the run that consists only of non-productive transitions.
The latter property can be translated to a regular property $\nxt_{\otype,\otype'}$
concerning the input annotated with two specific positions, $z$ and $z'$, 
assuming that $\otype=(a,i,j)$ and $\otype=(a',i',j')$ are the letters of 
the output positions $x$ and $x+1$.

It now remains to check the correctness of the output annotations
w.r.t.~the indices $i$ for the second transducer $T_2$. We follow
a principle similar to the one described above for $T_1$. The only
difference is that now, in the inductive step, we have work with the 
source origins $y$ and $y'$ of consecutive output positions $x$ and $x+1$. 
The additional difficulty is that, by definition, the relation $\nxt_{\otype,\otype'}$ 
can only refer to target origins. We overcome this problem by exploiting the partial
bijection between target and source origins, as defined by the relations 
$\move_\otype$ and $\move_{\otype'}$. Formally, we first define a relation 
$\nxt_{\otype,\otype'}^{\text{source}}$ as before, that constrain the indices 
$i$ and $i'$ associated with two consecutive output positions $x$ and $x+1$ 
labeled by $\otype$ and $\otype'$, respectively. We do this as if 
$\nxt_{\otype,\otype'}^{\text{source}}$ were able to speak about source origins. 
We then intersect the following relation with the previously 
defined relation $\nxt_{\otype,\otype'}$: 
\[
  \big\{ (u,z,z') ~\mid~ 
     \exists y,y'~ (u,y,y')\in\nxt_{\otype,\otype'}^{\text{source}}, ~
                   (u,y,z)\in\move_\otype, ~
                   (u,y',z')\in\move_{\otype'} \big\}.
\]
Since in the inductive step we assume that $x$ is correctly annotated with the
pair $(i,j)$ and $x+1$ is annotated with $(i',j')$, where $j'=j_x$ is correct 
by the previous arguments, there are unique $y,y'$ that satisfy 
$(u,y,z)\in\move_\otype$ and $(u,y',z')\in\move_\otype$ 
in the above definition, and these must be the source origins of $x$ and $x+1$.
This means that the above relation, which is definable by a unambiguous Parikh 
automaton, correctly verifies the correctness of the index $i'$ associated with $x+1$.

We conclude by observing a few properties of the defined Parikh resynchronizer $R$.
As already explained, the relation $\move_\otype$ defines a bijection between
pairs of input positions, so $R$ is a $1$-bounded 
Parikh resynchronizer.
As concerns its target set, that is the set of pairs $(u,z)$ 
such that $(u,y,z)\in\move_\otype$ for some $z\in\dom(u)$ and some 
$\otype\in\G\times\{1,\dots,n_2\}\times\{1,\dots,n_1\}$,
it coincides by construction with the target set of $T_1$.
Finally, since the relation $\nxt_{\otype,\otype'}$ is defined
by conjoining a regular property with the properties defined by
the relations $\move_\otype$ and $\move_{\otype'}$, we have that
$\nxt_{\otype,\otype'}$ is regular if $\move_\otype$ and $\move_{\otype'}$
are regular.
\end{proof}



\twoWayEquivalenceUpBis*

\begin{proof}
We prove the following implications in the order: 
1.~$\rightarrow$ 2.~$\rightarrow$ 3.~$\rightarrow$ 4.~$\rightarrow$ 1.

\paragraph*{From 1.~to~2.}
This is trivial since $\hat R$ is bounded 
and satisfies $T_1 \oeq \hat R(T_2)$, and hence $T_1\ocont \hat R(T_2)$.

\paragraph*{From 2.~to~3.}
Let $R$ be a $k$-bounded regular resynchronizer.
The goal is to construct an equivalent $1$-bounded regular resynchronizer $R'$
(note that this part of the proof does not depend on $T_1$ and $T_2$).
For this, we introduce a parameter $i_x\in\{1,\dots,k\}$ 
associated with each output position $x$, and require that 
for all output positions $x,x'$ having the same label $\otype$, 
and for all input positions $y,z$ 
such that $(u,y,z)\in \move_\otype$, if $i_x=i_{x'}$, then $x=x'$.
The existence of such a mapping $x\mapsto i_x$ follows easily 
from the assumption that $R$ is $k$-bounded. 
The relation $\move'_{(\otype,i)}$ of the new resynchronizer $R'$
redirects origins of output positions based on their annotations 
$(\otype,i)\in\G\times\{1,\dots,k\}$, as follows:
\[
\move'_{(\otype,i)} = 
\big\{ (w,y,z) ~\mid~
(w,y,z)\in \move_\otype, ~ 
\exists^{!i} y'~~ y'<y \wedge (w,y',z)\in \move_\otype \big\}
\]
where $\exists^{!i} y'$ is an abbreviation for ``there exist exactly $i$ positions $y'$ such that\dots''.
As for the relation $\nxt'_{(\otype,i),(\otype',j)}$, this coincides with
$\nxt_{\otype,\otype'}$, so it does not take into account the new annotations.
Thus, the defined resynchronizer $R'$ is $1$-bounded, regular, and defines the 
same resynchronization as $R$.

\paragraph*{From 3.~to~4.}
Suppose that $R$ is a $1$-bounded regular resynchronizer with input
alphabet $\S$ and output alphabet $\G$, such that $T_1 \ocont R(T_2)$. 
The goal is to construct a $1$-bounded regular resynchronizer $R'$
with the same target set as $T_1$ and such that $T_1 \oeq R'(T_2)$.
For the sake of simplicity, we assume 
that $R$ has no input parameters, and similarly
$T_1$ has no common guess (the more general cases can be dealt with
by annotating the considered inputs with the possible parameters and
the common guess).
The idea for defining the desired resynchronizer $R'$ is as follows. 
We first restrict each relation $\move_\otype$ so as to make it a 
partial bijection, that is, for every input $u$, and every source origin $y\in\dom(u)$, 
there is an annotation $w$ of the input and at most one target origin $z$ 
that corresponds to $y$ in $u\otimes w$ (and conversely, since $R$ is $1$-bounded, 
for every target origin $z$ there is a unique source origin $y$ that 
corresponds to $z$). This step requires the use of appropriate input
parameters that determine a unique target origin $z$ from any given 
source origin $y$. Then, we restrict further the relation $\move_{\otype}$
so that every target origin $z$ is witnessed by $T_1$.
Formally, we introduce input parameters ranging over $\bbB^\G$
and work with annotated inputs of the form $u\otimes w$,
with $u\in\S^*$ and $w\in(\bbB^\G)^*$.
Given $u\in\S^*$, we define $O_u$ as the set of all positions $z=\Orig(\synch)(x)$
where $\synch$ is an origin graph of $T_1$, $x\in\dom(\Out(\synch))$,
and $\In(\synch)=u$.
The new relation $\move'_{\otype}$ that redirects source origins 
to target origins is defined as the following restriction of 
$\move_{\otype}$:
\[
\move'_{\otype} = 
\big\{ (u\otimes w,y,z) ~\mid~
(u,y,z)\in \move_\otype, ~ w(z)(\otype)=1, ~ z\in O_u ~
\big\}.
\]
Clearly, the above relation is regular and contained in $\move_{\otype}$. 
However, it is still possible that $\move'_{\otype}$ associates multiple target
origins with the same source origin.

To get a partial bijection from $\move'_{\otype}$ we need
to constrain the possible annotated input $u\otimes w$. 
We do so by requiring that, for every output letter $\otype\in\G$ 
and every position $y$ in $u\otimes w$, 
if there is $z$ satisfying $(u,y,z)\in \move_\otype$,
then there is {\sl exactly one} $z'$ satisfying $(u,y,z')\in \move_\otype$
and $w(\otype)(z)=1$.
Note that the latter property is again regular, and thus could be 
conjoined with the original relation $\ipar$ to form the new
relation $\ipar'$.
Accordingly, the relation $\nxt'_{\otype,\otype'}$ of the 
desired resynchronizer $R'$ defines the same language as 
$\nxt_{\otype,\otype'}$, but expanded with arbitrary
input annotations over $\bbB^\G$.

It is now easy to see that the the resulting resynchronizer $R'$ is $1$-bounded,
and in fact, on each input, defines a partial bijection between source and target
origins in such a way that the target set coincides with that of $T_1$.
By pairing this with the containments $R'(T_2)\ocont R(T_2)$ and $T_1 \ocont R(T_2)$, 
we obtain $T_1 \oeq R'(T_2)$.

\paragraph*{From 4.~to~1.}
Knowing that $\hat R(T_2) \oeq T_1 \oeq R(T_2)$ for two $1$-bounded resynchronizers
$R,\hat R$ with the same target sets as $T_1$ implies that the relations 
$\move_\otype$ and $\move'_\otype$, from $R$ and $\hat R$ respectively, coincide. 
Moreover, since the relation $\move_\otype$ of $R$ is assumed regular, this means 
that $\move'_\otype$ is regular too. Finally, we recall that  $\hat R$
is such that $\nxt'_{\otype,\otype'}$ is regular whenever $\move_\otype$ and $\move_{\otype'}$
are. We can then conclude that the relations $\nxt'_{\otype,\otype'}$ from $\hat R$
are also regular, and hence $\hat R$ is a regular resynchronizer.
\end{proof}





\fi
\fi

\end{document}
